\documentclass[11pt]{article}

 \usepackage{amsmath,   amssymb,amsfonts,xspace,graphicx,relsize,bm,mathtools,xcolor}
    \usepackage{mathrsfs}
    \usepackage[pagebackref]{hyperref}
    \usepackage{enumerate,enumitem}
    \usepackage{bm}
 \hypersetup{
	colorlinks,
	linkcolor={blue!100!black},
	citecolor={red!100!black},
}
\usepackage[margin=1in]{geometry}
\usepackage{amsmath,amssymb,amsthm}
\usepackage{graphicx}
\usepackage{comment}
\newcommand{\ra}{\rangle}
\newcommand{\la}{\langle}
\usepackage[capitalise]{cleveref}
\def\01{\F}
\usepackage{tcolorbox}
\newtheorem{theorem}{Theorem}[section]

\newtheorem{claim}[theorem]{Claim}
\newtheorem{definition}[theorem]{Definition}
\newtheorem{fact}[theorem]{Fact}
\newtheorem{corollary}[theorem]{Corollary}
\newtheorem{proposition}[theorem]{Proposition}

\newtheorem{lemma}[theorem]{Lemma}

\newcommand{\E}{\mathop{\mathbb{E}}}
\newcommand{\FF}{\mathbb{F}}
\usepackage[linesnumbered,ruled,vlined]{algorithm2e}
\usepackage{circledsteps}
\newcommand{\calF}{{\cal F }}

\makeatletter
\def\widebreve{\mathpalette\wide@breve}
\def\wide@breve#1#2{\sbox\z@{$#1#2$}%
     \mathop{\vbox{\m@th\ialign{##\crcr
\kern0.08em\brevefill#1{0.8\wd\z@}\crcr\noalign{\nointerlineskip}%
                    $\hss#1#2\hss$\crcr}}}\nolimits}
\def\brevefill#1#2{$\m@th\sbox\tw@{$#1($}%
  \hss\resizebox{#2}{\wd\tw@}{\rotatebox[origin=c]{90}{\upshape(}}\hss$}

\usepackage{url}
\def\01{\F}

\newcommand{\eps}{\varepsilon}
\newcommand{\ketbra}[2]{|#1\rangle\langle#2|}

\newcommand{\braketbra}[3]{\langle #1|#2| #3 \rangle}

\newcommand{\diag}{\mbox{\rm diag}}

\newcommand{\Exp}{{\mathbb{E}}}

\newcommand{\Sh}{\ensuremath{\textsf{Stab}}}

\newcommand{\T}{\ensuremath{\mathcal{T}}}

\newcommand{\id}{\ensuremath{\mathbb{I}}}

\usepackage{enumitem}

\newcommand{\Z}{\mathsf{Z}}
\newcommand{\X}{\mathsf{X}}
\usepackage{braket}

\newcommand{\Fe}{\ensuremath{\mathcal{F}}}

\newcommand{\Tr}{\mathsf{Tr}}

\newcommand{\F}{\mathbb{F}}
\newcommand{\bk}[2]{\langle #1|#2\rangle}
\usepackage{setspace}
\SetKwInput{KwInput}{Input}
\SetKwInput{KwOutput}{Output}
\SetKw{Reject}{reject}
\SetKw{Accept}{accept}
\usepackage{thm-restate,mathrsfs} 

\newcommand{\Graph}{\operatorname{Graph}}

\newcommand{\TV}{d_{\text{TV}}}

\newcommand{\Oc}{\mathcal{O}}
\newcommand{\SWAP}{\operatorname{SWAP}}
\DeclareMathOperator{\rank}{rank}

\usepackage{graphicx}

\newcommand{\xll}{\vec{x}_{<i}}

\newcommand{\Hc}{\mathcal{H}}
\newcommand{\Mc}{\mathcal{M}}

\usepackage{subcaption}

\usepackage{quantikz}

\begin{document}

\title{Optimal stabilizer testing and learning\\ with limited quantum memory}
\author{
Srinivasan Arunachalam\\[2mm]
\small IBM Research\\\small Silicon Valley\\
\small \texttt{Srinivasan.Arunachalam@ibm.com}
\and
Louis Schatzki\thanks{Current affiliation: IBM Research}\\[2mm]
\small  Dahlem Center for Complex Quantum Systems\\\small  Freie Universität Berlin\\
\small \texttt{l.schatzki@ibm.com}
}

\maketitle
\begin{abstract}
We study stabilizer state testing and learning with limited coherent quantum memory. Here an algorithm sequentially receives copies of an unknown $n$-qubit state, but may keep only $k$ qubits of coherent quantum memory between measurements. With unrestricted memory, seminal work of Gross, Nezami and Walter~\cite{gross2021schur} showed how to test $n$-qubit stabilizer states  using $6$ copies, which is \emph{dimension independent}, unlike the learning complexity of $\Theta(n)$. We show that this testing-vs-learning separation is lost under memory constraints. More concretely we show~that \vspace{1mm}

 \begin{enumerate}
       \item The sample complexity of testing stabilizer states in the $k$-qubit memory framework is $\Theta(n-k)$. Our upper bound goes via a novel connection to the hidden shift problem and the lower bound is proven using a novel approach to average case bounds on likelihood ratios via combinatorics of the stochastic orthogonal group.
    
       \item The sample complexity of learning stabilizer states with $k$ qubits of memory, in the non-adaptive framework,  is $\Theta(n^2/k)$.\vspace{1mm} 
   \end{enumerate}

\noindent As a further application of our techniques, we prove an exponential lower bound for purity testing even when the memory may be left \emph{coherent} throughout the protocol. Our main results identify coherent quantum memory as the resource enabling the usual separation between stabilizer testing and learning. In particular, even with $k=0.99n$ qubits of memory, there is no constant-copy stabilizer tester; furthermore for $k=cn$ qubits of memory (for $0< c < 1$), stabilizer testing is as hard as learning, with both requiring $\Theta(n)$ copies.
\end{abstract}

\newpage

\setcounter{tocdepth}{2}
{\footnotesize \tableofcontents}

\newpage 
\section{Introduction}

Learning and testing properties of unknown $n$-qubit quantum systems are fundamental problems in quantum information science. In quantum state tomography, the goal is to learn a \emph{classical description} of an unknown state given several copies of that state. In quantum property testing, the goal is instead to decide whether the unknown state satisfies a certain \emph{property}, again given several copies of the state. While a full tomography would suffice for property testing as well, ideally testing requires far fewer samples. As we approach the era of small-scale fault-tolerant quantum computers, we are entering the regime where such learning and testing algorithms may be implemented. However, even on these devices, a major bottleneck is the difficulty of realizing coherent quantum memory. Ideally we would like to restrict the algorithms  to store only $k$ of the $n$ qubits coherently, while measuring the remaining qubits in arbitrary bases, where $k$  is as small as possible. This captures a natural memory-limited streaming model for quantum data, where the algorithm sees many qubits but is only allowed to preserve a small coherent workspace~between~measurements.

In this spirit, a flurry of works over the past few years have studied quantum algorithms under streaming and memory constraints~\cite{chen2022exponential,chen2024optimal,chen2025efficient,huang2022quantum,arunachalam2025generalized,gong2024sample,hinsche2025single,aharonov2022quantum,bubeck2020entanglement,chen2022tight,liu2023memory,beckey2025product,arunachalam2023role,oufkir2023sample,caro2024learning,lowe2025lower,king2024triply,liu2024role,chen2024optimaltradeoff,chen2024tight,chen2022quantum,hinsche2026clifford}. One of the main messages from this bulk of work is that restrictions to single-copy measurements ($k=0$) often lead to large, potentially exponential, separations in sample complexity, suggesting that memory is not merely an implementation detail which can be removed for free, but rather a computational resource that can dramatically change the complexity of the problem. Although the tasks considered in these papers are very natural, the field still lacks a good understanding of the learning and testing complexity of the ``simplest'' family of  states in a memory-limited~setting, namely  $n$-qubit \emph{stabilizer states}.

In this regard, a seminal result of Gross, Nezami, and Walter~\cite{gross2021schur} showed that stabilizer states can be tested using only $6$ copies. This result was striking since it gave the first constant-copy tester for a natural class of states, and it showed that testing stabilizer states is vastly easier than learning them, for which the sample complexity is known to be $\Theta(n)$ by a result of Aaronson and Gottesman~\cite{aaronsontalk}, later reproven using the Bell sampling framework by Montanaro~\cite{montanaro2017learning}. Both these learning and testing results crucially rely on a subroutine called \emph{Bell sampling}, which takes two copies of the state and, for every $i\in [n]$, measures qubits $(i,n+i)$ in the Bell basis. Since then, Bell sampling has become the basis of several algorithms in the field~\cite{grewal2022low,grewal2024agnostic,chen2024optimal,chen2024stabilizer,ad2024tolerant,haug2024efficient,hangleiter2024bell,leone2024learning,leone2022stabilizer,leone2024learningstabdoped} for a variety of tasks. Although Bell sampling is conceptually simple and extremely powerful, its standard implementation requires keeping one entire copy of the state coherent ($k=n$) while the other is processed. 
Of course, it is desirable to store as few coherent qubits as possible, motivating the central questions of our~work:
\begin{quote}
\begin{center}
\textbf{\emph{1.}} Can we obtain constant-copy testers with limited memory?\\[0.5mm]
    \textbf{\emph{2.}} Is stabilizer testing always easier than learning even with limited memory?
\end{center}
\end{quote}
To tackle these questions, we consider a smooth interpolation between single-copy and two-copy measurements. That is, in each round a fresh copy of an input state $\ket{\psi}$ is provided to the learning/testing algorithm. On top of this input, the algorithm is allowed to maintain $k$ qubits of coherent memory, interpolating between single-copy measurements ($k=0$) and two-copy measurements ($k=n$). See Figure~\ref{fig:protocol_with_mem} for an illustration.  As far as we are aware, this model has been studied substantially less than separations between single and multi-copy measurements.

\begin{figure}[!ht]
    \centering
    \includegraphics[width=0.7\linewidth]{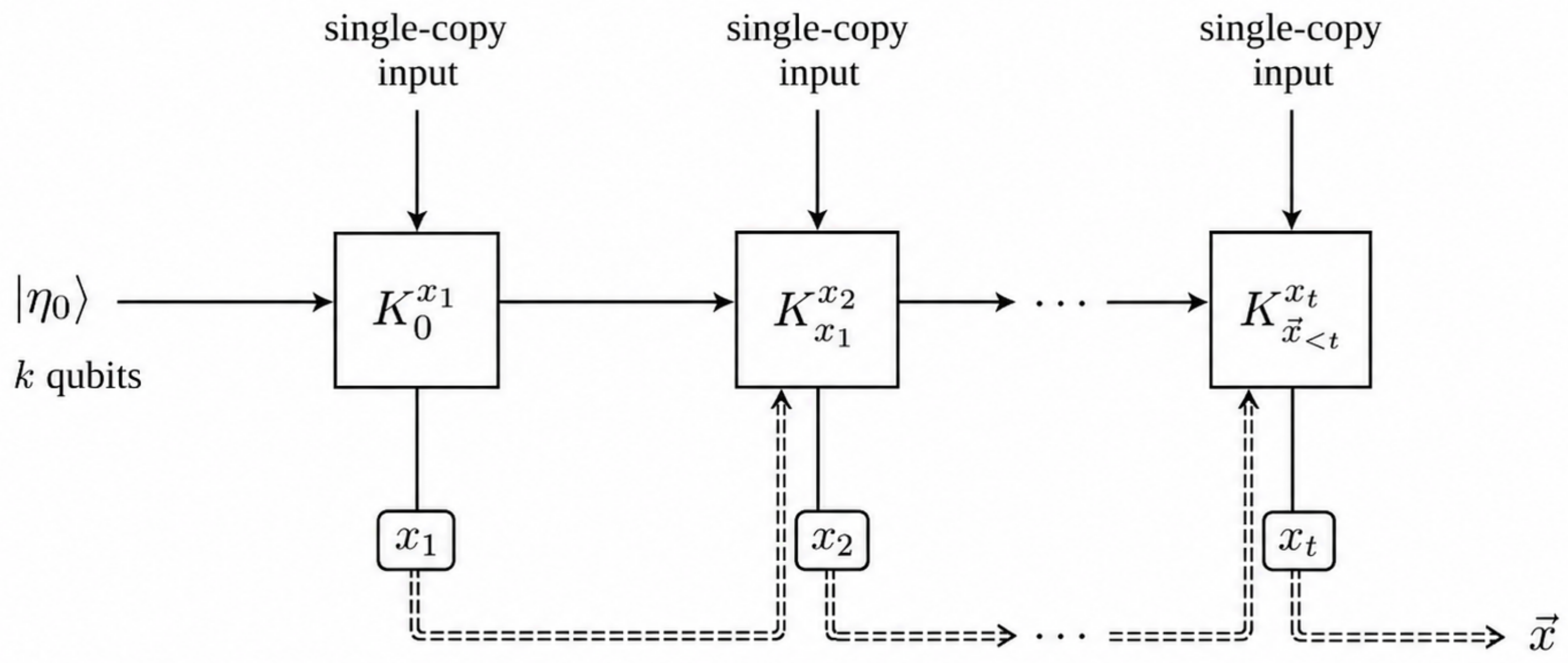}
    \caption{Protocol using $k$-qubits of memory and single-copy measurements. In each round a new copy of an input state $\rho$ is loaded. There is also a coherent memory register of $k$-qubits. The operations in each round act on the memory and the input at that step. The channels applied may depend classically upon the prior observations.}
    \label{fig:protocol_with_mem}
\end{figure}

\subsection{Main results}
Our main result is to settle both of these questions. Even with $0.99n$ qubits of coherent memory, $\Theta(n)$ copies are necessary and sufficient to test stabilizer states. Interestingly, with $0.99n$ qubits of memory, $\Theta(n)$ copies are also necessary and sufficient to \emph{learn} an unknown stabilizer state. Thus, the $6$-copy tester of~\cite{gross2021schur} crucially needs the additional $n$ qubits of coherent memory. More generally, we prove that with $k$ qubits of memory, the sample complexity of testing and learning stabilizer states is $\Theta(n-k)$ and $O(n^2/k)$ respectively, the latter being tight for non-adaptive learners. This exhibits a curious difference in how the sample complexity scales with limited memory for the same class of states. Take $k=0$, i.e., single-copy measurements. Then, testing requires $\Theta(n)$ samples but the best known learning algorithm requires $O(n^2)$ samples (with a matching lower bound for non-adaptive algorithms). As $k$ grows towards $n$, the sample complexity of learning decreases at a much faster rate than that of testing and, for $k = c\cdot n $ with $c<1$ a constant, both become $\Theta(n)$. See Figure~\ref{fig:learning_vs_testing} for an illustration.

We now formally present our results and compare them with prior work. For notational convenience, throughout this paper we let $\Sh$ be the class of $n$-qubit stabilizer states. For a state~$\ket{\psi}$,~let
$$
\mathcal{F}_{\cal \Sh}(\ket{\psi}):=\max_{\ket{\phi}\in \Sh}|\langle \psi|\phi\rangle|^2.
$$
be the fidelity of $\ket{\psi}$ with respect to $\Sh$ (i.e., quantifies how close $\ket{\psi}$ is to a stabilizer state).

\begin{restatable}{theorem}{testingupper}{(Optimal testing bounds)}
\label{thm:stabtestingupper}
Let $k\geq 0, \varepsilon>0$. There is an adaptive protocol that uses $k$ qubits of memory and $O((n-k)/\varepsilon)$ copies of an unknown $\ket{\psi}$ to distinguish between $\mathcal{F}_{\cal \Sh}(\ket{\psi})=1$ vs.~$\mathcal{F}_{\cal \Sh}(\ket{\psi})\leq 1-\varepsilon$. Also every such tester needs $\Omega(n-k)$ copies.
\end{restatable}
\textbf{Prior work.} There were two known results in this area. The first case is when $k=n$ (i.e., $2$ copy measurements) where there was a constant copy tester, hence the bound above is tight in this case. The second  case is when $k=0$ (i.e., single-copy protocols) where a striking work of Hinsche and Helsen~\cite{hinsche2025single} showed that $O(n)$ copies are sufficient for testing, matching our upper bound and $\Omega(\sqrt{n})$ copies are necessary. It seemed conceivable for a while that there was some sort of a birthday paradox argument, which could presumably result in an $O(\sqrt{n})$ upper bound for the single-copy case. Our main result first rules this out (answering a question of~\cite{hinsche2025single}) and cleanly interpolates between $k=0$ and $k=n$. Further, our protocol scales with $1/\varepsilon$, improving the $1/\varepsilon^2$ scaling from computational difference sampling~\cite{hinsche2025single} in the case of $k=0$. As far as we know, the tight bound of $\Theta(n-k)$ that we prove for testing is the first fine-grained (instead of the exponential bounds often seen) separation proven for property testing in the $k$-qubits of memory model.

\begin{figure}
    \centering
    \begin{subfigure}[t]{0.45\textwidth}
        \includegraphics[width=\textwidth]{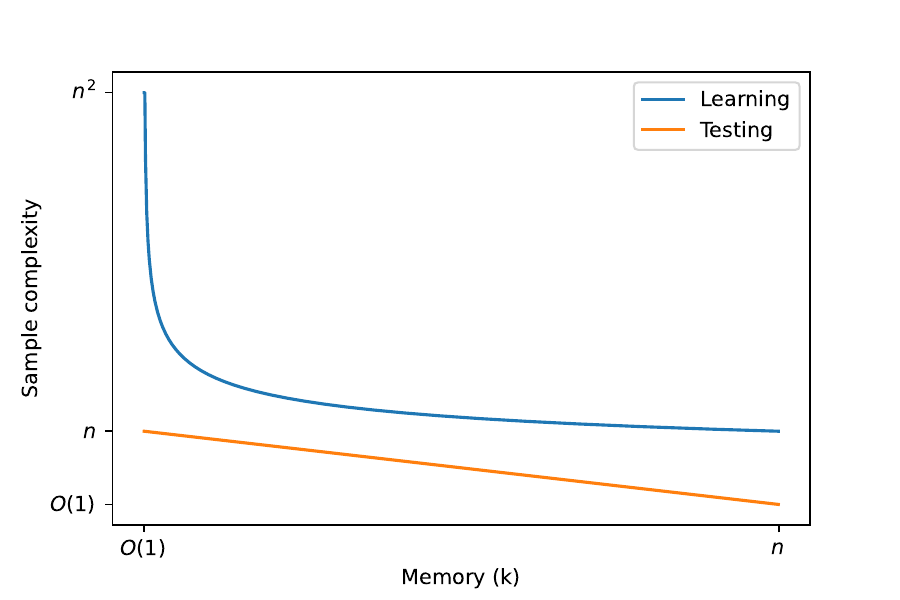}
    \caption{With memory, the sample complexity for learning decreases at a  faster rate than testing.}
    \end{subfigure}
    \hfill
    \begin{subfigure}[t]{0.45\textwidth}
        \includegraphics[width=\textwidth]{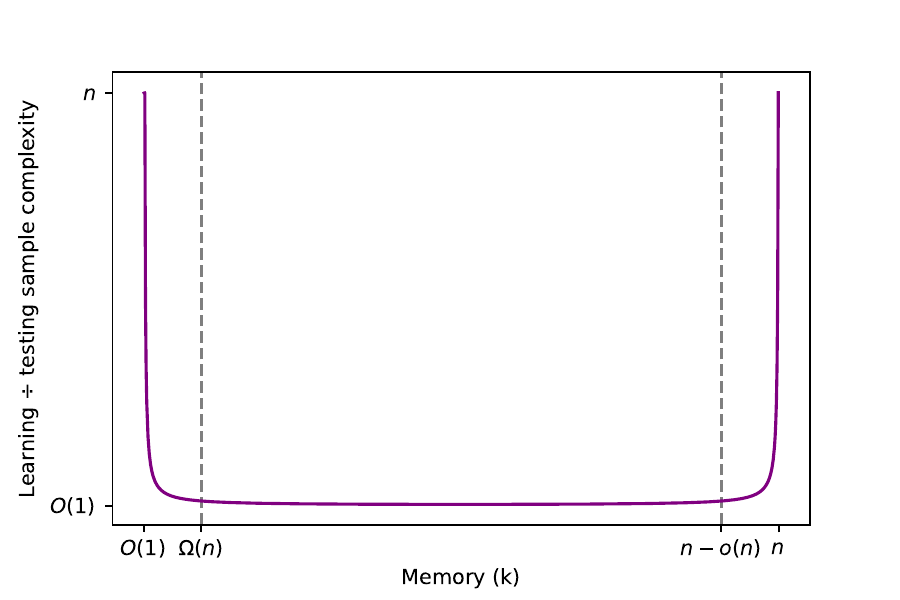}
    \caption{Ratio between learning and testing sample complexity as a function of memory qubits. For $k=cn$, with $0<c<1$, $\Theta(n)$ samples is necessary and sufficient for both learning and testing.}
    \end{subfigure}
    \label{fig:placeholder}
    \caption{Sample complexity of learning and testing stabilizer states with limited quantum~memory.}
    \label{fig:learning_vs_testing}
\end{figure}

\begin{restatable}{theorem}{learning}{(Optimal learning bounds)}\label{thm:stablearning}
Let $k\geq 1$. There is a non-adaptive protocol that uses $k$ qubits of memory~and~$O(n^2/k)$ copies of an unknown $\ket{\psi}\in \Sh$ to learn $\ket{\psi}$. Also, every non-adaptive protocol needs  $\Omega(n^2/k)$~copies.
\end{restatable}

\textbf{Prior work.} The work of~\cite{aaronsontalk} learned stabilizer states with $O(n)$ copies using an entangled measurement on all of them, which was  simplified to require $2$-copy measurements (i.e., $k=n$ qubits of memory) by Montanaro~\cite{montanaro2017learning}. On the other end, the work of~\cite{aaronsontalk} and subsequent works of~\cite{grewal2025efficient,chia2024efficient} showed how to learn stabilizer states using $O(n^3)$ copies and single-copy measurements. As far as we are aware, all these learning algorithms are \emph{adaptive},\footnote{Perhaps it is folklore this is improvable  to $O(n^2)$, but we haven't found any explicit reference/proof of~this.} i.e., future measurements depend on prior measurement outcomes. Our main contribution is twofold: $(i)$ interpolate and give upper bounds for the range  $k=1$ to  $k=n$; $(ii)$ show that the upper bounds can be made \emph{non-adaptive}, which is a more restrictive model of learning. As for lower bounds, information-theoretically the results above are optimal: for  entangled measurements, it takes $O(n^2)$ bits to specify a stabilizer state and one copy of a state has $n$ qubits of information, implying a general $\Omega(n)$ lower bound even for adaptive algorithms. For the single-copy case, a work of~\cite{arunachalam2022optimal} showed that even to learn the subclass of degree-$2$ phase states, one needs $\Omega(n^2)$ copies. We adapt their lower bound and are able to show that $k$ qubits of memory implies a lower bound of $\Omega(n^2/k)$ for non-adaptive~algorithms.

\textbf{Role of Bell sampling.} Like previous works, Bell sampling will still be a useful primitive for us. In particular, our protocols will Bell sample $k$ qubits at a time. This works quite well for the learning algorithm which essentially Bell samples $n/k$ distinct blocks of $k$ qubits. One may hope that a similar strategy then suffices for testing. However, our tester is substantially different and inspired by the hidden shift problem instead. While the learning algorithm performs many Bell samples, the tester only Bell samples $O(1/\eps)$ times and does not use the memory further. The remaining $O((n-k)/\eps)$ samples are used for single-copy measurements, serving as additional evidence that memory is less useful for testing than it is for learning.

\paragraph{Further results.} Beyond the optimal testing and learning bounds, our lower-bound mechanism is versatile enough to prove further lower bounds, which we discuss now. In the task of purity testing,  given an unknown state $\rho$, the goal is to decide if it is a pure state or the maximally mixed state. Prior works showed that the sample complexity with single-copy measurements is $2^{\Omega(n)}$~\cite{chen2022exponential} and $\Omega(\min\{2^{n-k}, 2^{n/2}\})$ with $k$-qubits of memory~\cite{chen2024optimal, gong2024sample}. However, for the latter bounds the authors require that the memory is \emph{measured} \emph{every other round}. We are able to use our techniques to remove that assumption and recover a $2^{\Omega(n-k)}$ lower bound even when the \emph{memory may be coherent} throughout the entire protocol. Also, our lower bound techniques seem to readily apply to other ensembles of states satisfying two ingredients (which we highlight after discussing the proof sketches and in Appendix~\ref{app:purity}). These properties hold for stabilizer states, Haar random states, and could be of independent interest for analyzing other ensembles.

\subsection{Proof sketch of testing upper bound}

We briefly summarize the structure of the tester.  Throughout the section we denote $m=n-k$ for simplicity in notation.   The central
difficulty in our testing algorithm is that the tester has only $k$ qubits of coherent memory, so it
cannot directly Bell-sample all $n$ qubits. Instead, our tester first obtains a partial Bell sample, which reveals a \emph{prefix} of a Pauli label and leaves an
\emph{unknown suffix}. Our key contribution will then be a single-copy algorithm that takes a prefix and checks for the existence of such a suffix (and could be used to learn this suffix as well). We do so via a novel connection to the \emph{hidden shift problem}. We note here that state versions of the hidden subgroup problem have been studied recently, with applications to finding entanglement and learning stabilizer groups~\cite{bouland2025state, hinsche2026abelian}. However, these algorithms require multi-copy measurements (and reduce to Bell sampling for qubit states) while our protocol is single-copy and quite distinct from Bell sampling. Very recent work~\cite{gheorghiu2026quantum} considered state versions of the hidden shift problem given circuits for two states as well.

\begin{enumerate}
\item \textbf{Partial Bell sampling.}
As always, a Pauli label is written as $(x,z)\in\F_2^n\times\F_2^n$. Our partial Bell difference sampling subroutine takes two copies of the unknown state $\ket{\psi}$, measures the
last $m$ qubits of both states in the computational basis and Bell-measures the first $k$ qubits of both states. Repeating this a second time and adding the outcomes bitwise gives a distribution
$$
    Q(a,r)
    =
    \sum_{t\in\F_2^m}
    (p_\psi\star p_\psi)(a_\X r, a_\Z t) \ ,
$$
which is the marginal over $t$ of the standard Bell difference sampling distribution.
For a stabilizer state, $p_\psi$ is uniform on the unsigned stabilizer group
$M_\psi$.  Hence if the prefix $(a,r)$ is observed, then there exists $t^* \in \F_{2}^m$ such that 
\begin{align}
    (a_\X r, a_\Z t^*) \in M_\psi\ .
\end{align}

\item \textbf{The hidden shift.}
For functions, the hidden shift problem is: given access to two unknown functions $f$ and $g$ on an Abelian group $G$ such that there exists some $t^*$ such that $f(x) = g(x+t^*)$, find $t^*$. This problem has been well studied in quantum information as a special instance of the hidden subgroup problem with connections to lattice cryptography~\cite{van2006quantum, childs2010quantum, regev2004quantum, kuperberg2005subexponential,rotteler2010quantum,simon1997power}. 

We will show that recovering $t^*$ can be viewed as a sort of generalization of the hidden shift problem over $\F_2^m$. For the observed prefix $(a,r)$, we define two Pauli families
$$
    F_{a,r}(t):=P_{a_\X r, a_\Z  t},
    \qquad
    G(t):=\id_k\otimes Z(t).
$$
If $F_{a,r}(t^*)$ is a stabilizer, then
$F_{a,r}(t^\ast)\ket{\psi}=\lambda\ket{\psi}$ for some
$\lambda\in\{\pm1\}$. Now compare any other completion $t$ to this stabilizer
prefix and the corresponding Pauli string $F_{a,r}(t)$.  Their product is, up to a phase,
a pure $Z$ operator on the missing $m$ qubits:
$$
    F_{a,r}(t)
    =
    i^{(t^*-t)\cdot r}
    G(t+t^\ast) F_{a,r}(t^*).
$$
Applying this identity to $\ket{\psi}$ gives
$$
    F_{a,r}(t)\ket{\psi}
    =
    \lambda\, i^{(t^*-t)\cdot r}
    G(t+t^\ast)\ket{\psi}.
$$
Thus, if a stabilizer completion $t^\ast$ exists (which it will if the unknown state is a stabilizer state), the indexed family
$t\mapsto F_{a,r}(t)\ket{\psi}$ has the same action on $\ket{\psi}$ as the reference family
$t\mapsto G(t)\ket{\psi}$, but shifted by $t^\ast$ and multiplied by
phases. This precisely is the \emph{hidden-shift} structure that we exploit in our testing protocol, which is inspired by the usual Fourier sampling approach for the hidden subgroup problem.

This Fourier sampling step can be realized with \emph{single-copy measurements}\footnote{We will give two realizations: one directly inspired by standard Fourier sampling that uses $m+1$ ancilla qubits (which are not coherent between rounds) and another ancilla-free implementation that only requires Clifford gates and measurements.} and returns linear~constraints involving $t^*$. Up to some technicalities with phases, sampling a stabilizer state~returns 
\begin{align}
    (s_1, b_1), \ldots, (s_N,b_N) \in \F_2^m \times \F_2\ ,
\end{align}
where $b_i = t^*\cdot s_i \oplus c$, for a fixed $c \in \F_2$. In particular, the data lies entirely in the graph of the affine function $s\mapsto t^*\cdot s \oplus c$. But, if $\ket{\psi}$ is not a stabilizer state then no such $t^*$ need exist and, if not, then the samples do not lie in the graph of an affine function with high probability. Hence, the tester obtains $O(m)$ samples and checks for affine consistency.

\item \textbf{Bad prefixes and the role of random Cliffords.}
The affine-graph test has a conditional soundness guarantee: it rejects once the observed prefix has no large Pauli completion.  The remaining question is why such prefixes should occur often.  Without randomization, partial Bell sampling always examines the same coordinate split: the first $k$ qubits are
Bell-sampled, while the remaining $n-k$ qubits are only partially observed.  A non-stabilizer state could in principle have its large Pauli coefficients arranged so that, relative to this fixed split, many observed prefixes still look good. To deal with this, we observe that applying a \emph{random Clifford} prevents this alignment by viewing the Pauli coefficients of $\ket{\psi}$ in a random symplectic coordinate system. 
\item \textbf{Completing the tester.}
Conditioned on a bad prefix, the hidden-shift graph test rejects with constant
probability. After applying a random Clifford, partial Bell difference sampling produces a bad prefix with probability
$\Omega(\varepsilon)$.  Repeating $O(1/\varepsilon)$ independent outer rounds gives
constant soundness error.  Since each round uses $O(m)=O(n-k)$ copies, the
resulting sample complexity is
$$
    O\!\left(\frac{n-k}{\varepsilon}\right).
$$
\end{enumerate}
\subsection{Proof sketch of testing lower bound}
In order to  obtain our lower bounds, we consider the ensemble of random \emph{degree-2 phase states}, a subset of stabilizer states:
\begin{align}
    \ket{\psi_A} & = \frac{1}{\sqrt{2^n}}\sum_{x\in \F_2^n} (-1)^{x^\top A x}\ket{x}
\end{align}
and show that distinguishing these states  from the maximally mixed state (given $k$ qubits of memory requires $\Omega(n-k)$ samples.
Our lower bound is proven using the so-called likelihood ratios, which have been used often to prove quantum testing/learning lower bounds with memory constraints~\cite{chen2022exponential,chen2024optimal,bubeck2020entanglement}. Ideally, we would show that for a POVM $\{E_{\vec{x}}\}_{\vec{x}}$ representing a protocol with $k$-qubits of memory acting on $t$ copies of the input state, the likelihood ratios defined~as
$$
            L(\vec{x})
        =
        \frac{
        \mathbb E_A\!\left[
            \operatorname{Tr}(E_{\vec{x}}\psi_A^{\otimes t})
        \right]}
        {
        \operatorname{Tr}(E_{\vec{x}} I/2^{nt})
        }
        =
        \frac{2^{nt}}{\Tr[E_{\vec{x}}]}
        \mathbb E_A\!\left[
            \operatorname{Tr}(E_{\vec{x}}\psi_A^{\otimes t})
        \right].
$$ are lower bounded by $1-\delta$ for some small $\delta$, which would imply that the distributions over transcripts $\vec{x}$ induced by the maximally mixed state and a uniformly random $\psi_A$ are close in total variation distance~\cite{chen2022exponential}. Recently Hinsche and Helsen~\cite{hinsche2025single} observed an obstacle in using this method though by showing the existence of product measurements with likelihood ratios that can be \emph{zero}. Hence one cannot hope for a universal lower bound of $1-\delta$. To avoid their counterexample, we instead prove that \emph{most} likelihood ratios must be close to $1$. This indeed is sufficient to prove hardness of distinguishing a random phase state from maximally mixed~state. 

Our lower bounds will hinge upon proving that
\begin{align}
    \E_{\vec{x} \sim \mathcal{P}_{mm}}\left[ \vert L(\vec{x}) - 1\vert \right] = o(1)\ ,
\end{align}
when $t=o(n-k)$.  To bound the (average) likelihood ratios, the first step is to understand $\E_A[\psi_A^{\otimes t}]$. For this, we write out the ensemble average of $t$ copies of a random $\ket{\psi_A}$ and observe that the matrix elements are
    $$
        \mathbb E_A
        (-1)^{
            \sum_r (x^r)^\top A x^r
            +
            \sum_r (y^r)^\top A y^r
        }
        =
        \mathbf 1\!\left[
            \sum_r x^r\otimes x^r
            =
            \sum_r y^r\otimes y^r
        \right].
    $$
    Further, these correlations survive only when the degree-$2$ statistics of the two $t$-tuples match, and the $t$-th moment of the ensemble decomposes     into contributions indexed by orthogonal symmetries
    $O\in \Oc_t(\mathbb F_2)$, where $\Oc_t$ is the group of matrices $O \in \F_2^{t \times t}$ such that $O^\top O = O O^\top = \id$ and $R(O)$ is the representation of $GL(t,\F_2)$ that appears in the Clifford commutant~\cite{gross2021schur} (see Section~\ref{sec:cliff_com} for more details). We show that $\Exp_A[\psi_A^{\otimes t}]$ is close in trace distance to the unnormalized state
\begin{align}
    \sigma_P := \frac{1}{2^{nt}}\sum_{O \in \Oc_{t}}R(O)\ .
\end{align}
 We then split $\Oc_t$ into two parts: $\Oc_{t-1}$ (identified with all matrices that leave the first coordinate fixed) and its complement, $\mathcal M_t : = \Oc_t \backslash \Oc_{t-1}$. We now handle these two contributions separately.  The contribution of $\mathcal O_{t-1}$ will be handled recursively, so the main task is to prove that the $\mathcal M_t$-contribution has negligible average bias. So, by triangle inequality, it suffices to bound
\begin{align}
    \E_{\vec{x} \sim \mathcal{P}_{mm}}\left[ \frac{1}{\Tr[E_{\vec{x}}]}\left\vert \sum_{O \in \mathcal M_t} \Tr[R(O) E_{\vec{x}}]\right\vert \right] & = \frac{1}{2^{nt}}\sum_{\vec{x}} \left\vert \sum_{O \in \mathcal M_t} \Tr[R(O) E_{\vec{x}}]\right\vert\ .
\end{align}
Bounding this quantity is the main combinatorial contribution. Before we describe how to bound this, we give an intuition for the setting where the tester has no memory. Consider $k=0$ and $O=\SWAP_{1,2} \in \mathcal M_t$. Then, the entire protocol can be represented with a collection of product states $\{\varphi_{\vec{x}}\}_{\vec{x}}$ and 
\begin{align}
\Tr_1[(\SWAP_{1,2}\otimes \id_{3:t}) (\varphi_{x_1} \otimes \id_{2:t})] & = \varphi_{x_1} \otimes \id_{3:t}\ ,
\end{align}
and $\Vert \Tr_1[(\SWAP_{1,2}\otimes \id_{3:t}) (\varphi_{x_1} \otimes \id_{2:t})] \Vert_1 = \Tr[\varphi_{x_1}]2^{n(t-2)}$. That is, this contraction has reduced $\SWAP_{1,2}\otimes \id_{3:t}$ from a unitary with trace norm $2^{nt}$ to some operator with trace norm decreased by at least a factor of $2^{2n}$. See Figure~\ref{fig:swap_contract} for an illustration.
\begin{figure}[!ht]
    \centering
    \includegraphics[width=0.5\linewidth]{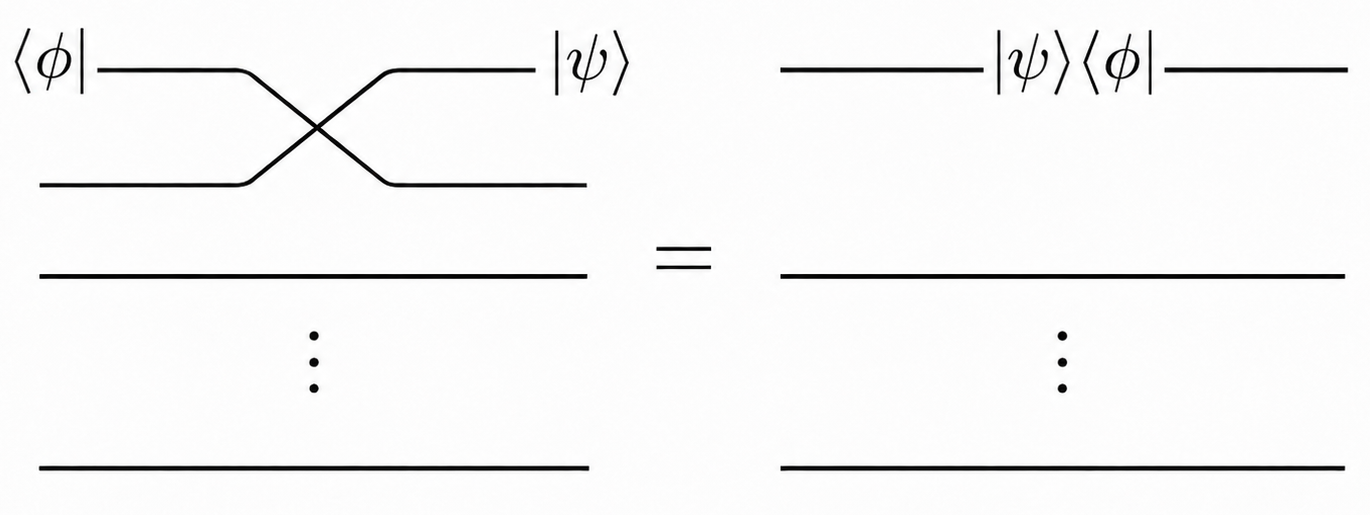}
    \caption{Partially contracting a swap operation. By contracting the first register of $\SWAP_{1,2}$ with an arbitrary rank $1$ operator $\ket{\psi}\bra{\phi}$, the trace norm goes from $2^{nt}$ to $2^{n(t-2)}$, where $t$ is the number of tensor factors in the entire Hilbert space $\Hc^{\otimes t}$.}
    \label{fig:swap_contract}
\end{figure}
We essentially generalize  this reduction in trace norm to not only the permutations  but also to the operator $\sum_{O\in \mathcal M_t}R(O)$. Once we have bounded the contribution from $\mathcal M_t$, the proof is completed by recursively bounding the contributions from $M_{t-1}:= \Oc_{t-1}\backslash \Oc_{t-2}$ and so on.

More rigorously, fix the first-round outcome $x_1$. In the learning-tree representation of a $k$-memory protocol, the first round induces a map
$
V_{x_1}:\mathcal H \to \mathcal M
$
from the first input copy into the $k$-qubit memory space. After this first step, the protocol may continue adaptively, with future measurements depending on both $x_1$ and the memory state. However, when we sum over all continuations $\vec x_{>1}$, the completeness of the remaining subtree allows us to upper bound the whole branch by a single trace norm. More formally, we bound the contribution of all transcripts beginning with $x_1$  by
$$
\left\Vert
(V_{x_1}\otimes \id_{2:t})\cdot \sum_{O\in\mathcal M_t}R(O)\cdot (V_{x_1}\otimes \id_{2:t})^\dagger
\right\Vert_1 .
$$
Thus, the bias of this entire branch of the learning tree is compressed into a trace-norm estimate for the operator $\sum_{O\in \mathcal{M}_t}R(O)$ after its first tensor factor has been passed through the $k$-qubit memory via the map $V_{x_1}$.

This is useful precisely because every $O\in\mathcal M_t$ moves the first basis vector $e_1$, hence the corresponding operator $R(O)$ couples the first copy to the remaining $t-1$ copies. Once the first copy has been measured, this coupling can only be carried forward through the $k$-qubit memory. Intuitively, the first measurement destroys an $n$-qubit correlation, while the memory can preserve at most $k$ qubits of it. The main technical lemma shows that this intuition remains valid even after summing coherently over all $O\in\mathcal M_t$: the possible interference among the many $O \in \mathcal M_t$ terms costs only a $2^{O(t)}$ factor, yielding a bound of the form
$$
\frac{1}{2^{nt}}\sum_{\vec x}
\left|
\sum_{O\in\mathcal M_t}\operatorname{Tr}[R(O)E_{\vec x}]
\right|
\leq
2^{1-(n-k-3t)/2}.
$$
Once this moving-sector contribution is bounded, the remaining terms are those in $\mathcal O_{t-1}$, namely the symmetries that fix the first copy. These have the same form as the original problem on $t-1$ copies. We therefore recurse over
$$
\mathcal M_j:=\mathcal O_j\setminus \mathcal O_{j-1},
\qquad j=t,t-1,\ldots,2,
$$
and eventually show that the average likelihood ratio remains close to $1$ whenever $t=o(n-k)$.

A key technical contribution in both reducing the task of distinguishing $t$ copies of a random phase state $\ket{\psi_A}$ from $\sigma_P$ and proving that the contribution to likelihood ratios from $M_{t}$ is small, is the combinatorial parameter $\rank(\id+O)$, which is the logarithm of the number of bitstrings not fixed by $O$ (i.e., number of $x$ such that $Ox \neq x$). Interestingly, this appears in slightly different ways. For reducing to $\sigma_P$, it is useful for computing the dimension of a certain invariant subspace. For bounding the contribution from $M_t$, it appears when computing the kernels of certain linear maps resulting from taking partial contractions of operators in $M_t$. Either way, one needs to count the number of $O$'s such that $\rank(\id+O) = r$, which we show to be at most $2^{rt}$, which scales with an exponent linear in $t$ instead of $t^2$. Note that the size of $\Oc_t$ scales as $2^{t^2}$ and this is partly why previous lower bounds only were able to handle up to $t = \Omega(\sqrt{n})$.

\paragraph{For the purity testing lower bound.} So far we sketched why distinguishing a random quadratic phase state from the maximally mixed state is hard. To conclude that stabilizer testing for pure states is hard, we will also prove hardness of distinguishing a Haar random state from the maximally mixed state. We get this via noting that our lower bound requires two components that can be extended to other ensembles:
$(1.)$ The $t$th moment of the ensemble is well approximated by a subgroup of binary orthogonal matrices. That is,
    \begin{align}
        \E_\psi[\psi^{\otimes t}] \approx \frac{1}{2^{nt}}\sum_{O \in G \leq \Oc_t} R(O)\ .
    \end{align}
   $(2.)$ Being able to count the number of $O\in G$ such that $\rank(\id+O) = r$.  Given both of these ingredients, our lower bounds techniques can be used to show hardness of distinguishing a random $\psi$ from $\id/2^n$. When $\psi$ is a Haar random state, property 1 above holds since $\E_\psi[\psi^{\otimes t}] \propto \sum_{\pi \in S_t}R(\pi)$~\cite{harrow2013church}. For property 2, note that $\rank(\id+\pi) = \vert \pi \vert$ is the Cayley distance on $S_t$ induced by transpositions~\cite{harrow2023approximate}. In particular, this is upper bounded by $t^{2r}$ for permutations, which is exponentially smaller in $t$ compared to $2^{rt}$ which we have shown above for $\Oc_t$. 

\subsection{Proof sketch of learning bounds}
\paragraph{Upper bound.}
We first describe the upper bound. For a stabilizer state $\ket{\psi}$, our first goal  is to recover the unsigned
stabilizer group $M_\psi$ and then the stabilizer signs.  Since the learner has
only $k$ qubits of coherent memory, it cannot perform a full Bell basis measurement on all $n$ qubits.  Instead, our learner partitions the qubits into blocks $S$ of size at most $k$. For each block, we Bell-measure only the qubits in $S$ across two copies and measure the remaining qubits in the computational basis.  This gives
samples of the form
$
    (a,\Pi_S b),
$
where $(a,b)$ is distributed as a Bell sample from $M_\psi$ and $\Pi_S$ is the projection of the stabilizer tableau onto the $S$ rows. To ensure that these samples are linearly independent, we perform this sampling a constant number of rounds with a random Clifford applied to the state in each round.  With constant probability, a random Clifford puts the state in \emph{full-support form}, i.e., a graph state, so its unsigned stabilizer group is
$$
    M=\{(u,Bu):u\in\F_2^n\}
$$
for a symmetric matrix $B$.  On this branch, a block Bell sample is exactly
$
    (a,\Pi_S Ba),
$
where $a$ is drawn uniform over $\F_2^n$.  Thus each sample gives a random linear equation for the row block $\Pi_S B$.  Taking $O(n)$ samples per block makes the sampled $a$'s span $\F_2^n$ with high probability, so $\Pi_S B$ is recovered by Gaussian elimination.  Repeating over $O(n/k)$ blocks recovers all of $B$, and hence $M$, using $O(n^2/k)$ copies.  

This only recovers the unsigned stabilizer group.  To recover signs, we use a separate precommitted layer of random stabilizer-basis measurements. For each
random Clifford basis, the measurement reveals eigenvalues of a known random commuting Pauli subspace $L_s$. This determines the signs for the Pauli strings in the intersection $L_s \cap M_\psi$, which is nontrivial with constant probability. With $O(n)$ random Clifford basis measurements, the intersections span $M_\psi$ with high probability, which then fully determines the signs of all Paulis in $M_\psi$.  This sign-recovery stage costs only $O(n)$ copies, so the total remains $O(n^2/k)$.

\paragraph{Lower bound.} For this, it again suffices to consider the smaller ensemble of real degree-$2$ phase~states
$$
    \ket{\psi_A}
    =
    2^{-n/2}\sum_{x\in\F_2^n}(-1)^{x^TAx}\ket{x},
$$
where $A$ is a uniformly random upper-triangular matrix over $\F_2$.  The  parameter $A$ contains $\Theta(n^2)$ bits.  A one-copy accessible-information calculation shows that any measurement on a single fresh copy reveals only $O(1)$ bits about $A$.  If the protocol also carries a $k$-qubit coherent memory, then one round can increase the information about $A$ by at most $O(k+1)$ bits (this is where non-adaptivity is used). Thus, after $T$ copies, the final transcript contains at most $O(T(k+1))$ bits of information about $A$.  On the other hand, any learner that identifies the state with constant success probability must identify $A$ with constant success
probability, and hence must obtain $\Omega(n^2)$ bits of information by Fano's inequality.  Therefore
$$
    T(k+1)=\Omega(n^2) \implies T\geq \Omega(n^2/k)
$$
for $1\le k\le n$. We remark that we use the non-adaptivity only in the part where we say each state reveals $\leq k$ qubits.

\subsection{Discussion and open questions}
Our work opens up a few natural questions  for future research. 
\begin{enumerate}
    \item \emph{\textbf{Tolerance in testing.}} In this work, we looked at testing stabilizer states in the usual property testing framework. More recently there have been works that have looked at tolerant testing~\cite{ad2024tolerant,bao2024tolerant,mehraban2024improved}, where the goal is to test if the unknown state is \emph{close} or \emph{far} from stabilizer states. What is the complexity of tolerant testing with $k$ qubits of memory? We believe that our tester could be used for a computationally inefficient tolerant tester by computing the largest correlation with any affine subspace, but is there an efficient tolerant~tester?
\item    \emph{\textbf{Getting the optimal $\varepsilon$ dependence.}} We improved the $\eps$ dependence for testing to $1/\eps$, but is this optimal? Perhaps one method for proving a lower bound with $\eps$ dependence would be to take random quadratic functions and randomly perturb their outputs with, say, i.i.d Bernoulli noise. The corresponding ensemble average state would smoothly interpolate between a sum over all of $\Oc_t$ (the case of quadratic phase states studied here) and only permutations (Haar random states).
    \item \emph{{\textbf{Other ensembles of states}}.} As previously mentioned, our lower bounds readily extend to Haar random states because $t$-copies of such a state can be represented by a subgroup of $\Oc_t$. A similar statement may hold for, say, higher-degree phase states as well. One could hope to interpolate between the exponential hardness of Haar random states and the linear hardness of degree-$2$ phase states as a function of degree. This is potentially related to $\eps$-dependence as discussed above.
    \item \emph{\textbf{Testing doped states.}} Apart from stabilizer states, there have been many recent works that have looked at learning and testing $t$-doped states, i.e., states produced by Clifford circuits along with $t$ many $T$ gates. Understanding the testing and learning complexity of these states with limited memory is a natural follow-up question.
    \item \emph{\textbf{Adaptive learning lower bound.}} Our  lower bound for learning algorithms for stabilizer states only applies for non-adaptive algorithms. It is unclear how to extend our lower bounds for \emph{adaptive} algorithms. As far as we are aware, there aren't many techniques that are able to prove such lower bounds, so proving an adaptive learning lower bound for stabilizer states is an interesting question.
    \item \emph{\textbf{Continuous variable testing.}} Apart from the discrete set of stabilizer states, there have been many works that have looked at learning and testing Fermionic Gaussian states and Bosonic states~\cite{bittel2025optimal,mele2025efficient,mele2026quantum}. We leave open the question of what is the optimal dependence of learning and testing these states with $k$ qubits of memory.
\end{enumerate}

\paragraph{Acknowledgment.} We thank  Marcel Hinsche, Jonas Helsen, Lennart Bittel, and Arkopal Dutt for several discussions during the initial stages of this project. Additionally, we thank Sergey Bravyi, Sitan Chen, the members of the Eisert group, and Sabee Grewal for helpful discussions. L.S. acknowledges funding from Munich Quantum Valley, Berlin Quantum, BMFTR (PasQUops, Hybrid++, QuSol), ERC (DebuQC), DFG (CRC 183, SPP 2514), and MATH+.

\noindent \textbf{Use of LLM.}  In version $1$ of this paper, we first had the following results: 
\begin{enumerate}[label=$(\roman*)$,itemsep=0.5pt]
    \item Tight learning bound of $\Theta(n^2/k)$
    \item Testing upper bound of $O(n-k)$ 
    \item Testing lower bound of $\Omega(\sqrt{n-k})$ (for adaptive testers); $\Omega(n-k)$ (for non-adaptive testers).
\end{enumerate}  
Before uploading our paper on arXiv, we asked  Claude if there was any immediate weakness in our lower bound  and surprisingly it found one combinatorial lemma where our lower bounds were loose. Our prior lower bounds were proven using nearly the same steps, i.e., fixing parts of the transcript, partially contracting $\sigma_P$ with the corresponding POVM aspects, and summing over the remaining rounds. For adaptive algorithms, we were using a \emph{double coset decomposition} $\Oc_{t-1} \backslash \Oc_{t} / \Oc_{t-1}$ to group elements of $\Oc_{t}$ together. Effectively, this amounted to a finer graining of $\Mc_t$, and we then bounded \emph{each} of the resulting terms using Lemma~\ref{lem:cut_rank} and applied recursion to $\Oc_{t-1}$ and so on. Claude's suggestion was to simplify this argument by instead just \emph{bounding all of $\mathcal M_t$} instead of these distinct double cosets. To make this rigorous, we needed Lemma~\ref{lem:cross_term}, an extension of Lemma~\ref{lem:cut_rank}, which was suggested by Claude. With everything, we were able to generalize our testing lower bound to the optimal $\Omega(n-k)$ for even \emph{adaptive testers}. The entire proof has been rewritten and proofread by us (and any further mistakes are our own).

\section{Preliminaries}

\subsection{Paulis and Cliffords} 
The $2$-qubit Pauli  matrices are defined as follows
$$\id=\begin{pmatrix}
1 & 0\\
0 & 1
\end{pmatrix}, X=\begin{pmatrix}
0 & 1\\
1 & 0
\end{pmatrix}, Y=\begin{pmatrix}
0 & -i\\
i & 0
\end{pmatrix},Z=\begin{pmatrix}
1 & 0\\
0 & -1
\end{pmatrix}
$$
It is well-known that the $n$-qubit Pauli matrices $\{\id,X,Y,Z\}^n$ form an  {orthonormal basis} for $\mathcal{B}((\mathbb{C}^2)^{\otimes n})$.  We write a Pauli label on $n$ qubits as
$(x,z)\in \F_2^n\times \F_2^n$, where $x$ is the $\X$-label and $z$ is the
$\Z$-label and use the Hermitian Pauli convention
$$
    P_{x,z}:=i^{x\cdot z} \X^x \Z^z\ .
$$
These operators are often referred to as \emph{Weyl operators} and it is not hard to see that  these operators $\{P_{x,z}\}_{x,z \in \FF_2^{n}}$ are orthogonal. For $x,y \in \mathbb{F}_2^{2n}$, where we write $x=(x_1, x_2)$ with $x_1$ denoting the first $n$ bits of $x$ and $x_2$ denoting the last $n$ bits (similarly for $y=(y_1,y_2)$), we define the \emph{symplectic inner product} as
\begin{equation}
    [x,y] = \la x_1, y_2 \ra + \la x_2, y_1 \ra \mod 2.
    \label{eq:symplectic_inner_product}
\end{equation}
Here, all additions are over $\F_2$. Under this inner product, multiplication of Paulis takes the form
\begin{align}
    P_u P_v & = (-1)^{u_Z \cdot v_\X}i^{u_\X \cdot u_\Z + v_\X \cdot v_\Z - (u_\X \oplus v_\X)\cdot (u_\Z \oplus v_\Z)}P_{u + v}\ ,
\end{align}
where inner products above are integer valued. Observe that 
\begin{align}
    P_u P_v = (-1)^{[u,v]}P_{v}P_u\ ,
\end{align}
and the Paulis $P_{x,z}$ and $P_{x',z'}$ commute if and only if this symplectic inner
product vanishes, i.e. $u_\X \cdot v_\Z = u_\Z \cdot v_\X$. Since a set of commuting Paulis corresponds to a set of elements of $\F_2^{2n}$ with $[u,v] = 0$, we can associate such sets with certain subspaces of $\F_2^{2n}$. Such subspaces will be isotropic, meaning that $[u,v] = 0$ for all $u,v \in W$. Maximal isotropic subspaces play an important role in studying stabilizer states.

\begin{definition}[Lagrangian subspace]
A subspace $M \leq \F_2^{2n}$ is said to be Lagrangian if $M=M^\omega$, where $\omega$ denotes the symplectic complement. Equivalently,
\begin{itemize}
    \item For all $x,y\in M$, $[x,y] = 0$.
    \item $\dim(M) = n$ and hence $M$ is a maximal isotropic subspace.
\end{itemize}
\end{definition}

This, plus a sign function, defines a stabilizer state.

\begin{definition}[Stabilizer state]
    A pure $n$-qubit state $\ket{\psi}$ is a stabilizer state if there exists a maximal Abelian subgroup of Pauli matrices such that $P\ket{\psi} = \pm\ket{\psi}$ for all $P$ in this group. In the symplectic representation, $\ket{\psi}$ is fully determined by:
    \begin{enumerate}
        \item A Lagrangian subspace $M_\psi \leq \F_2^{2n}$.
        \item A sign function $\chi: M_\psi \rightarrow \{\pm 1\}$ such that $P_u\ket{\psi} = \chi(u)\ket{\psi}$.
    \end{enumerate}
\end{definition}

For a stabilizer state $\ket{\psi}$, we will use $M_\psi$ to denote the corresponding Lagrangian subspace.
    
\begin{fact}
\label{fact:cocycle-extension}
Let $M_\psi \le\F_2^{2n}$ be a Lagrangian subspace and let $\chi:M_\psi \to\{\pm1\}$ be the
sign function of some stabilizer state with label space $M_\psi$.  Then for all
$u,v\in M_\psi$,
$$
    \chi(u+v)=(-1)^{\varphi(u,v)}\chi(u)\chi(v)\ .
$$
In particular, the values of $\chi$ on any spanning set of $M_\psi$ determine $\chi$ on
all of $M_\psi$.
\end{fact}

\begin{proof}
Let the phase be
\begin{align}
    \varphi(u,v) & := \frac{1}{2}(2u_\Z\cdot v_\X + u_\X \cdot u_\Z + v_\X \cdot v_\Z - (u_\X \oplus v_\X)\cdot (u_\Z \oplus v_\Z))\ ,
\end{align}
which is an integer since $P_u$ and $P_v$ commute. Then, we have that $P_u P_v = (-1)^{\varphi(u,v)} P_{u+v}$. Using $P_u \ket{\psi} = \chi(u)\ket{\psi}$ and $P_v \ket{\psi} = \chi(v)\ket{\psi}$, we obtain the identity. Iterating expresses $\chi$ on any element of $M_\psi$ as a product of values of $\chi$ on the spanning set multiplied by known signs.
\end{proof}

We will also require a normal form for stabilizer states.

\begin{fact}
\label{fact:normal-form}
Every $n$-qubit stabilizer state can be written, up to global phase, as
$$
    \ket\psi
    =
    \frac1{\sqrt{|V|}}
    \sum_{x\in V} i^{\ell(x)}(-1)^{q(x)}\ket{x},
$$
where $V\subseteq\F_2^n$ is an affine subspace, $\ell:V\to\F_2$ is affine-linear,
and $q:V\to\F_2$ is quadratic.  After a known affine change of variables sending
$V$ to $\F_2^r\times\{0^{n-r}\}$, the state has the form
$$
    \ket{\psi_{A,\ell}}
    =
    2^{-r/2}\sum_{x\in\F_2^r} i^{\ell\cdot x}(-1)^{x^TAx}\ket{x}\otimes
    \ket{0^{n-r}},
$$
where $A\in\F_2^{r\times r}$ is upper triangular.  On the active $r$ qubits,
the stabilizer label subspace is the graph
$$
    M_B=\{(u,Bu):u\in\F_2^r\},
    \qquad B=A+A^T + \diag(\ell).
$$
The remaining information in the stabilizer state is the vector of signs of any
generating set for this Lagrangian subspace.
\end{fact}
We call a stabilizer state full-support if, in the normal form as defined above, its support is all of $\F_2^n$, i.e., up to global phase, can be written as
$
    2^{-n/2}\sum_{x\in\F_2^n} i^{\ell(x)}(-1)^{q(x)}\ket{x},
$
where $\ell$ is linear and $q$ is quadratic

\paragraph{Cliffords and stabilizers.} Clifford unitaries are those generated by Hadamard gate $\textsf{Had}=\frac{1}{\sqrt{2}}\begin{pmatrix}
1 & 1\\
1 & -1
\end{pmatrix}$, controlled-$X$ gate and $S=\begin{pmatrix}
1 & 0\\
0 & i
\end{pmatrix}$ gate. The output of Clifford circuits on the all $\ket{0^n}$ input are always stabilizer states. Note that the Clifford group forms the normalizer of the Pauli group in the unitary group on $n$-qubits. That is $CP_uC^\dagger = (-1)^{c(u)} P_v$ for some $v\in\F_2^{2n}$ and quadratic function $c$. In the symplectic representation, conjugation by a Clifford unitary corresponds to multiplication by some matrix in $\text{Sp}(2n,\F_2)$. That is, $u$ is mapped to $Bu$ for $B \in \text{Sp}(2n,\F_2)$. Further, the Clifford group acts transitively on the Pauli group. The following fact is a consequence of this.

\begin{fact}\label{fact:random_lagrangian}
    Let $M \leq \F_2^{2n}$ be a uniformly random Lagrangian subspace and $\mathcal{Z}:= \{0\}\times \F_2^n$ the all $\Z$ Lagrangian. Then,
    \begin{align}
        \Pr[\dim(M \cap \mathcal{Z}) = 0] = \prod_{j=1}^n (1+2^{-j})^{-1} \geq 0.4\ .
    \end{align}
\end{fact}
\begin{proof}
    This fact, as well as extensions to larger intersections, was proved as Corollary 2 in ~\cite{kueng2015qubit}. We provide a proof here for convenience. 

    The Clifford group acts transitively on Lagrangian subspaces through the usual symplectic representation. Let $C$ be a uniformly random Clifford unitary. Since $C$ is uniformly random, we have that $C(M)$ is uniformly distributed over all Lagrangian subspaces. The number of such subspaces in $\F_2^{2n}$ is $\prod_{j=1}^n(2^j+1)$. Lagrangian subspaces such that $M\cap \mathcal{Z} = \{0\}$ must take the form $\{(u,Bu): u \in \F_2^{n}\}$, where $B\in \F_2^{n \times n}$ is symmetric. There are $2^{n(n+1)/2}$ symmetric matrices over $\F_2$. Therefore,
    \begin{align}
        \Pr[\dim(M \cap \mathcal{Z}) = 0] = \frac{2^{n(n+1)/2}}{\prod_{j=1}^n (2^j+1)} =  \prod_{j=1}^n (1+2^{-j})^{-1} \geq 0.4\ .
    \end{align}
\end{proof}

\paragraph{Characteristic distribution.} For an $n$-qubit pure state $\ket\psi$, define the characteristic function and
characteristic distribution
$$
    c_\psi(u):=2^{-n/2}\Tr(P_u\psi),
    \qquad
    p_\psi(u):=|c_\psi(u)|^2=2^{-n}|\braketbra{\psi}{P_u}{\psi}|^2.
$$
If $\ket\psi$ is a stabilizer state and $M_\psi\le\F_2^{2n}$ is the corresponding Lagrangian subspace, then
$$
    p_\psi(u)=2^{-n}\cdot \mathbf 1[u\in M_\psi].
$$
A partial converse holds. If $p_\psi$ is concentrated well on some Lagrangian subspace $M$, then the stabilizer fidelity of $\ket{\psi}$ cannot be too small.
\begin{fact}[{\cite[Theorem 3.3]{gross2021schur}}; {\cite[Corollary 7.4]{grewal2024improved}}]
\label{fact:lower_bound_stabilizer_fidelity_pPsi_lagrangian_subspace}
For any $n$-qubit quantum state $\ket{\psi}$ and a Lagrangian subspace $M \subset \FF_2^{2n}$
\begin{equation*}
    \calF_{\Sh}\ket{\psi}) \geq \sum_{x \in M} p_\psi(x)\ , \quad \mathcal{F}_{\cal \Sh}(\ket{\psi}):=\max_{\ket{\phi}\in \Sh}|\langle \psi|\phi\rangle|^2 \ .
\end{equation*}
\end{fact}

\subsection{Fourier analysis}
We will work with complex-valued Boolean functions $f: \mathbb{F}_2^n \rightarrow \mathbb{C}$. The inner product (or correlation) of two functions $f, g: \mathbb{F}_2^n \rightarrow \mathbb{C}$ is given by
\begin{equation}
    \la f , g \ra = \Exp_x [f(x) \overline{g(x)}].
\end{equation}
Throughout this work we will work with the standard Fourier transform, wherein 
  the Fourier decomposition of $f$ is defined  as
$$
f(x)=\sum_{S \in \FF_2^{n}} \widehat{f}(S)\chi_S(x)=\sum_{S \in \FF_2^{n}}  \widehat{f}(S)(-1)^{\langle S,x\rangle }
$$
where $\chi_S(x)=(-1)^{\langle S, x\rangle}$ and the \emph{Fourier coefficients} $\widehat{f}(S)\in \mathbb{C}$ are defined as
$$
 \widehat{f}(a) = \frac{1}{2^n} \sum_{x \in \FF_2^{n}} (-1)^{\langle a,x\rangle } f(x).
$$
We define the convolutions of two functions $f, g: \mathbb{F}_2^n \rightarrow \mathbb{C}$ as
\begin{equation}
    (f \star g)(x) = \sum_{z \in \mathbb{F}_2^n} \left[f(z) g(x+z) \right].
\end{equation}

We will require the following characterization of the Fourier transform of the characteristic distribution of a pure state $\ket{\psi}$.
\begin{fact}[{\cite[Proposition 8.4]{grewal2025efficient}}]
\label{fact:fourier-duality}
    For a pure state $\ket{\psi}$, the standard Fourier transform of the characteristic distribution is given by
    \begin{align}
        p_\psi(a_\X, a_\Z) & = \sum_{b_\X, b_\Z \in \F_2^n} \widehat{p}_\psi(b_\X, b_\Z) (-1)^{a_\X \cdot b_\X + a_\Z \cdot b_\Z}\ ,
    \end{align}
    where the Fourier coefficients are 
    \begin{align}
        \widehat{p}_\psi(b_\X, b_\Z) & = \frac{1}{2^n}p_\psi(b_\Z, b_\X)\ .
    \end{align}
\end{fact}

\subsection{Stochastic orthogonal group }\label{sec:cliff_com}
We will repeatedly use the stochastic orthogonal group and a specific representation of this group.

\begin{definition}[Stochastic orthogonal group]
    The stochastic orthogonal group $\Oc_t$ is the subgroup of matrices in $\F_2^{t \times t}$ such that $OO^\top = O^\top O = \id$. Equivalently, $\langle x, y\rangle = \langle Ox, Oy\rangle$ for all $x,y \in \F_2^t$.
\end{definition}

This group has a defining representation on $\mathbb{C}^{2^t}$ given by
\begin{align}
    r(O)\ket{x} & = \ket{Ox}\ ,
\end{align}
where $x\in \F_2^{t}$. In particular, this is a permutation module. This can be readily extended to a representation on $(\mathbb{C}^{2^n})^{\otimes t}$ which we will make extensive usage of. This representation is
\begin{align}
    R(O) := r(O)^{\otimes n}\ ,
\end{align}
where $r(O)$ is understood to act transversally on a collection of $t$ qubits, one from each copy of $\mathbb{C}^{2^n}$. To make this explicit, consider the state $\ket{x} = \otimes_i \ket{x_i}$ in $(\mathbb{C}^{2^n})^{\otimes t}$. To this state associate a matrix
\begin{align}
  X := \begin{pmatrix}
        -x_1 -\\
        -x_2-\\
        \vdots\\
        -x_t-
    \end{pmatrix}\in \F_2^{t \times n} .
\end{align}
Then, $R(O)$ acts by permuting these basis states according to $OX$.

It is worthwhile to mention here that this group is intimately connected to the commutant of the Clifford group. In particular, the seminal work~\cite{gross2021schur} showed that, when $t-1 \leq n$, the Clifford commutant is spanned by $R(L)$ where $L$ is a stochastic Lagrangian subspace (the characterization of the commutant was then extended to arbitrary $t$~\cite{bittel2025complete}). Here $R(L) = r(L)^{\otimes n}$ with $r(L) = \sum_{(x,y)\in L}\ket{x}\bra{y}$, which may not be a unitary. To each $O \in \Oc_{t}$ there is an associated Lagrangian subspace $L_O := (Ox,x)$. Then, $R(L_O) = R(O)$ as described above, which is a unitary. Our lower bound for testing sidesteps dealing with the non-unitary elements of the Clifford commutant.  
We will need the following fact which shows that stochastic orthogonal matrices act transitively on certain equivalence classes of binary matrices $X\in \F_2^{t \times n}$ (again, think of $X$ as encoding $\ket{x} = \otimes \ket{x_i}$).
\begin{lemma}\label{lem:stoch_transform}
    Let $n > t$ and $X,Y\in \F_2^{t \times n}$ be full rank matrices. Then, $X^\top X = Y^\top Y$ if and only if there is a matrix $O\in \F_2^{t \times t}$ such that $OY = X$ and $OO^\top = O^\top O = \id$.
\end{lemma}
\begin{proof}
    Since $X$ has full row rank, it has a right inverse, which we will denote by $X^{-R}$. Similarly, $X^\top$ has a left inverse, which is $(X^{-R})^\top$. Hence,
    \begin{align}
        \id & = (X^{-R})^\top Y^\top Y X^{-R}\ .
    \end{align}
    Define $O:= (X^{-R})^\top Y^\top$, which is orthogonal, completes the proof.
\end{proof}

\section{Technical toolkit}

\subsection{Partial Bell sampling}
Bell sampling  is a well-known subroutine that has been used often in stabilizer learning and testing. We briefly describe that first subroutine before describing the ``partial" version of it. 

\begin{figure}
    \centering
    \includegraphics[width=0.65\linewidth]{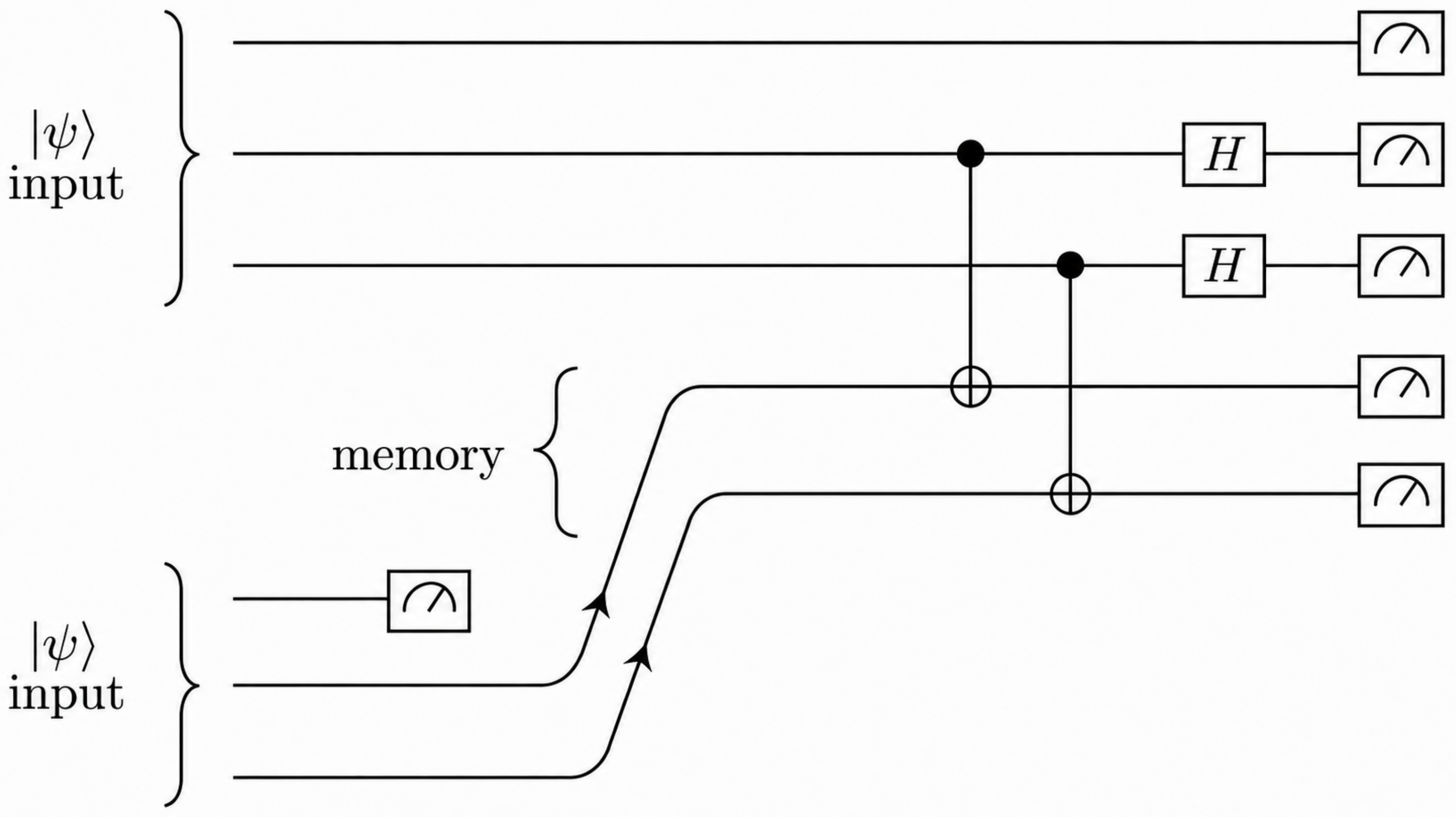}
    \caption{Partial Bell sampling using a quantum memory. Part of the input state $\ket{\psi}$ is measured in the computational basis while the rest is copied onto the $k$-qubit memory. In this figure, $\ket{\psi}$ is a $3$ qubit state and there are $2$ memory qubits. Then, a fresh sample of $\ket{\psi}$ is loaded and a Bell sample is performed between the stored qubits and the corresponding qubits of the new copy.}
    \label{fig:bell_sampling}
\end{figure}

\textbf{Bell sampling.}  As its name suggests, Bell sampling on the state $\ket{\psi} \otimes \ket{\phi}$ corresponds to measuring $\ket{\psi} \otimes \ket{\phi}$ in the Bell basis i.e., the orthonormal basis of $(W_x \otimes I) \ket{\Phi^{+}}$ with $\ket{\Phi^+}$ being the state of $n$ EPR pairs (over $2n$ qubits) $\ket{\Phi^+} := 2^{-n/2} \sum_{x \in \FF_2^n} \ket{x}\ket{x}$. The measurement outcome from Bell sampling is thus a $2n$ bit string $x \in \FF_2^{2n}$ that corresponds to a Weyl operator $W_x$. It was observed by Montanaro~\cite{montanaro2017learning} that one can learn an $n$-qubit stabilizer state using $O(n)$ samples from Bell sampling on $\ket{\psi}^{\otimes 2}$. Similarly, Bell difference sampling corresponds to Bell sampling on $\ket{\psi}^{\otimes 2}$ twice to produce outcomes $x,y \in \FF_2^{2n}$ and then returning $z = x + y$. Bell difference sampling was proposed in~\cite{gross2021schur} for intolerant testing stabilizer states. 

\begin{fact}[Theorem 3.2 of~\cite{gross2021schur}]\label{fact:bell-sampling}
    Let $\ket{\psi}$ be an $n$-qubit pure state. The distribution over Pauli labels $u\in\F_2^{2n}$ of Bell sampling and Bell difference sampling are, respectively,
    \begin{align}
        q_\psi(u)  = 2^{-n}\vert \langle \psi^* \vert P_u \vert \psi\rangle\vert^2,\quad 
        Q_\psi(u)  = q_\psi \star q_\psi (u) = p_\psi \star p_\psi(u)\ .
    \end{align}
\end{fact}
If $\ket{\psi}$ is indeed a stabilizer state, $p_\psi \star p_\psi = p_\psi$ and hence Bell difference sampling is uniform on the state's stabilizer group.

A caveat of both Bell sampling and Bell difference sampling however is that they require having $n$ qubits of quantum memory. Since the primary focus of this work is algorithms with $k<n$ qubits of memory, we introduce \emph{partial bell sampling}, wherein one measures part of two copies of a state in the Bell basis and the rest in the computational basis. We define this protocol in Algorithm~\ref{alg:partial-bell_samp} and illustrate it in Figure~\ref{fig:bell_sampling}. Similarly, partial Bell difference sampling (Algorithm~\ref{alg:partial-bell}) corresponds to partial Bell sampling $\ket{\psi}$ twice to obtain outcomes $(a_1,r_1)$ and $(a_2,r_2)$ and returns $(a_1+a_2, r_1+r_2)$.

\begin{algorithm}[h]
\caption{Partial Bell sampling using $k$ qubits of memory}
\label{alg:partial-bell_samp}
\begin{spacing}{1.05}
\KwInput{Two copies of $\ket\psi$.}
\KwOutput{$(a,r)$ with $a\in\F_2^{2k}$ and $r\in\F_2^m$.}
Store the first $k$ qubits of one copy of $\ket\psi$\;
Measure the remaining $m$ qubits in the computational basis, obtaining
$z_{0}\in\F_2^m$\;
Load another copy of $\ket\psi$\;
Bell-sample its first $k$ qubits against the $k$ qubits in memory,
obtaining $a\in\F_2^{2k}$\;
Measure the remaining $m$ qubits in the computational basis, obtaining
$z_{i}\in\F_2^m$\;
\Return{$(a,\ z_0 \oplus z_1)$}\;
\end{spacing}
\end{algorithm}

\begin{lemma}\label{lem:partial_bell_marginal}
    Let $q^k_\psi$ denote the distribution induced on $(a,r)$, where $a\in\F_2^{2k}$ and $r\in\F_2^{n-k}$, by partial Bell sampling from a pure state $\psi$. Then,
    \begin{align}
        q^k_\psi(a,r) & = \sum_{t\in \F_2^{n-k}} q_\psi(a_\X r, a_\Z t)\ .
    \end{align}
\end{lemma}
\begin{proof}
    To the outcome $(a,r)$ we can associate the POVM aspect
    \begin{align}
        \ket{\Phi_a}\bra{\Phi_a} \otimes \sum_x \ket{x,x+r}\bra{x,x+r}\ ,
    \end{align}
    where $\ket{\Phi_a} = 2^{-k/2}\sum_y (-1)^{y\cdot a_\Z}\ket{y,y+a_\X}$ is a Bell state on the first $k$ qubits. We now sum over Bell states on the remaining $n-k$ qubits to obtain
    \begin{align}
        \sum_{t\in\F_2^{n-k}} \ket{\Phi_{r,t}}\bra{\Phi_{r,t}} & = \sum_{x,y }\frac{1}{2^{n-k}} \ket{x,x+r}\bra{y,y+r} \sum_t (-1)^{t\cdot (x+y)} = \sum_x \ket{x,x+r}\bra{x,x+r}\ .
    \end{align}
    Hence, the POVM aspect corresponding to outcome $(a,r)$ can be written as $\sum_{t}\ket{\Phi_{a_\Z r, a_\X t}}\bra{\Phi_{a_\Z r, a_\X t}}$. 
\end{proof}
\begin{corollary}\label{corr:partial_bell_diff}
    Let $Q_\psi^k$ denote the distribution on $(a,r)$, where $a\in \F_2^{2k}$ and $r \in \F_2^{n-k}$, induced by partial Bell difference sampling from a pure state $\psi$. Then,
    \begin{align}
        Q_\psi^k(a,r) & = \sum_{t\in\F_2^{n-k}} (p_\psi \star p_\psi)(a_\X r, a_\Z t)\ .
    \end{align}
\end{corollary}
\begin{proof}
    By Lemma~\ref{lem:partial_bell_marginal}, the distribution is given by
    \begin{align}
        Q_\psi^k(a,r) & = \sum_{b\in\F_2^{2k}}\sum_{s\in\F_2^{n-k}} q_\psi^k(b,s)q_\psi^k(a+b,s+r)\\
        & = \sum_{b\in\F_2^{2k}}\sum_{s,u\in\F_2^{n-k}}  q_\psi(b_\X s, b_\Z u)\left( \sum_{t\in\F_2^{n-k}}q_\psi((a_\X + b_\X)(s+r), (a_\Z + b_\Z)t \right)\ .
    \end{align}
    Relabeling $t$ with $t+u$ and uses Fact~\ref{fact:bell-sampling} completes the proof.
\end{proof}

We will require one more result about partial Bell sampling.

\begin{lemma}
    Let $H$ be a subspace of $\F_2^{2k}\times \F_2^{n-k}$. The probability of observing an element of $H$ from partial Bell difference sampling is
    \begin{align}
        Q_\psi^k(H) & = \vert H \vert 2^{n-k}\sum_{a \in (H\times \mathcal{Z}_{n-k})^\omega}p^2_\psi(a)\ ,
    \end{align}
    where $H\times \mathcal{Z}_{n-k} = \{(b_\X r, b_\Z t) \ \vert \ (b,r)\in H, t \in \F_2^{n-k}\}$ is all possible $Z$ completions of strings in $H$ and $\omega$ denotes the symplectic complement.
\end{lemma}
\begin{proof}
    By Corollary~\ref{corr:partial_bell_diff}, we have that
    \begin{align*}
        Q_\psi^k(H) & = \sum_{(b,r)\in H} \sum_{t \in \F_2^{n-k}} (p_\psi \star p_\psi)(b_\X r, b_\Z t)\\
        & = \sum_{(b,r)\in H} \sum_{t \in \F_2^{n-k}} \sum_{v,w\in \F_2^{n}} p_\psi(v,w) p_\psi(v+b_\X r, w+b_\Z t)\\
        & = \sum_{(b,r)\in H} \sum_{t \in \F_2^{n-k}} \sum_{v,w\in \F_2^{n}} \left( \sum_{c,d\in\F_2^{n}} \widehat{p}_\psi(c,d) (-1)^{v\cdot c + w\cdot d}\right)\left( \sum_{c,d\in\F_2^{n}} \widehat{p}_\psi(c',d') (-1)^{(v+b_\X r)\cdot c' + (w+b_\Z t)\cdot d'}\right)\\
        & = \sum_{v,w \in \F_2^{n}}\sum_{c,d\in \F_2^n} \sum_{c',d'\in\F_2^n}\widehat{p}_\psi(c,d)\widehat{p}_\psi(c',d') (-1)^{v\cdot(c+c')+w\cdot(d+d')}\sum_{(b,r)\in H}\sum_{t\in \F_2^{n-k}}(-1)^{c'\cdot b_\X r + d'\cdot b_\Z t}\ ,
    \end{align*}
    where $\widehat{p}_\psi$ is the standard Fourier transform of the characteristic function. 

    Now, the summations over $H$ and $t$ can be combined by noting that this is equivalent to summing over $H\times \mathcal{Z}_{n-k}$. That is, we sum over $a_\X r$ and $a_\Z t$ for $(a,r)\in H$ and arbitrary $t$. We now appeal to the fact that
    \begin{align*}
        \sum_{w \in W}(-1)^{\langle x,w\rangle} & = \vert W \vert \id\{ x\in W^\perp\}\ ,
    \end{align*}
    for any inner product on a vector space $V$ over $\F_2$, subspace $W \subseteq V$, and arbitrary element $x\in V$. Using this fact, we have that
    \begin{align*}
        \sum_{(b,r)\in H}\sum_{t\in \F_2^{n-k}}(-1)^{c'\cdot b_\X r + d'\cdot b_\Z t} & = \vert H \times \mathcal{Z}_{n-k}\vert \id\{(c',d') \in (H \times \mathcal{Z}_{n-k})^\perp\}\\
        & = \vert H \vert 2^{n-k}\id\{(c',d') \in (H \times \mathcal{Z}_{n-k})^\perp\}\ ,
    \end{align*}
    where $\perp$ here indicates the orthogonal complement with respect to the standard inner product.
    Since $v$ and $w$ are unconstrained, summing over these terms introduces the factors $2^{n}\delta_{c,c'}$ and $2^{n}\delta_{d,d'}$. Hence, we are left with
    \begin{align*}
        Q_\psi^k(H) & = \vert H \vert 2^{3n-k} \sum_{(v,w)\in (H\times\mathcal{Z}_{n-k})^\perp}\widehat{p}_\psi(v,w)^2 = \vert H \vert 2^{n-k} \sum_{(v,w)\in (H\times\mathcal{Z}_{n-k})^\omega}\widehat{p}_\psi(v,w)^2\ ,
    \end{align*}
    where we have used the fact that $\widehat{p}_\psi(v,w)  = \frac{1}{2^n}p_\psi(w,v)$ (Fact~\ref{fact:fourier-duality}).
\end{proof}

\subsection{Ensembles corresponding to deg-2 phase states}\label{sec:phase_states}
Let $A$ be uniformly distributed over upper-triangular matrices in
$\F_2^{n\times n}$, and define
$$
    \ket{\psi_A}
    :=
    2^{-n/2}
    \sum_{x\in\F_2^n}
    (-1)^{x^T A x}
    \ket{x}.
$$
That is, $\ket{\psi_A}$ is a degree 2 phase state. It is not hard to see that these states form a 1-design.

\begin{fact}\label{fact:phase_1_design}
    Let $A$ be a uniformly random upper triangle matrix in $\F_2^{n \times n}$. Then,
    \begin{align}  
        \E_A[\ket{\psi_A}\bra{\psi_A}] & = \frac{\id}{2^n}\ .
    \end{align}
\end{fact}
\begin{proof}
    By direct expansion,
    \begin{align}
        \E_A[\ket{\psi_A}\bra{\psi_A}] & = \frac{1}{2^n}\sum_{x,y} \E_A\left[ (-1)^{x^\top A x + y^\top A y} \right]\ket{x}\bra{y}\\ .
    \end{align}
    If $x\neq y$, there is some index $i$ such that $x_i = 0$ yet $y_i = 1$. Since $A_{i,i}$ is $0$ and $1$ with equal probability, the off diagonal terms vanish.
\end{proof}

Now we consider the ensemble corresponding to the $t$-fold tensor products of the states $\ket{\psi_A}$. To that end, the lemma below characterizes the entries of this ensemble matrix.

\begin{lemma}
\label{fact:expectatationoverrandomdegree-2}
For every $\vec{x},\vec{y}\in (\F_2^n)^t$, we have that 
$$
\mathbb{E}_A[(-1)^{\vec{x}^\top A \vec{x}+\vec{y}^\top A \vec{y}}]=\id\left[\sum_{r\in [t]}x^r\otimes x^r=\sum_{r\in [t]}y^r\otimes y^r\right],
$$
where the expectation is over a uniformly random upper triangular $A\in \F_2^{n\times n}$, $\vec{x}^\top A\vec{x} = \sum_i (x_i)^\top A x_i$, and $x\otimes x\in \F^{n^2}$.
\end{lemma}
\begin{proof}
   Writing $\vec{x}=(x_1,\ldots,x_t)$, the proof follows from first observing that
    \begin{align*}
        \vec{x}^\top A \vec{x}+\vec{y}^\top A \vec{y}&=\sum_i (x_i)^\top A x_i+\sum_i (y_i)^\top A y_i\\
        &=\sum_i \sum_{k\leq \ell}(x_i)_{k} A_{k,\ell} (x_i)_{\ell}+\sum_i \sum_{k,\ell}(y_i)_{k} A_{k,\ell} (y_i)_{\ell}\\
        &=\sum_{k\leq \ell}A_{k,\ell} \sum_i\Big((x_i\otimes x_i)_{k,\ell}+(y_i\otimes y_i)_{k,\ell}\Big),
    \end{align*}
    where $(x\otimes x)_{k,\ell}=x_kx_\ell$. So we have that
    \begin{align*}
        \mathbb{E}_A\Big[(-1)^{\vec{x}^\top A \vec{x}+\vec{y}^\top A \vec{y}}\Big]&=\mathbb{E}_A\Big[(-1)^{\sum_{k,\ell}B_{k,\ell} \sum_i\Big((x_i\otimes x_i)_{k,\ell}+(y_i\otimes y_i)_{k,\ell}\Big)}\Big]\\
        &=\mathbb{E}_A\Big[\prod_{k\leq \ell}(-1)^{A_{k,\ell} \sum_i\Big((x_i\otimes x_i)_{k,\ell}+(y_i\otimes y_i)_{k,\ell}\Big)}\Big]\\
        &=\prod_{k\leq \ell}\mathbb{E}_B\Big[(-1)^{A_{k,\ell} \sum_i\Big((x_i\otimes x_i)_{k,\ell}+(y_i\otimes y_i)_{k,\ell}\Big)}\Big]\\
        &=\prod_{k\leq \ell}\Big[\Big(\sum_i x_i\otimes x_i\Big)_{k,\ell}=\Big(\sum_i y_i\otimes y_i\Big)_{k,\ell}\Big]\\
        &=\id\Big[\sum_i x_i\otimes x_i=\sum_i y_i\otimes y_i\Big],
    \end{align*}
    where the final indicator used that $(x\otimes x)_{k,\ell}=(x\otimes x)_{\ell,k}$.
\end{proof}

Then, we have that
\begin{align*}
  \rho_P 
  &:= \mathbb{E}_A [\psi_{A}^{\otimes t}]=\frac{1}{2^{nt}}\sum_{\vec{x},\vec{y}\in (\F_2^{n})^t} \id\Big[\sum_{r\in [t]}x^r\otimes x^r=\sum_{r\in [t]}y^r\otimes y^r\Big]\ketbra{\vec{x}}{\vec{y}}
\end{align*}
Next, we note that the space $\vec{x},\vec{y}\in (\F_2^n)^t$ simplifies significantly when restricted to the space $\sum_{r\in [t]}x^r\otimes x^r=\sum_{r\in [t]}y^r\otimes y^r$. To see this, consider $\vec{x}$ and stack the bitstrings as the row of a matrix $X\in \F_2^{t \times n}$. That is,
\begin{align*}
    X & = \begin{pmatrix}
        -x_1-\\
        -x_2-\\
        \vdots \\
        -x_t-
    \end{pmatrix}\ .
\end{align*}
Then, one can verify that $(\sum_i x_i \otimes x_i)_{k,\ell} = (X^\top X)_{k,\ell}$. Let $Y$ correspond to $\vec{y}$. Then, $\sum_{r\in [t]}x^r\otimes x^r=\sum_{r\in [t]}y^r\otimes y^r$ if and only if $X^\top X = Y^\top Y$. Later when we invoke these ensembles of states, we argue that these $\vec{x}=(x^1,\ldots,x^t)\in (\F_2^n)^t$ are linearly independent and using Lemma~\ref{lem:stoch_transform} one can argue that $Y$ is related to $X$ by an orthogonal transformation. We discuss this in more detail when discussing the testing lower bounds in Section~\ref{sec:simplifyhardensemble}.

\subsection{Structure of protocols with quantum memory}
Now we detail the structure of protocols which use repeated copies of some unknown state and also a $k$-qubit quantum memory. At round $i$, the algorithm is given a fresh copy of the unknown state $\rho$ and the memory is in the state $\eta_i$. Let $\Hc$ denote the Hilbert space for the input ($\rho$) and $\Mc$ the Hilbert space for the memory. Then, the algorithm may perform an arbitrary channel
\begin{align*}
    \mathcal{N}: \mathcal{B}(\mathcal{H}_{in}\otimes \Mc) \rightarrow \mathcal{B}(\Mc)\ .
\end{align*}
This accounts for any operation on the inputs at round $i$ that may or may not leave the memory qubits coherent. This is illustrated in Figure~\ref{fig:protocol_with_mem}. We will formalize this model with a version of the learning tree framework, following the exposition of~\cite{chen2022exponential}.

\begin{definition}[Tree representation for  bounded memory protocols]\label{def:tree_mem}
The state of an algorithm that uses $t$ rounds and a $k$-qubit memory can be modeled as a tree $\mathcal{T}$ where each node $\vec{x}_{<i}$ indicates the history of the algorithm thus far. Further:

\begin{itemize}
    \item Each node $\vec{x}_{<i}$ is associated with a state $\eta_{\vec{x}_{<i}}$ for the memory (which depends on the input states $\rho$ and the transcript thus far).
    \item For every child node $\vec{x}_{\leq i} = \vec{x}_{<i}x_i$ of $\vec{x}_{<i}$, there is a linear, completely positive, and trace non-increasing~map 
    \begin{align*}
    T^{x_i}_{\vec{x_i}_{<i}}:\mathcal{B}(\mathcal{H}_{in}\otimes \mathcal{H}_m) \rightarrow \mathcal{B}(\mathcal{H}_m)
    \end{align*} 
    such that $\eta_{\vec{x}_{\leq i}} = T^{x_i}_{\vec{x}_{<i}}(\rho\otimes \eta_{\vec{x}_{< i}})$.

    \item Fixing a node $\vec{x}_{<i}$, $\sum_{x_i \in N_{\vec{x}_{<i}}} T_{\vec{x}_{<i}}^{x_i}$ is a CPTP channel. Here $N_{\vec{x}_{<i}}$ is all children of $\vec{x}_{<i}$.

    \item In the final round, the entire system is measured with some POVM $\{E_{\vec{x}_{< t}}^{x_t}\}_{x_t}$, where $x_t$ is the final outcome.
\end{itemize}
\end{definition}
Note that $T^{x_i}_{\vec{x}_{<i}}$ may depend on the transcript at step $i$ ($\vec{x}_{<i}$). Hence, this captures all protocols that may adaptively use single-copy measurements augmented with $k$-qubits of coherent memory.

The class of protocols captured in Definition~\ref{def:tree_mem} is a bounded memory version of a quantum causal tester/strategy. Closely related formalisms include quantum strategies and co-strategies, quantum combs, and memory channels~\cite{chiribella2009theoretical,gutoski2007toward,kretschmann2005quantum}. In this work we use an equivalent representation via Kraus operators and a tree structure as this allows us to explicitly work with maps acting on the coherent memory.

We will now show that we can put any such protocol into a standard form, which will be convenient for our proofs. If there is no quantum memory, it is known that the entire protocol can be represented with a collection of product states $\{w_{\vec{x}} \psi_{\vec{x}}\}$ such that $\sum_{\vec{x}} w_\ell \psi_{\vec{x}} = \id^{\otimes t}$. Here each $\psi_{\vec{x}}$ represents a specific leaf node in the learning tree $\mathcal{T}$. With quantum memory, this standard form no longer applies. However, we can show that a similar statement holds with the Kraus operators for each map in the protocol.

\begin{lemma}[Standard form for learning tree]\label{lem:tree_standard}
    Let $\mathcal{T}$ be the learning tree for an algorithm using $k$-qubits of memory. Then, it is without loss of generality to assume that
    \begin{enumerate}
        \item Each map $T_{\xll}^{x_i}$ is composed of a single Kraus operator:
    \begin{align*}
        T_{\vec{x}_{<i}}^{x_i}(\rho\otimes \eta_{\xll}) & = K_{\xll}^{x_i} (\rho \otimes \eta_{\xll}) (K_{\xll}^{x_i})^\dagger\ .
    \end{align*}
    Then, $\sum_{x_i \in N_{\xll}}(K_{\xll}^{x_i})^\dagger K_{\xll}^{x_i} = \id_{in} \otimes \id_m$ for all $\xll\in\mathcal{T}$. Also, the final measurement round can be assumed to be composed entirely of rank one aspects.

    \item The initial state of the memory is some fixed pure state $\ket{\eta_0}$.
    \end{enumerate}
\end{lemma}
\begin{proof}
  We prove each item separately. 
  
  \textbf{\emph{Item 1.}} Say that $T_{\xll}^{x_i}$ has the set of Kraus operators $\{K_{\xll}^{x_{i},j}\}_j$ such that $T_{\xll}^{x_{i}}(\rho\otimes \eta_{\xll}) = \sum_j K_{\xll}^{x_{i},j} (\rho\otimes \eta_{\xll})(K_{\xll}^{x_{i},j} )^\dagger$. To each operator $K_{\xll}^{x_i,j}$ we can associate a new node $\vec{x}_{\leq i};j$, and thus split $T^{x_i}_{\xll}$ into some number of new nodes and maps $T^{x_i,j}_{\xll}$ with a single Kraus operator. Clearly the resulting protocol could simulate the original protocol. 

    A similar argument holds for the POVM applied to the memory at the end of the protocol: take a spectral decomposition refined into rank one operators $E_{\vec{x}_{< t}}^{x_t} = \sum_j w_j \varphi_v^{\ell,j}$. We can associate a new leaf node to each operator in this sum and simulate the original protocol. Going forward, we thus denote the final measurement by $w_{\vec{x}} \varphi_{\vec{x}}$.

  \textbf{\emph{Item 2.}}  At the start, it is without loss of generality to assume that the memory state $\eta_0$ is a pure state because the maps in the first round $\{T_{x_0}^{x_1}\}_{x_1}$ could all first involve a preparation channel acting on the memory.
\end{proof}

A consequence of this normal form is that we can define operators that represent the protocol and allow us to map inputs to states of the memory.

\begin{definition}\label{def:pass_operators}
    Let $[i:j] = \{i,i+1,\ldots,j\}$ be an interval of integers. For a transcript $\vec{x}_{1:j}$, we define the forward operators as
    \begin{align}
        V^{\vec{x}_{i:j}}_{\vec{x}_{< i}} & : \Hc_{i:j} \otimes \Mc \rightarrow \Mc:: \bigotimes_{\ell=i}^j \ket{\phi_\ell} \otimes \ket{\eta} \mapsto K_{\vec{x}_{<j}}^{x_j}\left(\ket{\phi_j}\otimes K_{\vec{x}_{<{j-1}}}^{x_{j-1}}\left(\cdots K_{\vec{x}_{<i}}^{x_i}(\ket{\phi_i}\otimes \ket{\eta}) \right) \right)\ . 
    \end{align}
    If $j=t$, we take the forward operator to be
    \begin{align}
        V^{\vec{x}_{i:t}}_{\vec{x}_{< i}} :=  \sqrt{w(\vec{x})}\ket{\varphi_{\vec{x}}}\bra{\varphi_{\vec{x}}}V_{\vec{x}_{<i}}^{\vec{x}_{i:t-1}}: \Mc \rightarrow \Hc \otimes \Mc\ .
    \end{align}
\end{definition}
We think of $V_{\xll}^{\vec{x}_{i:j}}$ as taking the memory at round $i$ as well as $j-i+1$ copies of an input $\rho$ and mapping to the memory at state $j$ after measuring the string of outcomes $x_{i}x_{i+1}\ldots x_{j}$. For example, $V_0^{\vec{x}_{1:j}} \ket{\phi}^{\otimes j}\otimes \ket{\eta_0}$ is the state of the memory at round $i$ beginning from the start state for the memory $\ket{\eta_0}$.  Note that the prior outcomes/round $\vec{x}_{< i}$ may dictate the later maps, since the algorithm can be adaptive, but otherwise only appears in the state of the memory at stage $i$. However, the forward pass operators cannot depend on $\vec{x}_{>j}$ since the past does not depend on the future.

An important observation is that the forward operators compose. That is, $V_{\vec{x}_{<j}}^{\vec{x}_{j:k}}V_{\vec{x}_{<i}}^{\vec{x}_{i:j}} = V_{\vec{x}_{<i}}^{\vec{x}_{i:k}}$. This will be quite useful in our lower bound for testing. In particular, if we fix the first outcome $x_1$, then all subsequent outcomes can be represented by $V_{\vec{x}_{<2}}^{\vec{x}_{2:t}}V_{0}^{x_1}$. This allows us to construct a POVM with an iterative structure that represents the protocol.

\begin{lemma}\label{lem:protocol_povm}
    Let $\left\{V_{\vec{x}_{<i}}^{\vec{x}_{i:j}}\right\}_{i,j}$ be the forward operators for a protocol. Then, the probability of observing transcript $\vec{x}$ is given by $\Tr[E_{\vec{x}} \rho^{\otimes t}]$ where
    \begin{align}
        E_{\vec{x}} & = (\id \otimes \bra{\eta_0}) \left( V^{\vec{x}}_0 \right)^\dagger V_0^{\vec{x}} (\id \otimes \ket{\eta_0}) \in \mathcal{B}\left(\Hc^{\otimes t}\right)\ .
    \end{align}
\end{lemma}
\begin{proof}
    Say that the outcome $\vec{x}$ was observed upon measuring $\ket{\phi}^{\otimes t}$. By construction, $V_0^{\vec{x}_{1:t-1}}\ket{\phi}^{\otimes t}\otimes \ket{\eta_0}$
    is the state of the memory going into the final round. In the final round, the probability of observing $x_t$ is then given by
    \begin{align}
        w(\vec{x})\left\vert \langle \varphi_{\vec{x}}\vert V_0^{\vec{x}_{1:t-1}}(\ket{\phi}^{\otimes t}\otimes \ket{\eta_0})\right\vert^2 & = \Tr[E_{\vec{x}} \phi^{\otimes t}]\ .
    \end{align}
\end{proof}

\section{Testing upper bound}

In this section we will prove the following theorem.

\begin{theorem}\label{thm:upper}
    Let $k\geq 0,\varepsilon>0$. There is an adaptive protocol that uses $k$ qubits of memory and $O((n-k)/\varepsilon)$ copies of an unknown $\ket{\psi}$ to distinguish between $\mathcal{F}_{\cal \Sh}(\ket{\psi})=1$ vs.~$\mathcal{F}_{\cal \Sh}(\ket{\psi})\leq 1-\varepsilon$.
\end{theorem}


\subsection{Partial Bell sampling}

The first subroutine is a partial version of Bell difference sampling. It Bell-samples
the first $k$ qubits using the available memory, while measuring the remaining
$m=n-k$ qubits in the computational~basis.

\begin{algorithm}[H]
\caption{Partial Bell difference sampling using $k$ qubits of memory}
\label{alg:partial-bell}
\begin{spacing}{1.05}
\KwInput{Four copies of $\ket\psi$.}
\KwOutput{$(a,r)$ with $a\in\F_2^{2k}$ and $r\in\F_2^m$.}
\For{$i=0$ \KwTo $1$}{
    Store the first $k$ qubits of one copy of $\ket\psi$\;
    Measure the remaining $m$ qubits in the computational basis, obtaining
    $z_{i,0}\in\F_2^m$\;
    Load another copy of $\ket\psi$\;
    Bell-sample its first $k$ qubits against the $k$ qubits in memory,
    obtaining $a_i\in\F_2^{2k}$\;
    Measure the remaining $m$ qubits in the computational basis, obtaining
    $z_{i,1}\in\F_2^m$\;
}
\Return{$(a_0+a_1,\ z_{0,0}+z_{0,1}+z_{1,0}+z_{1,1})$}\;
\end{spacing}
\end{algorithm}

The following Lemma was proved in the technical toolkit (Lemma~\ref{corr:partial_bell_diff}), but we repeat it here for the reader.
\begin{lemma}
\label{lem:partial-bell-distribution}
Let $Q^k_\psi$ be the output distribution of Algorithm~\ref{alg:partial-bell}.  Then,
for every $(a,r)\in\F_2^{2k}\times\F_2^m$,
$$
    Q^k_\psi(a,r)=\sum_{t\in\F_2^m}(p_\psi\star p_\psi)\bigl(a_\X r, a_\Z t\bigr).
$$
\end{lemma}

When $\psi$ is a stabilizer state, $p_\psi \star p_\psi$ is again $p_\psi$. Let $M_\psi$ be the Lagrangian subspace corresponding to a stabilizer state $\psi$. Then,
\begin{align}
    Q_\psi^k(a,r) & = \frac{\vert \{ t\in \F_2^{n-k}\ \vert \ (a_\X r, a_\Z t) \in M_\psi \}\vert}{2^n}\ ,
\end{align}
which is the fraction of stabilizers of $\psi$ that have the prefixes $(a,r)$. In particular, if we measure $(a,r)$ and $\psi$ is a stabilizer state, this means that there must be some $t\in \F_2^{n-k}$ such that $(a_\X r, a_\Z t)\in M_\psi$. Hence, it remains to check that such a $t$ does indeed exist. This is exactly the role of the next subroutine. 

\subsection{The hidden-shift subroutine}

\subsubsection{Hidden shift problem}
Before presenting our algorithm, we will provide some background on the hidden shift problem. Let $G$ be a group (usually taken to be Abelian) and let $f$ and $g$ be functions taking $G$ to some finite set $S$. Further, we assume that there exists some $s\in G$ such that $f(x) = g(x+s)$ for all $x\in G$. Given quantum query access to these functions, the goal is to find the shift $s$\cite{childs2010quantum,van2006quantum,rotteler2010quantum,kuperberg2005subexponential,ettinger2000quantum}. Note that this is an instance of the hidden subgroup problem on the group $G \rtimes_\varphi \mathbb{Z}_2$ where $\varphi: \mathbb{Z}_2 \rightarrow \text{Aut}(G)$ is the homomorphism $\varphi(b)(x) = (-1)^b x$. For the special case of $G=\mathbb{Z}_N$, this problem is actually equivalent to the hidden subgroup problem on $D_N$~\cite{kuperberg2005subexponential,ettinger2000quantum}, a problem heavily studied for its connections to lattice cryptography~\cite{regev2004quantum}.

For our purposes, $G=\mathbb{Z}_2^n$ (the additive group of bitstrings), in which case $\mathbb{Z}_2^n \rtimes_\varphi \mathbb{Z}_2$ is simply $\mathbb{Z}_2^{n+1}$. Then, this can be solved as an instance of the Abelian hidden subgroup problem/Simon's problem~\cite{simon1997power,lomont2004hidden}. Let $h:\mathbb{Z}_2^{n+1}\rightarrow S$ be defined by $h(x,0) = f(x)$ and $h(x,1)=g(x)$. From here, the standard approach of Fourier sampling resolves the hidden shift:
\begin{enumerate}
    \item Initialize two registers $\ket{0}_G$ and $\ket{0}_f$.
    \item Apply $H^{\otimes n+1}$ to $\ket{0}_G$.
    \item Conditioned on the state of the first register ($G$), query the function $h(x,b)$, yielding the state
    \begin{align}
        \frac{1}{\sqrt{2^{n+1}}}\sum_x \ket{x}\otimes(\ket{0}\ket{f(x)} + \ket{1}\ket{g(x)})\ .
    \end{align}
    \item Apply $H^{\otimes n+1}$ to the first register.
    \item Measure the first register in the computational basis to obtain $y\in \F_2^{n+1}$ satisfying~$y\cdot s = 0$. 
\end{enumerate}
By repeating this protocol $O(n)$ times, enough information is learned to recover $s$ with high probability.

\subsubsection{Stabilizer hidden shift}
In this subsection we will provide an algorithm inspired by the hidden shift problem that takes a prefix $(a,r)$ and checks for the existence of a suffix $t^* \in \F_2^{m}$ such that $P_{(a_\X r, a_\Z t^*)}$ stabilizes the unknown state $\psi$. This algorithm could also be used to learn $t^*$ assuming it is unique (which occurs with at least constant probability when applying random Clifford to an input state)\footnote{More explicitly, if $M_{C\psi} \cap \mathcal{Z} = \{(0,0)\}$, then $t^*$ must be unique. This occurs with probability at least $0.4$ when applying a random Clifford to a stabilizer state $\psi$. Hence, the hidden shift protocol could be used to learn a stabilizer of $\psi$ with at least constant probability and $O(n)$ samples. Then, repeating yields an $O(n^2)$ single-copy learning algorithm.}.

The algorithm we present here makes use of $n+1$ ancilla qubits and is a direct analogue of Fourier sampling. While it uses ancilla qubits, this protocol is still a single-copy measurement since the ancillas are not left coherent and are simply for implementing the POVM. However, it may still be undesirable to use ancillas, and we present an implementation that uses no ancillas and only Clifford gates in Section~\ref{sec:shift_no_ancilla}. To keep the connection to the hidden shift problem palpable, we continue with Algorithm~\ref{alg:hidden-shift}.

We will need a couple of definitions for this section. For a pair $(a,r) \in\mathbb F_2^{2k}\times\mathbb F_2^m$, define the two Pauli families
$$
    F_{a,r}(t)
    :=
    P_{a_\X r, a_\Z t} = P_a\otimes P_{r,t},\text{ and }   G(t)
    :=
    P_{0,0t}
    =
    I_k\otimes Z(t), 
$$
where $Z(t)$ is an all-Z Pauli string. Now, say that we have obtained the bitstrings $(a,r)$ via partial Bell difference sampling. For readability, we will write $F(t) = F_{a,r}(t)$. Assuming that $\ket{\psi}$ is indeed a stabilizer state, Lemma~\ref{lem:partial-bell-distribution} implies that there exists some $t^* \in \F_2^{n-k}$ such that $P_{a_\X r, a_\Z t}\in M_\psi$. With our Pauli convention,
\begin{align}
    F(t) = i^{(t^*-t)\cdot r}G(t+t^*)F(t^*)\ ,
\end{align}
where $(t^*-t)\cdot r$ is the integer inner product. Therefore
\begin{align}
    F(t)\ket{\psi} = i^{(t^*-t)\cdot r}\lambda G(t+t^*)\ket{\psi}\ .
\end{align}
That is, $F(t)$ and $G(t)$ have the same action on $\ket{\psi}$ up to a shift $t^*$ and a sign $i^{(t^*-t)\cdot r}$. As we know the phase factor $i^{t\cdot r}$, this can be undone to prepare the state
\begin{align}
    i^{t\cdot r}F(t)\ket{\psi} & = i^{t^*\cdot r}\lambda G(t+t^*)\ket{\psi}\ .
\end{align}
Now, the action of $i^{t\cdot r}F(t)$ is the same as $G(t)$ up to a shift and a fixed phase factor $i^{t^*\cdot r}$. While we do not know this phase factor, since we do not know $t^*$, we can guess the parity of $t^* \cdot r$ and, if we guess correctly, make the phase some fixed real value. Hence, our stabilizer hidden shift protocol takes as input a guess of this phase, $\beta \in \F_2$.

\begin{algorithm}[H]
\caption{Stabilizer hidden-shift sample}
\label{alg:hidden-shift}
\begin{spacing}{1.05}
\KwInput{$(a,r)\in\F_2^{2k}\times\F_2^m$, a phase bit $\beta\in\F$, and one copy of $\ket\psi$.}
\KwOutput{$(s,b)\in\F_2^m\times\F_2$.}
Initialize an $m$-qubit register $q$, a control qubit $c$, and the input copy in
$\ket{0}^{\otimes m}_q\ket0_c\ket\psi$\;
Apply Hadamards to $q$ and $c$\;
Apply the controlled Pauli operation
$$
    \sum_{t\in\F_2^m}\ketbra{t}{t}_q\otimes
    \Big(\ketbra{0}{0}_c\otimes F_{a,r}(t)
    +\ketbra{1}{1}_c\otimes G(t)\Big).
$$
Apply the phase operation
$$
    \sum_{t\in\F_2^m}\ketbra{t}{t}_q\otimes
    \left(i^{\beta+t\cdot r}\ketbra{0}{0}_c+\ketbra{1}{1}_c\right)\otimes I;
$$
Apply Hadamards to $q$ and $c$\;
Measure $q$ and $c$ in the computational basis and output the result $(s,b)$\;
\end{spacing}
\end{algorithm}

Next we define the set of affine graphs, which will be related to the output of the hidden shift subroutine. 
The hidden-shift subroutine outputs pairs
$  (s,b)\in \mathbb F_2^m\times\mathbb F_2,
$ 
For $u\in\mathbb F_2^m$ and $c\in\mathbb F_2$,~let
$$
    \chi_{u,c}(s):=u\cdot s\oplus c
$$
be an affine function from $\mathbb F_2^m$ to $\mathbb F_2$.  Its graph is
the subset
\begin{align}
    \label{eq:graphchi}
    \operatorname{Graph}(\chi_{u,c})
    :=
    \{(s,b)\in\mathbb F_2^m\times\mathbb F_2:
      b=\chi_{u,c}(s)\}.
\end{align}
Thus saying that a distribution $P$ on
$\mathbb F_2^m\times\mathbb F_2$ is supported on the graph of an affine
function means that every outcome $P(s,b) > 0$ only if $b = u\cdot s \oplus c$ for the same fixed pair $(u,c)$.

\begin{lemma}
\label{lem:hidden-shift-completeness}
Let $\ket\psi$ be a stabilizer state, and let $M_\psi\le\F_2^{2n}$ be the
subspace of Pauli labels that stabilize $\ket\psi$ up to signs.  Suppose
partial Bell Difference sampling outputs $(a,r)\in\F_2^{2k}\times\F_2^m$.  Then there is
$t^*\in\F_2^m$ such that $(a_\X r, a_\Z t)\in M_\psi$. Moreover, for the phase choice $\beta=t^*\cdot r \pmod 2$, the output
distribution of Algorithm~\ref{alg:hidden-shift} is supported  on
$\operatorname{Graph}(\chi_{t^\ast,c})$
for some $c\in\mathbb F_2$.
\end{lemma}

\begin{proof}
For a stabilizer state, $p_\psi$ is uniform on $M_\psi$.  By
Lemma~\ref{lem:partial-bell-distribution}, if $(a,r)$ is observed, then there
must exist $t^*\in\F_2^m$ such that $(a_\X r, a_Zt^*) \in M_\psi$.  Hence for
some $\lambda\in\{\pm1\}$,
$$
    F_{a,r}(t^*)\ket\psi=\lambda\ket\psi.
$$
For readability write $F(t)=F_{a,r}(t)$. Recall that, with our Pauli convention,
$$
    F(t) \ket{\psi} = i^{(t^*-t)\cdot r}\lambda G(t+t^*)\ket{\psi}\ .
$$
After the first Hadamards and the controlled Pauli operation in
Algorithm~\ref{alg:hidden-shift}, the joint state is
$$
    \frac{1}{\sqrt{2^{m+1}}}
    \sum_{t\in\F_2^m}\ket t_q
    \left(\ket0_c F(t)\ket\psi+\ket1_c G(t)\ket\psi\right).
$$
Using the previous identity and then applying the phase operation gives
$$
    \frac{1}{\sqrt{2^{m+1}}}
    \sum_{t\in\F_2^m}\ket t_q
    \left(\ket0_c i^{\beta + t^*\cdot r}\lambda G(t+t^*)\ket\psi
    +\ket1_c G(t)\ket\psi\right).
$$
After the final Hadamards, the unnormalized post-measurement vector
corresponding to outcome $(s,b)$ is proportional to
$$
    \sum_{t\in\F_2^m}(-1)^{t\cdot s}
    \left(\lambda i^{\beta + t^*\cdot r}G(t+t^*)+(-1)^bG(t)\right)\ket\psi.
$$
Set $\eta=\lambda i^{\beta + t^*\cdot r}$. Changing variables $u=t+t^*$ in the
first sum, this becomes
$$
    \left(\eta(-1)^{t^*\cdot s}+(-1)^b\right)
    \sum_{u\in\F_2^m}(-1)^{u\cdot s}G(u)\ket\psi.
$$
Now choose $\beta=t^*\cdot r\pmod 2$.  Then $\eta=\pm\lambda\in\{\pm1\}$. If
$\eta=1$, the scalar prefactor vanishes unless $b=t^*\cdot s$.  If
$\eta=-1$, it vanishes unless $b=t^*\cdot s\oplus 1$.  Thus all possible
outputs satisfy $b=t^*\cdot s\oplus c$, where $c=0$ if $\eta=1$ and $c=1$
if $\eta=-1$. Here ``supported on the graph of an affine function'' means that the only
outcomes with nonzero probability lie in
$\Graph(\chi)$ (as defined in Eq.~\eqref{eq:graphchi}) 
For $\chi(s)=t^*\cdot s\oplus c$, this is exactly the condition
$b=t^*\cdot s\oplus c$ for every possible output $(s,b)$.
\end{proof}

Note that if $\beta \neq t^* \cdot r$ then, for all $(s,b)$ in the support of the sampling distribution, $(s,b)$ and $(s,b\oplus 1)$ are equally likely. That is, the distribution is far from being supported on the graph of an affine function.

\subsubsection{Ancilla-free stabilizer hidden shift}\label{sec:shift_no_ancilla}

In this subsection we will show that the stabilizer hidden shift (Algorithm~\ref{alg:hidden-shift}) can be implemented without ancillas and only Clifford gates. This implementation is exactly the same POVM. The reader that is not interested in removing the need for ancillas can safely skip to Section~\ref{sec:stab_test_full} where we present the complete testing algorithm.

We separate the implementations into two cases: $r\neq 0$ (Algorithm~\ref{alg:shift_simple}) and $r= 0$ (Algorithm~\ref{alg:shift_simple_r0}).

\begin{algorithm}[h]
\caption{Ancilla-free stabilizer hidden-shift sample ($r\neq 0$)}\label{alg:shift_simple}
\begin{spacing}{1.05}
\KwInput{$(a,r)\in\F_2^{2k}\times\F_2^m$ with $r$ non-zero, a phase bit $\beta\in\F$, and one copy of $\ket\psi$.}
\KwOutput{$(s,b)\in\F_2^m\times\F_2$.}
Measure $P_a$ on the first $k$ qubits, obtaining eigenvalue $\varphi$\;
Choose some $i\in [m]$ with $r_i=1$, corresponding to physical qubit $k+i$\;
Using qubit $k+i$ as the control, apply $X(r')$ to the remaining qubits, where $r'$ is the bitstring on all qubits in $\{k+1,\ldots,n\}$ other than $k+i$\;
If $\beta = 0$, apply $H$ to qubit $k+i$. Otherwise, apply $HS$ to qubit $k+i$, with $S$ the phase gate\;
Measure qubit $k+i$ and obtain a bit $c$\;
Measure the remaining $n-k-1$ qubits to obtain a bitstring $s'$\;
Uniformly at random return either $(0_is', c \oplus \frac{1-\varphi}{2})$ or $(1_i(s'+r'), c \oplus \frac{1-\varphi}{2} \oplus \beta)$\ ;
\end{spacing}
\end{algorithm}

\begin{algorithm}[h]
\caption{Ancilla-free stabilizer hidden-shift sample ($r= 0$)}\label{alg:shift_simple_r0}
\begin{spacing}{1.05}
\KwInput{$(a,0)\in\F_2^{2k}\times\F_2^m$, a phase bit $\beta\in\F$, and one copy of $\ket\psi$.}
\KwOutput{$(s,b)\in\F_2^m\times\F_2$.}
Measure the last $n-k$ qubits in the computational basis to obtain $s$\;
If $\beta = 0$, measure $P_a$ on the first $k$ qubits, obtaining eigenvalue $\varphi$ and output $(s,(1-\varphi)/2)$\;
Otherwise, output $(s,0)$ or $(s,1)$ with equal probability\;
\end{spacing}
\end{algorithm}

\begin{lemma}
    Algorithm~\ref{alg:hidden-shift} and Algorithm~\ref{alg:shift_simple}/Algorithm~\ref{alg:shift_simple_r0} perform the same POVM on an input state $\ket{\psi}$.
\end{lemma}
\begin{proof}
    First, assume that $r \neq 0$. After fixing $\beta$, the Kraus operator corresponding to the measurement $(s,b)$ is
    \begin{align}
        K_{s,b}: = \frac{1}{2^{n-k+1}}\sum_{t\in\F_2^{n-k}} (-1)^{t\cdot s} \left( i^{\beta + t\cdot r}F(t) + (-1)^b G(t) \right)\ .
    \end{align}
    It is straightforward to verify that
    \begin{align}
        \frac{1}{2^{n-k}} \sum_t (-1)^{t\cdot s} G(t) & = \id_k \otimes \ket{s}\bra{s}\ .
    \end{align}
    Then,
    \begin{align}
        i^{\beta + t\cdot r}F(t) & = i^\beta P_a \otimes Z(t) X(r)\ .
    \end{align}
    It follows that
    \begin{align}
        \frac{1}{2^{n-k}} \sum_t (-1)^{t\cdot s} i^{\beta + t\cdot r} F(t) & = i^\beta P_a \otimes \ket{s}\bra{s+r}\ .
    \end{align}
    So, we can rewrite the Kraus operators as
    \begin{align}
        K_{s,b} & = \frac{1}{2}\left( i^\beta P_a\otimes \ket{s}\bra{s+r} + (-1)^b \id_k \otimes \ket{s}\bra{s} \right)\ .
    \end{align}
    In the two dimensional subspace spanned by $\{\ket{s},\ket{s+r}\}$, the POVM aspects $E_{s,b} := K_{s,b}^\dagger K_{s,b}$ take the block form
    \begin{align}
        E_{s,b} & = \frac{1}{4}\begin{pmatrix}
          \id_k & i^\beta (-1)^b P_a\\
          i^{-\beta} (-1)^b P_a & \id_k
        \end{pmatrix}\ .
    \end{align}
    One can verify that $E_{s,b}^2 = \frac{1}{2}E_{s,b}$. That is, $2E_{s,b}$ is a projector. Further, $E_{s,b}=E_{s+r,b\oplus \beta}$ and pairing these together into $F_{s,b} := E_{s,b}+E_{s+r,b\oplus\beta}$ yields orthogonal projectors such that measuring $\{F_{s,b}\}_{s,b}$ simulates the original POVM.

    To concretely implement this measurement, consider the steps in Algorithm~\ref{alg:shift_simple}. Fix $i$ such that $r_i = 1$. Then, Heisenberg evolve the dual vector $\bra{b}_i\bra{s'}$, where $s'$ is the same as $s$ but after omitting $s_i$. The Hadamard and potential phase gate map this to
    \begin{align}
        \bra{b}_i\bra{s'} \mapsto \frac{1}{\sqrt{2}}(\bra{0}_i + (-1)^{b}i^\beta \bra{1}_i) \bra{s'}\ .
    \end{align}
    Then, the controlled $X$ takes this to
    \begin{align}
        \frac{1}{\sqrt{2}}(\bra{0}_i\bra{s'} + (-1)^{b_i}\bra{1}\bra{s'+r'})\ .
    \end{align}
    Now, after measuring $P_a$ on the first $k$ qubits, they are in some eigenstate with eigenvalue $\varphi$. In this subspace, $E_{s,b}$ has off-diagonal terms $i^\beta (-1)^b \varphi \ket{0_i s'}\bra{1_i(s'+r')}$ plus the Hermitian conjugate. Note that $(-1)^b \varphi = (-1)^{b+(1-\varphi)/2}$. Hence, in this subspace $\id \otimes \ket{b}_i\ket{s'}\bra{b}_i\bra{s'}$ Heisenberg evolves to $E_{0_is',b \oplus (1-\varphi)/2}$. Since $E_{0_is'b}$ and $E_{1(s'+r'),b\oplus \beta}$ are the same, randomizing over $0_is'$ and $1_i (s'+r')$ yields the exact same measurement statistics.

    When $r=0$, the Kraus operators and POVM aspects simplify to
    \begin{align}
        K_{s,b} & = \frac{1}{2}\left( i^\beta P_a +(-1)^b \id_k \right) \otimes \ket{s}\bra{s}\ ,\\
        E_{s,b} & = \frac{1}{2}\left(\id_k + (-1)^b (1-\beta) P_a \right)\otimes \ket{s}\bra{s}\ .
    \end{align}
    When $\beta = 1$, the algorithm  measures the last $m$ qubits in the computational basis and outputs a random $b \in \{0,1\}$. When $\beta = 0$, the algorithm measures $P_a$ on the first $k$ qubits and the last $m$ qubits in the computational basis. The output is then the sign of the measurement of~$P_a$.
\end{proof}

\subsection{The final tester}\label{sec:stab_test_full}

The full tester repeats partial Bell sampling followed by the hidden-shift
graph test.  Given samples $(s_j,b_j)$, checking whether they lie in the graph
of an affine function means checking whether there exist $u\in\F_2^m$ and
$c\in\F_2$ such that $b_j=u\cdot s_j\oplus c$ for every $j$.

\begin{algorithm}[H]
\caption{Stabilizer testing with $k$-qubit memory}
\label{alg:tester}
\begin{spacing}{1.05}
\KwInput{Copies of an unknown $n$-qubit pure state $\ket\psi$; parameters $M_{\rm out}$ and $N$.}
\KwOutput{accept or reject.}
\For{$i=1$ \KwTo $M_{\rm out}$}{
    Draw a random Clifford $C$ and apply it to each copy used in this round\;
    Run Algorithm~\ref{alg:partial-bell} on $C\ket\psi$ to obtain $(a,r)$\;
    \For{$j=1$ \KwTo $N$}{
        Run Algorithm~\ref{alg:hidden-shift} on $C\ket\psi$ with prefix $(a,r)$ and $\beta=0$, obtaining $(s_j,b_j)$\;
        Run Algorithm~\ref{alg:hidden-shift} on $C\ket\psi$ with prefix $(a,r)$ and $\beta=1$, obtaining $(s'_j,b'_j)$\;
    }
    \If{neither $\{(s_j,b_j)\}_{j=1}^N$ nor $\{(s'_j,b'_j)\}_{j=1}^N$ lies in an affine graph}{
        \Reject\;
    }
}
\Accept\;
\end{spacing}
\end{algorithm}
\paragraph{Checking affine consistency.}
Given samples
$$
    (s_1,b_1),\ldots,(s_N,b_N)\in\mathbb F_2^m\times\mathbb F_2,
$$
checking whether they lie in the graph of an affine function is equivalent to
checking whether there exist $u\in\mathbb F_2^m$ and $c\in\mathbb F_2$ such
that
$$
    b_j=u\cdot s_j\oplus c
    \qquad\text{for all }j\in[N].
$$
As the samples arrive,
one may maintain a linear basis for the observed $s_j$'s.  When a new
$s_j$ is linearly independent of the previous ones, it imposes a new
constraint on the unknown affine function.  When $s_j$ lies in the span of
previous samples, its value $b_j$ is forced by the previously observed
constraints; if the forced value disagrees with $b_j$, then the samples are
not contained in any affine graph. 
Equivalently, the samples lie in the graph of an affine function if and only if there is a solution to the linear system
\begin{align}
    \begin{pmatrix}
        b_1 \\
        b_2 \\
        \vdots \\
        b_N
    \end{pmatrix} = \begin{pmatrix}
        s_1 & 1 \\
        s_2 & 1\\
        \vdots & \vdots \\
        s_N & 1
    \end{pmatrix}\begin{pmatrix} u \\ c
    \end{pmatrix}\ .
\end{align}
Hence, Gaussian elimination can be used to check for the existence of $(u,c)$.

\subsection{Conditional soundness of the graph test}

In the case that $\ket{\psi}$ is a stabilizer state, for exactly one setting of $\beta$, the output of Algorithm~\ref{alg:hidden-shift} is contained in the graph of some affine function. When $\ket{\psi}$ is not a stabilizer state, we would like to know how likely it is that the samples still lie in such a graph. If there exists a $t^* \in \F_2^{n-k}$ such that $P_{a_\X r, a_\Z t}$ stabilizes $\ket{\psi}$, then the same analysis as for stabilizer states shows that one setting will always lie in an affine graph. Hence, we will instead consider the case where there is no such $t^*$. It turns out that we will be able to control the probabilities by bounding
\begin{align}
    A(a,r) := \max_{t\in \F_2^{n-k}} \vert \langle \psi \vert P_{a_\X r, a_\Z t}\vert \psi\rangle \vert\ ,
\end{align}
which captures the closest that all completions of $(a,r)$ into a full Pauli string can come to stabilizing $\ket{\psi}$ (ignoring phases). The following Lemma formalizes this.

\begin{lemma}
\label{lem:hidden-shift-soundness}
Fix $(a,r)\in\F_2^{2k}\times\F_2^m$ and $\beta\in\F$.  Let
$P^\beta_{a,r}$ be the output distribution of Algorithm~\ref{alg:hidden-shift}
on one copy of $\ket\psi$.  Define
$$
    A(a,r):=\max_{t\in\F_2^m}
    \left|\bra\psi P_{a_\X r, a_\Z t}\ket\psi\right|.
$$
Let $(s_1,b_1),\ldots,(s_N,b_N)$ be $N$ independent samples from
$P^\beta_{a,r}$.  Then
$$
    \Pr\Big[\exists u\in\F_2^m,\ c\in\F_2
    \text{ such that } b_j=u\cdot s_j\oplus c\text{ for all }j\in[N]\Big]
    \le 2^{m+1}\left(\frac{1+A(a,r)}2\right)^N.
$$
\end{lemma}

\begin{proof}
We first record the following observation: let $P$ be any distribution on
$\F_2^m\times\F_2$, and define its bias function by
$d_P(s)=P(s,0)-P(s,1)$.  For $u\in\F_2^m$, let
$$
    \widehat d_P(u):=\sum_{s\in\F_2^m}(-1)^{u\cdot s}d_P(s)\ ,
$$
which is the Fourier transform of the bias function. For the affine function $\chi_{u,c}(s)=u\cdot s\oplus c$,
$$
\begin{aligned}
    \Pr_{(s,b)\sim P}[b=\chi_{u,c}(s)]
    &=\E_{(s,b)\sim P}\left[\frac{1+(-1)^{b+\chi_{u,c}(s)}}2\right]=\frac12+\frac{(-1)^c}{2}\widehat d_P(u).
\end{aligned}
$$
Thus, if $|\widehat d_P(u)|\le A$ for every $u$, every affine graph has
$P$-mass at most $(1+A)/2$.

We apply this to $P^\beta_{a,r}$. The Kraus operator corresponding to outcome
$(s,b)$ from Algorithm~\ref{alg:hidden-shift}~is
$$
    K^\beta_{s,b}=\frac{1}{2^{m+1}}
    \sum_{t\in\F_2^m}(-1)^{t\cdot s}
    \left(i^{\beta+t\cdot r}F(t)+(-1)^bG(t)\right).
$$
Hence $P^\beta_{a,r}(s,b)=\|K^\beta_{s,b}\ket\psi\|^2$.  Let
$d_\beta(s)=P^\beta_{a,r}(s,0)-P^\beta_{a,r}(s,1)$.  We use the following Pauli multiplication identities.
\begin{align}
    G(t)G(v) & = G(t+v)\ , \\
\end{align}
and, with our Hermitian Pauli convention,
$$
    i^{(t-v)\cdot r}F(t)F(v)
    = G(t+v),
$$
while
$$
    i^{t\cdot r}F(t)G(v)
    = (-1)^{v\cdot r}i^{(t + v)\cdot r}F(t +v)
$$
These identities are obtained by commuting the $X(r)$ part past the
corresponding $Z$-labels on the last $m$ qubits. They are the only Pauli
algebra needed in the bias calculation. Now, expanding the difference
and using the Pauli multiplication relations for $F(t)$ and $G(t)$ gives
$$
    d_\beta(s)=\frac{1}{2^{m+1}}
    \sum_{t\in\F_2^m}(-1)^{t\cdot s}
    i^{\beta+t\cdot r}\bigl(1+(-1)^{\beta+t\cdot r}\bigr)
    \bra\psi F(t)\ket\psi.
$$
The factor $1+(-1)^{\beta+t\cdot r}$ vanishes unless $t\cdot r=\beta$.  Thus
$$
    d_\beta(s)=\frac{1}{2^m}
    \sum_{\substack{t\in\F_2^m\\ t\cdot r=\beta}}
    (-1)^{t\cdot s+(\beta+t\cdot r)/2}\bra\psi F(t)\ket\psi.
$$
Taking the Fourier transform and using orthogonality of characters gives

\begin{align}
    \widehat d_\beta(s) & = \begin{cases}
        (-1)^{(\beta + t\cdot r)/2} \langle \psi \vert F(t)\vert \psi \rangle & \beta \equiv t\cdot r\pmod{2}\\
        0 & \text{otherwise}\ .
    \end{cases}
\end{align}

Therefore $|\widehat d_\beta(u)|\le A(a,r)$ for every $u\in\F_2^m$.  Every
affine graph has mass at most $(1+A(a,r))/2$, so the probability that $N$
independent samples lie in any fixed affine graph is at most
$((1+A(a,r))/2)^N$.  There are $2^{m+1}$ affine functions
$\F_2^m\to\F_2$.  A union bound proves the lemma.
\end{proof}

\begin{corollary}
\label{cor:two-phase-hidden-shift-soundness}
Fix $(a,r)\in\F_2^{2k}\times\F_2^m$.  For each $\beta\in\F$, draw $N$
independent samples from Algorithm~\ref{alg:hidden-shift}.  The probability
that at least one of the two phase choices produces samples lying in an affine
graph is at most
$$
    2^{m+2}\left(\frac{1+A(a,r)}2\right)^N.
$$
\end{corollary}

\begin{proof}
Apply Lemma~\ref{lem:hidden-shift-soundness} to $\beta=0$ and $\beta=1$, and
union bound over the two choices of $\beta$.
\end{proof}

\subsection{Global soundness: finding a bad prefix}
\label{subsec:global-soundness}

We now show that if $\ket\psi$ is far from every stabilizer state, then a random Clifford  followed by one partial
Bell difference sample often produces a prefix with no large Pauli completion.  This is
the step that turns the conditional graph-test soundness into soundness of
Algorithm~\ref{alg:tester}. Define the set of large Pauli coefficients
$$
    M_0:=\left\{u\in\F_2^{2n}:2^n p_\psi(u)>\frac12\right\}.
$$
Equivalently, $M_0=\{u:|\braketbra{\psi}{P_u}{\psi}|>1/\sqrt2\}$.  The set $M_0$ is
isotropic: if $u,v\in M_0$ were noncommuting, then the anticommuting
observables $P_u$ and $P_v$ would satisfy
$$
    |\braketbra{\psi}{P_u}{\psi}|^2+|\braketbra{\psi}{P_v}{\psi}|^2\le 1\ .
$$
Indeed, if $a = \Tr[P_u \psi]$ and $b=\Tr[P_v \psi]$, then for $a^2 + b^2 > 0$, the observable
$$
    \frac{aP_u + bP_v}{\sqrt{a^2+b^2}}
$$
has square $\id$ and hence expectation value at most $1$ in absolute value. This gives $\sqrt{a^2+b^2} \leq 1$. But $M_0$ is the set of Paulis with large overlap and hence this would contradict the definition of $M_0$. 

\begin{definition}\label{def:m_good}
    Let $M$ be a subset of  $\F_2^{2n}$. We say that a prefix $(a,r) \in \F_2^{2k}\times \F_2^{n-k}$ is $M$-good if there exists $t\in \F_2^{n-k}$ such that $(a_\X r, a_\Z t)\in M$. Otherwise, $(a,r)$ is $M$-bad.
\end{definition}

Now, consider any Lagrangian subspace $M$ that contains $M_0$. If we measure a prefix $(a,r)$ that is $M$-bad, then no completion lies in $M_0$ and thus $A(a,r) \leq 1/\sqrt{2}$. By
Corollary~\ref{cor:two-phase-hidden-shift-soundness}, for an $M$-bad prefix
the probability of passing the affine-graph test is at most
$$
    2^{m+2}\left(\frac{1+2^{-1/2}}2\right)^N.
$$
Taking $N=c_1m$ for a sufficiently large universal constant $c_1$ makes this
probability at most $1/10$.  Thus, conditioned on obtaining an $M$-bad prefix,
the round rejects with probability at least $9/10$. It remains to upper bound the probability that partial Bell sampling lands in
the set of $M$-good prefixes.

Fix a subspace $H$ of $\F_2^{2k}\times \F_2^{n-k}$. We will call such an $H$ a prefix subspace. For such a subspace, we define its completion by
\begin{align}
    H \times \mathcal{Z}_m := \{(a_\X r, a_\Z t)\ \vert \ (a,r) \in H,\ t\in\F_2^m \} \leq \F_2^{2n}\ .
\end{align}
That is, $H\times \mathcal{Z}_m$ is all ways of completing a prefix in $H$ to a complete label for a Pauli string. In the reverse direction, a subspace $M\leq \F_2^{2n}$ defines a prefix subspace via
\begin{align}
    H_M & = \{ (a,r) \in \F_2^{2k}\times \F_2^{m}\ \vert \ \exists t\in \F_2^m: (a_\X r, a_\Z t) \in M \}\ .
\end{align}

The following lemma was proved in the Technical Toolkit, but we restate it here for completeness.

\begin{lemma}\label{lem:bell_subspace}
    Let $H \leq \F_2^{2k} \times \F_2^{m}$ be a prefix subspace and $Q_\psi^k$ the distribution for partial Bell difference sampling. Then,
    \begin{align}
        Q(H) & = \vert H \vert 2^m \sum_{u\in(H\times \mathcal{Z}_m)^\omega}p^2_\psi(u)\ .
    \end{align}
\end{lemma}
\begin{corollary}\label{cor:prefix_lagrangian}
    Let $M\leq \F_2^{2n}$ be a Lagrangian subspace. Fixing a $k$, let $H_M$ be the corresponding prefix subspace. Then,
    \begin{align}
        Q^k_\psi(H_M) & = \vert H \vert 2^m \sum_{u \in M\cap E_k}p_\psi^2(u)\ ,
    \end{align}
    where
    \begin{align}
        E_k := \{(x,z) \in \F_2^{2n}\ \vert \ x_{>k}= 0\}\ .
    \end{align}
\end{corollary}
\begin{proof}
    Since $M \subseteq H_M \times \mathcal{Z}_m$ and $M$ is self-dual, it follows that $u \in (H_M\times \mathcal{Z}_m)^\omega$ implies that $u\in M$ as well. Now, say that $u=(a_\X r, a_\Z t)$. Then, $(a_X r, a_Z s)\in H_M\times \mathcal{Z}_m$ for any $s \in \F_2^m$. By~assumption,
    \begin{align}
        0 & = [a,a] + r\cdot (t+s) \ ,
    \end{align}
    for all $s$. Hence, it must be that $r=0$ as well. 
    
    Conversely, take $u = (a_\X 0, a_\Z t) \in M\cap E_k$. For any $v = (b_X r', b_\Z t') \in H_M \times \mathcal{Z}_m$ there is some $s\in\F_2^m$ such that $v\oplus(00,0s) \in M$. We then have that
    \begin{align}
        0 & = [u,v+(00,0s)]\\
        & = [u,v] + 0\cdot (t'+s) = [u,v]\ .
    \end{align}
    
    Hence, we have shown that $(H_m \times \mathcal{Z}_m)^\omega = M\cap E_k$. The corollary then follows from Lemma~\ref{lem:bell_subspace}.
\end{proof}

Recall that the protocol works by first applying a uniformly random clifford $C$ the $\ket{\psi}$. This maps the Lagrangian subspace $M$ to some other subspace $C(M)$. Hence, we would like to analyze $\E_C[Q_{C\psi C^\dagger}^k[H_{C(M)}]]$, the probability of drawing a prefix of $M$ averaged over applying random Clifford unitaries to $\ket{\psi}$. Equivalently, we can keep $M$ and $\psi$ fixed and replace the coordinate subspace $E$ with $C^{-1}E$ since $p_{C\psi C^\dagger}(u) = p_\psi(C^{-1}u)$. That is,
$$
\begin{aligned}
    \sum_{u\in C(M)\cap E} p_{C\psi}(u)^2
    &=
    \sum_{u\in C(M)\cap E} p_\psi(C^{-1}u)^2=
    \sum_{y\in M\cap C^{-1}E} p_\psi(y)^2 .
\end{aligned}
$$
Thus, the relevant random subspace is $W_C := M \cap C^{-1}E$.

\begin{claim}
\label{claim:random-intersection-symmetry}
Let $M\le \mathbb F_2^{2n}$ be a fixed Lagrangian subspace, $C$ be a random Clifford, and~let
$
    W_C:=M\cap C^{-1}E.
$
Conditioned on $r=\dim W_C$, the subspace $W_C$ is uniformly distributed
among the $r$-dimensional subspaces of $M$.  Consequently, for every
nonzero $y\in M$,
$$
    \Pr\big[y\in W_C\mid \dim W_C=r \big]
    =
    \frac{2^r-1}{2^n-1}.
$$
\end{claim}

\begin{proof}
We only use the induced symplectic action of the Clifford group on
$\mathbb F_2^{2n}$.  Let $\operatorname{Stab}(M)$ denote the subgroup of
the symplectic group that preserves $M$. First, the distribution of $W_C$ is invariant under the action of
$\operatorname{Stab}(M)$.  Indeed, for any $g\in \operatorname{Stab}(M)$,
the Clifford $Cg$ is distributed exactly as $C$, and
$$
    W_{Cg}
    =
    M\cap (Cg)^{-1}E
    =
    M\cap g^{-1}C^{-1}E.
$$
Since $gM=M$, this equals
$$
    g^{-1}(M\cap C^{-1}E)
    =
    g^{-1}W_C.
$$
Thus $W_C$ and $g^{-1}W_C$ have the same distribution. Second, $\operatorname{Stab}(M)$ acts as the full general linear group on
$M$.  To see this, choose a symplectic basis
$$
    e_1,\ldots,e_n,f_1,\ldots,f_n
$$
with
$
    M=\operatorname{span}\{e_1,\ldots,e_n\}.
$
Every $A\in \operatorname{GL}(M)$ extends to a symplectic map by sending
$$
    e\mapsto Ae,
    \qquad
    f\mapsto (A^{-1})^{T}f.
$$
Hence the stabilizer of $M$ acts transitively on the
$r$-dimensional subspaces of $M$. Combining these two observations, conditioned on the event
$\dim W_C=r$, the distribution of $W_C$ is invariant under a transitive
action on the set of $r$-dimensional subspaces of $M$.  Therefore it is
uniform on that set. Finally, a $r$-dimensional subspace of $M$ contains $2^r-1$ nonzero
vectors, while $M$ contains $2^n-1$ nonzero vectors.  By uniformity, any
fixed nonzero $y\in M$ is included with probability $\frac{2^r-1}{2^n-1}$.
\end{proof}
\begin{proposition}
\label{prop:clifford-average}
Let $M\le\mathbb F_2^{2n}$ be a Lagrangian subspace and $H_M$ be its
prefix subspace.  Let $C$ be a random Clifford, and let $Q^k_{C\psi}$
denote the partial Bell difference distribution for the state $C\ket{\psi}$.  Then
$$
    \mathbb E_C[Q_C(H_{C(M)})]
    \le
    \frac{(2^{n+k}-2^n)\cdot \Fe_{\textsf{stab}}(\psi)+2^n-2^k}{2^k(2^n-1)}.
$$
\end{proposition}

\begin{proof}
Using Corollary~\ref{corr:partial_bell_diff} and that $p_{C\psi}(u) = p_\psi(C^{-1}u)$, we have that
\begin{align}
    Q_{C\psi}^k(H_{C(M)}) & = \frac{2^{2n}}{\vert W_C\vert} \sum_{y \in W_C}p_\psi^2(y)\ ,
\end{align}
where
$ W_C:=M\cap C^{-1}E.$ Now, if
$    r:=\dim W_C$, we have $|W_C|=2^r$. By
Claim~\ref{claim:random-intersection-symmetry}, observe that $W_C$ is uniformly
distributed among the $r$-dimensional subspaces of $M$.  Therefore
\begin{align}
    \mathbb E_C\!\left[
        Q_C(H_{C(M)})\mid \dim W_C=r
    \right]
    &=
    \frac{2^{2n}}{2^r}
    \left(
        p_\psi^2(0)
        +
        \frac{2^r-1}{2^n-1}
        \sum_{y\in M\setminus\{0\}}p_\psi^2(y)
    \right)\label{eq:exp_prefix_maxx}
\end{align}
Note that $p_\psi(0) = 2^{-n}$ always. Now, using Fact~\ref{fact:lower_bound_stabilizer_fidelity_pPsi_lagrangian_subspace} and that for every Pauli label $y$, we have $p_\psi(y)\le 2^{-n}$
$$
    \sum_{y\in M}p_\psi(y)\le \Fe_{\Sh}
\implies
    \sum_{y\in M\setminus\{0\}}p_\psi(y)^2
    \le
    2^{-n}\left(\Fe_{\Sh}-2^{-n}\right).
$$
Substituting this bound into Eq.~\eqref{eq:exp_prefix_maxx} gives
$$
\begin{aligned}
    \mathbb E_C\!\left[
        Q_C(H_{C(M)})\mid \dim W_C=r
    \right]
    &\le
    2^{2n-r}
    \left(
        2^{-2n}
        +
        \frac{2^r-1}{2^n-1}
        2^{-n}(\Fe_{\Sh}-2^{-n})
    \right) \\
    &=
    \frac{(2^{n+r}-2^n)\Fe_{\Sh}+2^n-2^r}
         {2^r(2^n-1)}.
\end{aligned}
$$
Now we make two observations: $(i)$ as $r$ increases, the RHS decreases and $(ii)$ since
$ W_C=M\cap C^{-1}E$ and $\dim M=n$, while $\dim E=n+k$, we have
$$
    \dim W_C
    \ge
    \dim M+\dim E-2n
    =
    k.
$$
So the largest value is attained at $r=k$ and we have that
$$
    \mathbb E_C[Q_C(H_{C(M)})]
    \le
    \frac{(2^{n+k}-2^n)\Fe_{\Sh}+2^n-2^k}{2^k(2^n-1)},
$$
proving the statement.
\end{proof}
\begin{corollary}
\label{cor:bad-prefix-probability}
Assume $\Fe_{\textsf{stab}}(\psi)\le 1-\varepsilon$ and $k\geq 1$.  Then, in a random-Clifford
round, the probability that partial Bell sampling outputs an $M$-bad prefix is
at least $\varepsilon/2$.
\end{corollary}

\begin{proof}
The $M$-good prefixes are exactly the prefix projection $H_M$.  Therefore
the probability of an $M$-good prefix in a random-Clifford round is
$    \mathbb E_C [Q_C(H_{C(M)})]$. 
Thus the probability of an $M$-bad prefix is
$    1-\mathbb E_C [Q_C(H_{C(M)})]$. 
By Proposition~\ref{prop:clifford-average},
$$
    1-\mathbb E_C [Q_C(H_{C(M)})]
    \ge
    (1-\Fe_{\textsf{stab}}(\psi))
    \frac{2^n(2^k-1)}{2^k(2^n-1)}.
$$
If $\Fe_{\textsf{stab}}(\psi)\le 1-\varepsilon$ and $k\ge 1$, then
the RHS of the expression above is at least $\varepsilon/2$.
\end{proof}

In the case that $k=0$, the bound from Proposition~\ref{prop:clifford-average} is trivial. However, this is readily fixed by conditioning upon $\dim W_C \geq 1$.

\begin{corollary}
    Assume $\Fe_{\textsf{stab}}(\psi)\le 1-\varepsilon$ and $k= 0$.  Then, in a random-Clifford
round, the probability that partial Bell sampling outputs an $M$-bad prefix is
at least $\eps/6$.
\end{corollary}
\begin{proof}
    Condition upon $\dim W_C \geq 1$. In this case, note that the probability that $\dim W_C \geq 1$ is exactly the probability that a random Lagrangian subspace of $\F_2^{2n}$ has an overlap with $\mathcal{Z}$, the all Z's Lagrangian, of dimension at least $1$. From Fact~\ref{fact:random_lagrangian}, this occurs with probability at least $1/3$ for a random Lagrangian subspace. Hence, the probability of observing a bad prefix is at least $1/3$ times the probability of obtaining a bad prefix assuming that $\dim W_c \geq 1$. By the analysis above, this is at least $\eps/6$. 
\end{proof}

\paragraph{Completing soundness.}
Choose $N=c_1m$, where $c_1$ is a sufficiently large constant so
that
$$
    2^{m+2}
    \left(\frac{1+2^{-1/2}}2\right)^N
    \le \frac1{10}.
$$
If the prefix is $M$-bad, then $A(a,r)\le 2^{-1/2}$, and by
Corollary~\ref{cor:two-phase-hidden-shift-soundness}, the two-phase hidden-shift graph test
rejects with probability at least $9/10$. By Corollary~\ref{cor:bad-prefix-probability}, when
$\Fe_{\textsf{stab}}(\psi)\le 1-\varepsilon$, a random-Clifford round produces
an $M$-bad prefix with probability at least $\eps/2$ ($\eps/6$ for $k=0$). Therefore a single round
rejects with probability $\Omega(\varepsilon)$.  In
particular, repeating $O(1/\varepsilon)$ independent rounds is more than
sufficient to obtain constant soundness error.  Since each round uses
$4+2N=O(m)=O(n-k)$ copies, the total sample complexity is
$$
    O\left(\frac{n-k}{\varepsilon}\right).
$$
This proves Theorem~\ref{thm:upper}.

\section{Testing lower bounds}
The goal of this section is to prove the lower bound in our main testing theorem, which we restate here for convenience.
\testingupper*

\subsection{Information Theoretic Lower Bounds}
We will follow the standard approach of considering a one versus many hypothesis test. That is, given an unknown quantum state $\rho$, decide if $\rho=\sigma$ for some fixed state $\sigma$ or if $\rho \in \mathcal{C}$ for some set of states $\mathcal{C}$. From here, we can apply a version of Le Cam's lemma specialized to decision problems.

\begin{lemma}[Le Cam's Lemma for Decision Problems]\label{lem:lecam}
    Let $\mathcal{P}_1$ and $\mathcal{P}_2$ be two families of probability distributions on some measure space. Say there exist distributions $\mu_1$ and $\mu_2$ in the convex hulls of $\mathcal{P}_1$ and $\mathcal{P}_2$ respectively such that $\TV(\mu_1,\mu_2) < 1/3$. Let $P$ be a distribution promised to be either in $P_1$ or $P_2$. From a single sample, there is no algorithm which can correctly identify if $P \in \mathcal{P}_1$ or $P \in \mathcal{P}_2$ with probability at least $2/3$ for all possible distributions $P$.
\end{lemma}

Now, let $\mathcal{C}_1$ and $\mathcal{C}_2$ be two sets of quantum states. Given access to an unknown state $\rho$, a testing algorithm $\mathcal{A}$ performs some series of measurements and then uses these outcomes to return a decision: $\rho \in \mathcal{C}_1$ or $\rho \in \mathcal{C}_2$. Thus, every possible state $\rho$ induces some distribution $P_\rho^{\mathcal{A}}$ over measurement transcripts $\vec{x}$. It is to these collections of distributions, $\mathcal{P}_1^{\mathcal{A}} = \{P_\rho^{\mathcal{A}}\}_{\rho \in \mathcal{C}_1}$ and $\mathcal{P}_2^{\mathcal{A}} = \{P_\rho^{\mathcal{A}}\}_{\rho \in \mathcal{C}_2}$, that we apply Lemma~\ref{lem:lecam}. Here $\mu_2$ will be generated by taking some distribution $D$ over the set of states $\mathcal{C}_2$ and taking $\mu_2 := \sum_{\rho \in \mathcal{C}_2} D(\rho) P_\rho^{\mathcal{A}}$.

Let $\mu_1$ and $\mu_2$ be two distributions in the convex hulls of $\mathcal{P}_1^{\mathcal{A}}$ and $\mathcal{P}_2^{\mathcal{A}}$, respectively. It turns out that it suffices to consider the likelihood ratios defined by $L(\vec{x})  = \frac{\mu_2(\vec{x})}{\mu_1(\vec{x})}$. In particular, if $L(\vec{x}) \geq 1-\delta$ for all $\vec{x}$, then the total variational distance is at most $\delta$. This was used, for example, to prove exponential lower bounds on purity testing with single-copy measurements~\cite{chen2022exponential}. 
However, it may not be true that the likelihood ratios are always lower bounded by $1-\delta$. Fortunately, this one-sided bound does not need to hold for all $\vec{x}$, only most $\vec{x}$.

\begin{lemma}[{\cite[Lemma 6]{chen2024optimal}}]\label{lem:like_ratio_prob}
    Let $\mu_1$ and $\mu_2$ be two distributions. If 
    \begin{align}
        \Pr_{\vec{x} \sim \mu_1}[L(\vec{x}) \geq 1-\delta] \geq 1-\beta\ ,
    \end{align}
    then it holds that $\TV(\mu_1,\mu_2) \leq \delta + \beta$.
\end{lemma}
\begin{proof}
    Let $E$ denote the event that $L(\vec{x})=\frac{\mu_2(\vec{x})}{\mu_1(\vec{x})} \geq 1-\delta$.
    \begin{align}
        \TV(\mu_1,\mu_2) & = \sum_{\vec{x}: \mu_1(\vec{x}) \geq \mu_2(\vec{x})} \mu_1(\vec{x}) - \mu_2(\vec{x})\\
        & = \sum_{\vec{x}\in E, \mu_1(\vec{x}) \geq \mu_2(\vec{x})}\mu_1(\vec{x}) - \mu_2(\vec{x}) + \sum_{\vec{x}\not\in E, \mu_1(\vec{x}) \geq \mu_2(\vec{x})}\mu_1(\vec{x}) - \mu_2(\vec{x})\\
        & \leq \sum_{\vec{x}\in E, \mu_1(\vec{x}) \geq \mu_2(\vec{x})}\delta \mu_1(\vec{x}) + \sum_{\vec{x}\not\in E, \mu_1(\vec{x}) \geq \mu_2(\vec{x})}\mu_1(\vec{x})\\
        & \leq \delta + \beta\ ,
    \end{align}
    proving the lemma
\end{proof}

In particular, it suffices to bound the deviations of $L(\vec{x})$ around $1$. Piecing everything together, we obtain the following lemma.

\begin{lemma}\label{lem:le_cam_exp}
    Let $\mathcal{P}_1$ and $\mathcal{P}_2$ be two sets of distributions. Take distribution $\mu_1$  and $\mu_2$ in the convex hull of $\mathcal{P}_1$ and $\mathcal{P}_2$ respectively. If 
    \begin{align}
        \E_{\vec{x}\sim \mu_1}\left[\vert  L(\vec{x})-1\vert \right] \leq \alpha\ ,
    \end{align}
    then $\TV(\mu_1,\mu_2) \leq 2\sqrt{\alpha}$. If $\alpha \leq 1/36$, then no algorithm can distinguish between distributions in $\mathcal{P}_1$ and $\mathcal{P}_2$ with probability at least $2/3$.
\end{lemma}
\begin{proof}
    By Le Cam's method (Lemma~\ref{lem:lecam}), it suffices to prove that $\TV(\mu_1,\mu_2) < 1/3$. Suppose we have that $\E_{\vec{x}\sim \mu_1}\left[\vert  L(\vec{x})-1\vert \right] \leq \alpha$, 
    then by Markov's inequality,
    \begin{align}
        \Pr_{\vec{x}\sim \mu_1}[L(\vec{x})\geq 1-t]\geq 1-{\alpha}/{t} .
    \end{align}
    Using Lemma~\ref{lem:like_ratio_prob}, this implies that $\TV(\mu_1,\mu_2) \leq t + {\alpha}/{t}$. Fixing $\alpha$, the smallest upper bound is given by taking $t := \sqrt{\alpha}$, for which this bound becomes $\TV(\mu_1,\mu_2) \leq 2\sqrt{\alpha}\leq 1/3$ for $\alpha < 1/36$.
\end{proof}

In our lower bounds, $\mathcal{C}_1$ will always be fixed to be the maximally mixed state $\id/{2^n}$. Then, $\mathcal{P}_1^{\mathcal{A}} = \{P_{\textsf{mm}}^{\mathcal{A}}\}$, where $P_{\textsf{mm}}^{\mathcal{A}}$ is the distribution induced by a POVM $\{E_{\vec{x}}\}_{\vec{x}}$ and $t$ copies of the maximally mixed state:
\begin{align}
    P_{\textsf{mm}}^{\mathcal{A}}(\vec{x}) & = \frac{\Tr[E_{\vec{x}}]}{2^{nt}}\ .
\end{align}
In this case, the likelihood ratios can always be written as 
\begin{align}
    L(\vec{x}) & = 2^{nt}\frac{\mu_2(\vec{x})}{\Tr[E_{\vec{x}}]}\ .
\end{align}
{The general proof outline from now onwards is as follows:
\begin{enumerate}
    \item We construct our hard ensemble from random degree $2$ phase states, which were introduced in Section~\ref{sec:hardensemble}.
    \item We further show that one can simplify this hard ensemble to be proportional to a projector consisting of orthogonal matrices from Section~\ref{sec:simplifyhardensemble}.
    \item Next we break this orthogonal set of matrices, denoted $O_t$ into two sets $O_{t-1}$ and $M_t=O_t\backslash O_{t-1}$ and using the likelihood ratio method, show that the value for the $O\in M_t$ is at most $2^{t-n+k}$ (which is done in Section~\ref{sec:lb_tech_proof}) and conclude the proof of our theorem in~\ref{sec:lb_testing}.
\end{enumerate}}

\subsection{Construction of hard ensemble}
\label{sec:hardensemble}
 We will consider uniformly random \emph{degree 2 phase} states, which is a subset of stabilizer states (i.e., the subspace it is defined on is over $\F_2^n$ and there are not complex amplitudes). Recall that a phase state takes the~form
\begin{align}
    \ket{\psi_f}& = \frac{1}{\sqrt{2^{n}}}\sum_x (-1)^{f(x)}\ket{x}\ .
\end{align}
When $f(x)$ is a degree two polynomial, we can write it as $f(x) = x^\top A x$, with $A\in \F_2^{n \times n}$ an upper triangular matrix. Accordingly, we relabel the states as $\ket{\psi_A}$.

There are three ensembles we will consider in proving our bounds:
\begin{enumerate}
    \item Uniformly random degree two phase states $\rho_P := \E_{A}[\psi_A^{\otimes t}]$.
    \item Haar random pure states $\rho_H := \E_\psi[\psi^{\otimes t}] = \frac{1}{\dim \overset{t}{\vee}\mathbb{C}^{2^n}}\sum_{\pi \in S_t} \pi$.
    \item The maximally mixed state $\rho_{\textsf{\textsf{mm}}} = \id^{\otimes t}/2^{nt}$.
\end{enumerate}

We will argue that distinguishing $\rho_P$ from $\rho_H$ is difficult under various settings. Further, we claim that this implies hardness of stabilizer testing. This follows from the fact that a Haar random state is far from any stabilizer state with high probability.
\begin{fact}
    Let $0 < \eps < 1-n^2/2^n$. Then, any algorithm for stabilizer testing to accuracy $\eps$ can distinguish $\rho_P$ from $\rho_H$ with probability greater than $2/3$.
\end{fact}
\begin{proof}
    This is a corollary of ~\cite[Lemma~4.3,4.4]{hinsche2025single} wherein the authors show that an algorithm for stabilizer testing would be able to distinguish a Haar random state from all stabilizer states (with high probability over the Haar random state). The fact then follows from phase states being a subset of stabilizer states.
\end{proof}
In order to prove our lower bound, instead of directly distinguishing between $\rho_P$ and $\rho_H$, we will instead argue through the intermediate step of distinguishing $\rho_P$ (or $\rho_H$) from $\rho_{\textsf{mm}}$. To this end, suppose we fix an arbitrary protocol and then the states induce classical distributions $\mathcal{P}_{P}$, $\mathcal{P}_H$, and $\mathcal{P}_{\textsf{mm}}$. Using Lemma~\ref{lem:lecam}, our hardness would be implied by arguing $\TV(\mathcal{P}_P, \mathcal{P}_H)$ is small. To argue this, first observe by  triangle inequality that
\begin{align}
    \TV(\mathcal{P}_P, \mathcal{P}_H) \leq \TV(\mathcal{P}_P, \mathcal{P}_{\textsf{mm}}) + \TV(\mathcal{P}_H, \mathcal{P}_{\textsf{mm}})\ .
\end{align}
Hence, going forward we will exclusively consider distinguishing either $\rho_P$ (or $\rho_H$) from $\rho_{\textsf{mm}}$. We will show in Theorem~\ref{thm{lb_tv}} that $\TV(\mathcal{P}_{mm}, \mathcal{P}_P)$ is small. In Therorem~\ref{thm:purity_lb} in Appendix~\ref{app:purity}, we will also show that $\TV(\mathcal{P}_{mm}, \mathcal{P}_{H})$ is small as well when $t = o(2^{n-k})$ \footnote{Prior work~\cite{chen2024optimal} proved a lower bound of $2^{O(n-k)}$ for this task under a slightly weaker model where the distinguisher must measure the memory every other round. We allow for the memory to be left coherent throughout the entire protocol and need to prove hardness in this setting.}. This then proves our lower bound.

\subsection{Rewriting our phase state ensemble}
\label{sec:simplifyhardensemble}

In section~\ref{sec:phase_states}, we showed the following: let $\ket{\psi_A}$ be a quadratic phase state, then the ensemble corresponding to the $t$-fold tensor product of these states can be written as
$$
\rho=\Exp_A[\ketbra{\psi_A}{\psi_A}^{\otimes t}]=\frac{1}{2^{nt}}\sum_{\vec{x},\vec{y}\in (\F_2^n)^t}\id\Big[X^\top X = Y^\top Y]\ketbra{\vec{x}}{\vec{y}},
$$
where $X,Y \in F_2^{t \times n}$ are obtained by stacking the vectors in $\vec{x}$ and $\vec{y}$ as rows.

Our first observation below will be that one can replace the degree-$2$ ensemble state $\rho_P$ with
\begin{align}
    \sigma_P := \frac{1}{2^{nt}}\sum_{O \in \Oc_t}R(O),
\end{align}
with a small loss in trace distance. 
Observe that $\sigma_{P}$ is unnormalized and, in particular, not equal to $\rho_P$, but is close in trace distance. Hence, by another triangle inequality, it suffices to argue that the distributions induced by measuring $\sigma_P$ are close in total variational distance to those induced by $\rho_{\textsf{mm}}$.  In this rest of this subsection, we will derive the following bound on the trace distance between $\rho_P$ and $\sigma_P$:

\begin{theorem}\label{thm:tvd_orthogonal}
    Let $\rho_P := \E_{A}[\psi_A^{\otimes t}]$ be the ensemble average of $t$ copies of a uniformly random degree two phase state and $\sigma_P = \frac{1}{2^{nt}} \sum_{O \in \Oc_t}R(O)$. Then,
    \begin{align}
        \Vert \rho_P - \sigma_P\Vert_1 \leq O(2^{t-n})\ .
    \end{align}
\end{theorem}

We now start to massage this into something that resembles $\sigma_t$. To do so, we first introduce the subspace of linearly independent sets of vectors. That is,
\begin{align}
    V_{LI} & = \text{span}\{\vec{x} \in \F_2^{nt} \ \vert \ \vec{x} \text{ is a linearly independent set}\},
\end{align}
where by linearly independent we mean $x^1,\ldots,x^t \in \mathbb{F}_2^n$ is linearly independent. Equivalently, this means that $X$, defined by stacking the vectors in $\vec{x}$ as rows, is full rank. 
 As an intermediate step, define $ \tilde{\sigma}_P := \Pi_{LI} \rho_P \Pi_{LI}\ .$ 
It turns out that this is close in trace distance to $\rho_P$:
\begin{lemma}
    \begin{align}
        \Vert \rho_P - \tilde{\sigma}_P \Vert_1 \leq 2^{2+(t+1-n)/2}\ . 
    \end{align}
\end{lemma}
\begin{proof}
    To each $\ket{\psi_A}^{\otimes t}$ define $\ket{\widehat{\psi}_{A,t}}= \Pi_{LI}\ket{\psi_A}^{\otimes t}
    $. Now, observe that
    \begin{align}
        \left(\langle \psi_A\vert^{\otimes t}\right)\vert\widehat{\psi}_{A,t}\rangle & = \frac{\dim V_{LI}}{2^{nt}}
    \end{align}
     Furthermore, we can lower bound $\dim V_{LI}$ as follows
     \begin{align}
        \dim V_{LI} & = \prod_{j=0}^{t-1}(2^n - 2^j) = 2^{nt} \prod_{j=0}^{t-1}(1-2^{j-n}) \geq 2^{nt}\left(1-2^{-n}\sum_{j=0}^{t-1}2^{-j} \right)\geq 2^{nt}\left(1 - 2^{-n+t+1}\right)
    \end{align}
    Hence, we have that
    \begin{align}
        \Vert \psi_A^{\otimes t} - \widehat{\psi}_{A,t} \Vert_1 \leq 2^{2+(t+1-n)/2}\ .
    \end{align}
    Since $\rho_P$ is a convex combination of the $\psi_A^{\otimes t}$'s, this concludes the proof.
\end{proof}

Now, we can expand $\tilde{\sigma}_P$.
\begin{align}
    \tilde{\sigma}_P & = \Pi_{VI} \rho_P \Pi_{VI}\\
    & = \frac{1}{2^{nt}} \sum_{\vec{x},\vec{y} \in V_{LI}}\id\{X^\top X = Y^\top Y\}\ket{\vec{x}}\bra{\vec{y}}\ .
\end{align}
Recall that Lemma~\ref{lem:stoch_transform} states that $X^\T X = Y^\top Y$ for full rank $X,Y \in \F_2^{t\times n}$ if and only if there is an orthogonal matrix such that $OY=X$. Since $\vec{x}, \vec{y} \in V_{LI}$, $X$ and $Y$ are full rank and this lemma applies. Further, notice that $OY\neq Y$ for any $O\neq \id$. Indeed, assume that $OY=Y$, then multiplying with the right inverse of $Y$ yields that $O = \id$. So, fixing $\vec{y}$, it suffices to sum over $O \in \Oc_t$. Notice that the map $Y\mapsto OY$ is exactly the representation $R(O)$ previously discussed in Section~\ref{sec:cliff_com}. That is, $R(O)$ acts transversally on blocks of $t$ qubits via $R(O)\ket{X} = \ket{OX}$. We can then recognize that
\begin{align}
    \tilde{\sigma}_P = \frac{1}{2^{nt}}\Pi_{LI} \sum_{O\in \Oc_t} R(O) \Pi_{LI}\ .
\end{align}

Lastly, we will remove the projectors onto $V_{LI}$. First, notice that $V_{LI}$ is an invariant subspace of $\Oc_t$ since the rank of $\vec{x}$ does not change under invertible linear transformations. Second, $\sum_{O\in\Oc_t} R(O)$ is proportional to a projector. Then, $\Pi_{LI} \sum_{O\in\Oc_t} R(O) \Pi_{LI} \preceq \sum_{O\in\Oc_t} R(O)$. Using this,
\begin{align}
    \Vert \sigma_P - \tilde{\sigma}_P \Vert_1 & = \frac{1}{2^{nt}}\Tr\left[\sum_{O\in\Oc_t}R(O) - \Pi_{LI} \sum_{O\in\Oc_t} R(O) \Pi_{LI}\right]\ .
\end{align}
Other than $\id$, $\Tr[\Pi_{LI} R(O) \Pi_{LI}] = 0$ since $OY\neq Y$ for all $O\neq \id$. Hence, $\Tr[\Pi_{LI} \sum_{O\in\Oc_t} R(O) \Pi_{LI}] = \dim V_{LI}$. The other trace we bound by first noting that
\begin{align}
    \Tr[R(O)] & = 2^{n(\dim\ker(\id+O))} = 2^{n(t-\rank(\id+O))}\ .
\end{align}
Thus,
\begin{align}
    \Vert \sigma_P - \tilde{\sigma}_P \Vert_1 & = \sum_{\substack{O\in \Oc_{t} \\ O \neq \id}}2^{-n\rank(\id+O)}+1-\frac{\dim V_{LI}}{2^{nt}} \leq 2^{t+1-n} + \sum_{\substack{O\in \Oc_{t} \\ O \neq \id}}2^{-n\rank(\id+O)}\ .
\end{align}
To finish our trace distance calculation, it suffices to count the number of $O$'s such that $\rank(\id+O)=r$ for $1 \leq r \leq t$.
\begin{lemma}\label{lem:fixed_counting}
    The number of $O \in \Oc_{t}$ such that $\rank(\id+O) = r$ is no more than $2^{rt}$.
\end{lemma}
\begin{proof}
    $O$ is completely determined by $A:=\id+O$ and vice versa. Let $W := \ker(A)$. Take any $x$ such that $Ox \neq x$ and hence $Ax \neq 0$. Let $w \in W$ be some arbitrary element of the kernel. Then,
    \begin{align}
        \langle Ax,w \rangle & = \langle (\id + O)x, w\rangle \\
        & = \langle x , w \rangle  + \langle Ox, w\rangle\\
        & = \langle x,w\rangle + \langle x, O^\top w\rangle \\
        & = 2\langle x, w\rangle \equiv 0 \ ,
    \end{align}
    where $O^\top w = w$ because $Ow = w$ implies that $O^\top w = w$ as well. Hence, the image of $A$ is contained in $W^\perp$. By counting dimensions, the image of $A$ must be exactly $W^\perp$. Then, $A$ is fully determined by two elements:
    \begin{enumerate}
        \item The kernel $W$, which is of dimension $t-r$.
        \item An invertible linear map $\overline{A}:\F_2^t / W \rightarrow W^\perp$.
    \end{enumerate}
    Hence, the number of such $A$'s is upper bounded by
    \begin{align}
        \binom{t}{t-r}_2 \left\vert \mathsf{GL}(r,2)\right\vert & = \prod_{j=0}^{r-1}(2^t-2^j) \leq  2^{rt},
    \end{align}
    proving the desired bound.
\end{proof}

We can now finish the proof of Theorem~\ref{thm:tvd_orthogonal}
\begin{align}
        \Vert \rho_P - \sigma_P\Vert_1 & \leq \Vert \rho_P - \tilde{\sigma}_P\Vert_1 + \Vert \tilde{\sigma}_P - \sigma_P\Vert_1\\
        & \leq 2^{2+(t+1-n)/2} + 2^{t+1-n} + \sum_{\substack{O \in \Oc_t \\ O \neq \id}} 2^{-n\rank(\id+O)}\\
        &\leq 2^{O(t-n)}+\sum_{r=1}^t 2^{-r(n-t)}\\
        &\leq  2^{O(t-n)}
    \end{align}
where the third inequality used Lemma~\ref{lem:fixed_counting}. This concludes the proof of the theorem

\subsection{Lower bound for testing}\label{sec:lb_testing}
Our starting point is Theorem~\ref{thm:tvd_orthogonal}. In particular, we show that the (unnormalized) distribution induced by measuring $\sigma_P$ is close in $1$-norm to that obtained by measuring the maximally mixed state. That is, we apply triangle inequality as follows:
\begin{align}
    \TV(\mathcal{P}_{mm},\mathcal{P}_P) \leq \TV(\mathcal{P}_{mm}, \mathcal{P}_{\sigma_P}) + \TV(\mathcal{P}_{\sigma_P},\mathcal{P}_P)\ .
\end{align}
Theorem~\ref{thm:tvd_orthogonal} implies that the second term on the right hand side is negligible as long as $t \ll n$ and we will focus on the first term. To bound this term, we consider the likelihood ratios
\begin{align}
    L(\vec{x}) & = \frac{\mathcal{P}_{\sigma_P}(\vec{x})}{\mathcal{P}_{mm}(\vec{x})}  = \frac{1}{\Tr[E_{\vec{x}}]} \Tr\left[\sum_{O \in \Oc_t}R(O) E_{\vec{x}} \right].
\end{align}
Ideally we would show that $L(\vec{x}) \geq 1-\delta$ for $\delta = O(1)$ when $t \ll n$ for \emph{every  aspect} $\psi_{\vec{x}}$. However, this cannot be possible. Indeed, Hinsche and Helsen showed that there exist product states which are orthogonal to $\sigma_P$ (see~\cite[Lemma 4.14]{hinsche2025single}), for which $L(\vec{x}) = 0$, which is of course as far from $1$ as possible. To circumvent this, we will instead consider \emph{average-case} bounds on likelihood ratios and use Lemma~\ref{lem:le_cam_exp}. A slight technicality is that $\mathcal{P}_{\sigma_P}$ is an unnormalized measure. We show that this is no issue in Appendix~\ref{app:tvd}. Hence, the rest of this section consists of arguing that
\begin{align}
    \E_{\vec{x}\sim P_{mm}}[ \left\vert L(\vec{x}) - 1\right\vert ] & = \frac{1}{2^{nt}}\sum_{\vec{x}} \left\vert\sum_{O \neq \id} \Tr[R(O)E_{\vec{x}}] \right\vert\ ,
\end{align}
is small when $t \ll n$. At a high level, we will do so by splitting $\Oc_t$ into $\Oc_{t-1}$ and $M_t := \Oc_t \backslash \Oc_{t-1}$, where we identify $\Oc_{t-1}$ as all matrices that fix $e_1 = (1,0,0,\ldots,0)$. We will bound the expected value of $\sum_{O\in M_t}\Tr[R(O)]E_{\vec{x}}$ and then recursively apply the same argument to $\Oc_{t-1}$ and so on. Our main technical lemma for making this argument is the following:

\begin{lemma}\label{lem:mt_bound}
    Let $t < n$. For any protocol using single-copy measurement and $k$ qubits of memory, it holds that
    \begin{align}
        \frac{1}{2^{nt}}\sum_{\vec{x}}\left\vert \sum_{O \in M_t}\Tr[R(O)E_{\vec{x}}] \right\vert & \leq 2^{1-(n-k-7t)/2}\ .
    \end{align}
\end{lemma}
This lemma is proved in Section~\ref{sec:lb_tech_proof}.

Assuming Lemma~\ref{lem:mt_bound}, we can prove the desired bound on the total variational distance.
\begin{theorem}\label{thm{lb_tv}}
    Let $t < n$, then
    \begin{align}
        \E_{\vec{x}\sim P_{mm}}[\vert L(\vec{x})-1\vert] \leq 2^{2-(n-k-7t)/2}\ .
    \end{align}
    Consequently, $t=\Omega(n-k)$ is required to distinguish between a random degree two phase state and the maximally mixed state.
\end{theorem}
\begin{proof}
    From Theorem~\ref{thm:tvd_orthogonal}, we have that
    \begin{align}
        \TV(\mathcal{P}_{mm},\mathcal{P}_P) \leq 2^{O(t-n)} +\TV(\mathcal{P}_{mm},\mathcal{P}_{\sigma_P})\ . 
    \end{align}
    We use Lemma~\ref{lem:le_cam_exp} to obtain the bound
    \begin{align}
        \TV(\mathcal{P}_{mm},\mathcal{P}_{\sigma_P}) & \leq 2\left( \E_{\vec{x}\sim P_{mm}}[\vert L(\vec{x})-1\vert] \right)^{1/2} \leq 2^{2-(n-k-7t)/4}\ ,
    \end{align}
    where we have assumed that $\E_{\vec{x}\sim P_{mm}}[\vert L(\vec{x})-1\vert] \leq 2^{2-(n-k-7t)/2}$ as claimed. To prove this claim, we split the expectation value into $O_{t-1}$ and $M_t$.

    \begin{align}
        \E_{\vec{x}\sim P_{mm}}[\vert L(\vec{x})-1\vert] & = \frac{1}{2^{nt}}\sum_{\vec{x}} \left\vert\sum_{O \neq \id} \Tr[R(O)E_{\vec{x}}] \right\vert\\
        & \leq \frac{1}{2^{nt}}\sum_{\vec{x}} \left\vert\sum_{\substack{O \in \Oc_{t-1}\\ O\neq \id}} \Tr[R(O)E_{\vec{x}}] \right\vert + \frac{1}{2^{nt}}\sum_{\vec{x}} \left\vert\sum_{O\in M_t} \Tr[R(O)E_{\vec{x}}] \right\vert\\
        & \leq \frac{1}{2^{nt}}\sum_{\vec{x}} \left\vert\sum_{\substack{O \in \Oc_{t-1}\\ O\neq \id}} \Tr[R(O)E_{\vec{x}}] \right\vert + 2^{1-(n-k-7t)/2} \ ,
    \end{align}
    where we have used Lemma~\ref{lem:mt_bound}.

    Now, we will recursively apply the same bound to the sum over $O \in \Oc_{t-1}$. To do so, we note that each outcome $x_1$ in the first round can be viewed as the start of a new protocol. Fixing $x_1$, we define
    \begin{align}
        V_{x_1} := V_0^{x_1} (\id_1 \otimes \ket{\eta_0}) = K_{0}^{x_1}(\id_1 \otimes \ket{\eta_0}): \Hc \rightarrow \Mc \ ,
    \end{align}
    where we are using the forward operator notation from Definition~\ref{def:pass_operators}. For reference, recall that these are defined as
    \begin{align}
        V_{\vec{x}_{<i}}^{\vec{x}_{i:j}}: \Hc_{i:j}\otimes \Mc\rightarrow \Mc :: \left( \bigotimes_{\ell=i}^j \ket{\phi_\ell}\otimes \ket{\eta} \right) \mapsto K_{\vec{x}_{<j}}^{x_j} \left( \ket{\phi_j}\otimes K_{\vec{x}_{< j-1}}^{x_{j-1}}\left( \cdots K_{\vec{x}_{< i}}^{x_i}(\ket{\phi_i}\otimes \ket{\eta}) \right) \right) \ ,
    \end{align}
    which represents the map from inputs $i$ to $j$ to the memory enacting by these rounds of the protocol (for a fixed transcript). Further, recall that the probability of a transcript $\vec{x}$ is given by $\Tr[E_{\vec{x}} \rho^{\otimes t}]$, where
    \begin{align}
        E_{\vec{x}} & = (\id_1 \otimes \bra{\eta_0})(V_0^{\vec{x}})^\dagger V_0^{\vec{x}}(\id_1\otimes \ket{\eta_0}) \\
        & = (V_{x_1})^\dagger (V_{x_1}^{\vec{x}_{>1}})^\dagger (V_{x_1}^{\vec{x}_{>1}}) V_{x_1}\ .
    \end{align}
    
    Now let $A_{t-1} := \sum_{O \neq \id \in \Oc_{t-1}} R(O)$, taken as an operator on $\Hc_{2:t}$. Then,
    \begin{align}
        \Tr[(\id_1 \otimes A_{t-1}) E_{\vec{x}}] & = \Tr[(\id_1 \otimes A_{t-1}) (V_{x_1})^\dagger (V_{x_1}^{\vec{x}_{>1}})^\dagger (V_{x_1}^{\vec{x}_{>1}}) V_{x_1}] \\
        & = \Tr[(V_{x_1}V_{x_1}^\dagger \otimes A_{t-1}) (V_{x_1}^{\vec{x}_{>1}})^\dagger (V_{x_1}^{\vec{x}_{>1}})]\ , \label{eq:recurse_1}
    \end{align}
    where the second equality follows from cyclicity of trace and that $V_{x_1}$ acts trivially on $\Hc_{2:t}$. Define
    \begin{align}
        \eta_{x_1} := \frac{V_{x_1}V_{x_1}^\dagger}{\Tr[V_{x_1} V_{x_1}^\dagger]} \in \mathcal{B}(\Mc)\ , \quad p(x_1) & = \frac{\Tr[V_{x_1}V_{x_1}^\dagger]}{2^n} \ .
    \end{align}
    By completeness, $\sum_{x_1} (V_0^{x^1})^\dagger V_0^{x_1} = \id_{\Hc}\otimes \id_{\Mc}$ and hence $\sum_{x_1}p(x_1) = 1$. Hence, we can recognize Eq~\eqref{eq:recurse_1} as  
    \begin{align}
        \Tr[(\id_1 \otimes A_{t-1}) E_{\vec{x}}] & = 2^n p(x_1) \Tr[(\eta_{x_1}\otimes A_{t-1}) (V_{x_1}^{\vec{x}_{>1}})^\dagger (V_{x_1}^{\vec{x}_{>1}})]\ .
    \end{align}
    By taking a spectral decomposition of $\eta_{x_1}$ and applying the triangle inequality, we can assume that each $\eta_{x_1}$ is pure. That is, let $\eta_{x_1} = \sum_j \lambda_{x_1,j} \ket{\eta_{x_1,j}}\bra{\eta_{x_1,j}}$ and
    \begin{align}
        \vert \Tr[(\id_1 \otimes A_{t-1}) E_{\vec{x}}]\vert  & = 2^n p(x_1) \vert \Tr[ (\eta_{x_1} \otimes A_{t-1}) (V_{x_1}^{\vec{x}_{>1}})^\dagger (V_{x_1}^{\vec{x}_1})] \vert\\
        & \leq 2^n p(x_1) \sum_j \lambda_{x_1,j} \vert \Tr[ ( \ket{\eta_{x_1,j}}\bra{\eta_{x_1,j}} \otimes A_{t-1}) (V_{x_1}^{\vec{x}_{>1}})^\dagger (V_{x_1}^{\vec{x}_1})] \ .
    \end{align}
    Hence, we will assume that $\eta_{x_1}$ is some pure state. Further, we will view $\eta_{x_1}$ as the starting state of the memory of a new protocol using $t-1$ copies of some input $\rho$. Varying $\vec{x}_{>1}$ for some fixed $x_1$, the forward operators $\{ V_{x_1}^{\vec{x}_{>1}}\}_{\vec{x}_{>1} \vert x_1}$ form a complete POVM. To make this explicit, we can define the new set of POVM aspects
    \begin{align}
        E_{\vec{x}_{>1} \vert x_1} : = (\id_2 \otimes \bra{\eta_{x_1}}) (V_{x_1}^{\vec{x}_{>1}})^\dagger (V_{x_1}^{\vec{x}_{>1}})(\id_2 \otimes \ket{\eta_{x_1}})\ .
    \end{align}
    Hence, 
    \begin{align}
        \frac{1}{2^{nt}}\sum_{\vec{x}} \vert\Tr[(\id_1\otimes A_t) E_{\vec{x}}] \vert & = \frac{1}{2^{nt}}\sum_{x_1} \sum_{\vec{x}_{>1}} \vert\Tr[(\id_1\otimes A_t) E_{\vec{x}}]\\
        & = \sum_{x_1} p(x_1) \frac{1}{2^{n(t-1)}}\sum_{\vec{x}_{>1}} \vert \Tr[A_{t-1} E_{\vec{x}_{>1} \vert x_1} ] \vert\ .
    \end{align}
    Thus, we have split the expectation value into a convex combination of protocols using $t-1$ copies. We can now split $\Oc_{t-1}$ into $\Oc_{t-2}$ and $M_{t-1}$ and again apply Lemma~\ref{lem:mt_bound}. Then $\Oc_{t-2}$ can be split into $\Oc_{t-3}$ and $M_{t-2}$. Continuing this way, we obtain a bias of $2^{1-(n-k-7t')/2}$ for each $t' \in [2,t]$ (since, in the case of $t'=1$ there is only the identity operator). Then,
    \begin{align}
        \E_{\vec{x}\sim P_{mm}}[\vert L(\vec{x})-1\vert] \leq 2^{1-(n-k)/2}\sum_{j=2}^t 2^{7j/2}< 2^{2-(n-k-7t)/2}\ .
    \end{align}
\end{proof}

\subsection{Proof of Lemma~\ref{lem:mt_bound}}
\label{sec:lb_tech_proof}

We start with a simple fact that will be crucial in our lower bounds.

\begin{fact}\label{fact:trace_norm_inequality}
    Let $H$ be a Hermitian operator on some finite dimensional Hilbert space. Also let $\{E_{\vec{x}}\}_{\vec{x}}$ be some collection of PSD operators such that $\sum_{\vec{x}} E_{\vec{x}} \preceq \id$, then
    \begin{align}
        \sum_{\vec{x}} \vert \Tr[H E_{\vec{x}}]\vert \leq \Vert H \Vert_1\ .
    \end{align}
\end{fact}
\begin{proof}
    Split $H = H_+ - H_-$ into its positive and negative parts. Then,
    \begin{align}
        \sum_{\vec{x}} \vert \Tr[H E_{\vec{x}}]\vert & \leq \sum_{\vec{x}} (\Tr[H_+ E_{\vec{x}}] + \Tr[H_- E_{\vec{x}}]) \leq \Tr[H_+] + \Tr[H_-] = \Vert H \Vert_1\ .
    \end{align}
\end{proof}

To give some intuition for the lower bound proof, consider the case of $k=0$, where the POVM consists entirely of product states. That is, $\Tr[\rho^{\otimes t} \otimes_{i=1}^t \varphi_{x_i\vert \vec{x}_{<i}}]$ for some collection of product states and inputs $\rho^{\otimes t}$. We can then fix $x_1$ and sum over all competitions of the transcript $\vec{x}_{>1}$. Since $\sum_{\vec{x}_{>1}}\varphi_{\vec{x}_{2:t}\vert x_1} =\id_{2:t}$ (by completeness of the POVM), we can use Fact~\ref{fact:trace_norm_inequality} to bound the entire sum over this branch of the learning tree $\mathcal{T}$ with
\begin{align}
    \sum_{\vec{x}_{>1}} \left\vert \Tr\left[\sum_{O \in M_t} R(O) \varphi_{x_1}\bigotimes_{i=2}^t \varphi_{x_i \vert \vec{x}_{<i}}  \right] \right\vert & \leq \Vert \Tr_1\left[ \sum_{O \in M_t} R(O)(\varphi_{x_1} \otimes \id_{2:t}) \right]\Vert_1\ .
\end{align}
Take $O = \SWAP_{1,2} \in M_t$. Then,
\begin{align}
    \Tr_1[(\SWAP_{1,2}\otimes \id_{3:t}) (\varphi_{x_1} \otimes \id_{2:t})] & = \varphi_{x_1} \otimes \id_{3:t}\ ,
\end{align}
and $\Vert \Tr_1[(\SWAP_{1,2}\otimes \id_{3:t}) (\varphi_{x_1} \otimes \id_{2:t})] \Vert_1 = \Tr[\varphi_{x_1}]2^{n(t-2)}$. That is, this contraction has reduced $\SWAP_{1,2}\otimes \id_{3:t}$ from a unitary with trace norm $2^{nt}$ to some operator with trace norm decreased by at least a factor of $2^{2n}$. Lemma~\ref{lem:mt_bound} is essentially a generalization of this reduction in trace norm to the entire collection of operators in $M_t$ by noting that these matrices have some nontrivial linear action on the first tensor copy of $\Hc$.

To continue and make this argument when $k\geq 0$, we first introduce a bit of notation. As in the prior section, we will adopt the notation
    \begin{align}
        V_{x_1} := V_0^{x_1}(\id_1 \otimes \ket{\eta_0}): \Hc\rightarrow\Mc ,
    \end{align}
which represents the operator passing from input in the first round to the memory when observing outcome $x_1$. Also, we will use $S_t := \sum_{O \in M_t}R(O)$. Lastly, we will use the forward pass operators from Definition~\ref{def:pass_operators} to further refine $S_t$. For reference, recall that these are defined as
    \begin{align}
        V_{\vec{x}_{<i}}^{\vec{x}_{i:j}}: \Hc_{i:j}\otimes \Mc\rightarrow \Mc :: \left( \bigotimes_{\ell=i}^j \ket{\phi_\ell}\otimes \ket{\eta} \right) \mapsto K_{\vec{x}_{<j}}^{x_j} \left( \ket{\phi_j}\otimes K_{\vec{x}_{< j-1}}^{x_{j-1}}\left( \cdots K_{\vec{x}_{< i}}^{x_i}(\ket{\phi_i}\otimes \ket{\eta}) \right) \right) \ ,
    \end{align}
    which represents the map from inputs $i$ to $j$ to the memory enacted by these rounds of the protocol (for a fixed transcript). For fixed $x_1$, the probability of observing a transcript $\vec{x}$ is then given by $\Tr[E_{\vec{x}} \rho^{\otimes t}]$, where $E_{\vec{x}} = (V_{x_1})^\dagger (V_{x_1}^{\vec{x}_{>1}})^\dagger (V_{x_1}^{\vec{x}_{>1}} V_{x_1})$. Using these, we refine $S_{t}$ by applying $V_{x_1}$ to map $\Hc_1$ to the memory space $\Mc$, essentially partially contracting $S_1$.
\begin{align}
    S_{t,x_1} : = (V_{x_1}\otimes \id) S_t (V_{x_1}\otimes \id)^\dagger \in \mathcal{B}\left(\Hc_{2:t}\otimes \Mc \right)\ .
\end{align}
Then, using the construction of the pass operators, it follows that
\begin{align}
    \frac{1}{2^{nt}}\sum_{\vec{x}}\left\vert \sum_{O \in M_t}\Tr[R(O)E_{\vec{x}}] \right\vert & = \frac{1}{2^{nt}}\sum_{x_1} \sum_{\vec{x}_{>1}}\Tr[S_{t,x_1} (V_{x_1}^{\vec{x}_{>1}})^\dagger (V_{x_1}^{\vec{x}_{>1}})]\\
    & \leq \frac{1}{2^{nt}} \sum_{x_1}\Vert S_{t,x_1}\Vert_1\ ,\label{eq:st_sum}
\end{align}
where the final inequality used Fact~\ref{fact:trace_norm_inequality} and the fact that $\sum_{\vec{x}_{>1}} (V_{x_1}^{\vec{x}_{>1}})^\dagger (V_{x_1}^{\vec{x}_{>1}}) = \id_{2:t}\otimes \id_{\Mc}$ by  the completeness of the protocol.

Next, we will need a few simple facts about $M_t$.
\begin{lemma}\label{lem:mt_facts}
    Let $M_t := \Oc_{t}-\Oc_{t-1}$ and $S_t := \sum_{O \in M_t} R(O)$. Then,
    \begin{enumerate}
        \item $O \in M_t$ if and only if 
        \begin{align}
            O & = \begin{pmatrix}
                O_{11} & O_{1*}\\
                O_{*1} & O_{**}
            \end{pmatrix}\ ,
        \end{align}
        where $O_{1*} \in \F_2^{1 \times (t-1)}$ and $O_{*1} \in \F_2^{(t-1)\times 1}$ are non-zero row and column vectors.
        \item $S_t^\dagger = S_t$.
        \item $\vert M_t \vert \leq 2^{t-1}\vert \Oc_{t-1}\vert$.
        \item $S_t O = O S_t = S_t$ for any $O \in \Oc_{t-1}$.
    \end{enumerate}
\end{lemma}
    \begin{proof}We prove these four items separately now.
        \begin{enumerate}
            \item Clearly $O_{*1}$ is non-zero if and only if $O \in M_t$. Note that $Oe_1 = e_1$ if and only if $e_1 = O^{-1}e_1 = O^\top e_1$. Hence, $O^\top$ is in $M_t$ if and only if $O \in M_t$ and $O^\top \in M_t$ if and only if $O_{1*}$ is non-zero.

            \item By the proof of (1) above, $M_t$ is closed under inverses.

            \item $M_t$ consists of all but one coset of $\Oc_{t-1}$ in $\Oc_t$. The cosets are determined by the image of $e_1$ under $O$. That is, $O$ and $W$ are in the same coset if and only if $Oe_1 = We_1$. Since $O$ is orthogonal, $Oe_1 \cdot Oe_1 = e_1 \cdot e_1$. Hence, the cosets correspond to all odd-weight bitstring other than $\vec{1}$ (because $O\vec{1} = \vec{1} $ for all $O$). Hence there are at most $2^{t-1}$ cosets.

            \item Take some $W \in M_t$. Since $O$ must fix $e_1$, $WOe_1 = We_1 \neq e_1$. Similarly, if $OWe_1 = e_1$, then $We_1 = e_1$. Since group multiplication is invertible, left and right multiplication by $O$ must simply permute $M_t$.
        \end{enumerate}
        This proves the lemma.
    \end{proof}

We continue by bounding $\Vert S_{t,x_1}\Vert_1$. To do so, we will first bound $\Vert S_{t,x_1}\Vert_2$ and then use the general fact that $\Vert A\Vert_1 \leq \sqrt{\rank A}\Vert A \Vert_2$ for any operator $A$. Notice that
\begin{align}
    \Vert S_{t,x_1}\Vert_2^2 & = \Tr\left[ (V_{x_1}\otimes \id_{[2:t]}) \sum_{O \in M_t} R(O) (V_{x_1}\otimes \id_{[2:t]})^\dagger (V_{x_1}\otimes \id_{[2:t]}) \sum_{W \in M_t}R(W) (V_{x_1}\otimes \id_{[2:t]})^\dagger  \right]\\
    & = \sum_{O,W\in M_t}G(O,W)\ , 
\end{align}
where we define
\begin{align}
    G(O,W) : = \Tr[(V_{x_1}^\dagger V_{x_1} \otimes \id)R(O) (V_{x_1}^\dagger V_{x_1} \otimes \id) R(W)].
\end{align}
To bound $G(O,W)$, we introduce a combinatorial factor: define 
$$
\chi(O,W) := \rank(\id+OW).
$$
This characterizes the number of fixed points of the orthogonal matrix $OW$ and captures linear dependencies in applying these matrices. Further, $\Tr[R(O)]  = 2^{n(t-\chi(O,\id))}$. When $\chi(O,W)$ is not very large (less than or equal to 5) Lemma~\ref{lem:cut_rank} below will be used by applying Cauchy Schwarz to bound $G(O,W)$ as
\begin{align}
    \vert G(O,W)\vert & \leq \Vert (V_{x_1} \otimes \id_{2:t}) R(O) (V_{x_1}^\dagger \otimes \id_{2:t}) \Vert_2 \Vert (V_{x_1} \otimes \id_{2:t}) R(W) (V_{x_1}^\dagger \otimes \id_{2:t}) \Vert_2 \ ,
\end{align}
and then bounding each of these terms individually. When $\chi(O,W) \geq 5$, we can use a generalization of the same argument as for Lemma~\ref{lem:cut_rank} to prove a stronger bound. This is captured in Lemma~\ref{lem:cross_term}. Combined, the following two lemmas will show that
\begin{align}
    \vert G(O,W)\vert & \leq \Tr[V_{x_1}^\dagger V_{x_1}]^2\cdot \min\{ 2^{n(t-2)}, 2^{n(t+3-\chi(O,W))} \}\ .
\end{align}

\begin{lemma}\label{lem:cut_rank}
    Let $V:\Hc\rightarrow\Mc$ be an operator and let $O\in M_t$. Then,
    \begin{align}
        \Vert (V\otimes \id_{[2:t]}) R(O)(V \otimes \id_{[2:t]})^\dagger \Vert_2 \leq 2^{n(t-2)/2}\Tr[V^\dagger V]
    \end{align}
\end{lemma}
\begin{proof}
    Let $V$ have singular value decomposition
    \begin{align}
        V & = \sum_{a=1}^{2^k} \lambda_a \ket{\mu_a}\bra{\zeta_a}\ .
    \end{align}
    For convenience, we adopt the notation
    \begin{align}
        H & : = (V\otimes \id_{2:t}) R(O) (V\otimes \id_{2:t})^\dagger = \sum_{a,b} \lambda_a \lambda_b \ket{\mu_a}\bra{\mu_b}\otimes \underbrace{(\bra{\zeta_a}\otimes \id)R(O)(\ket{\zeta_b}\otimes \id)}_{:=H_{a,b}}
    \end{align}
    Using the orthonormality of the singular value decomposition,
    \begin{align}
        \Vert H \Vert_2^2 & = \sum_{a,b}\lambda_a^2 \lambda_b^2 \Vert H_{a,b}\Vert_2^2\ .
    \end{align}
    We now bound each term $\Vert H_{a,b}\Vert_2$. Since $O\in M_t$, by Lemma~\ref{lem:mt_facts}, we can express $O$ as 
    \begin{align}
        O & = \begin{pmatrix}
            O_{11} & O_{1*} \\
            O_{*1} & O_{**}
        \end{pmatrix}\ .
    \end{align}
    Furthermore, by Lemma~\ref{lem:mt_facts}, $O_{1*}$ and $O_{*1}$ must both be non-zero since $O\in M_t$. Let $x\in \F_2^{1\times n}$ and $y\in \F_2^{(t-1)\times n}$ represent the state of the first party as well as parties $2$ to $t$ respectively. Then, $R(O)$ has the~action
    \begin{align}
        R(O)\ket{x,y} & = \ket{O_{1,1}x+O_{1*}y, O_{*1}x+O_{**}y}\ .
    \end{align}
    Then, it follows that
    \begin{align}
        [H_{a,b}]_{y',y} = \bra{y'}(\bra{\zeta_a}\otimes \id_{2:t})R(O)(\ket{\zeta_b}\otimes \id_{2:t})\ket{y} & = \sum_{x: O_{*1}x+O_{**}y = y'} \overline{\zeta_a(O_{11}x+O_{1*}y)}\zeta_b(x) \ .
    \end{align}
    Further,
    \begin{align}
        \Vert H_{a,b}\Vert_2^2 & = \sum_{y,y' \in \F_2^{(t-1)\times n}} \left\vert [H_{a,b}]_{y',y}\right\vert^2\\
        & = \sum_y \sum_{y'} \left\vert \sum_{x: O_{*1}x+O_{**}y = y'} \overline{\zeta_a(O_{11}x+O_{1*}y)}\zeta_b(x) \right\vert^2\\
        & \leq \sum_y \sum_{y'} \left\vert\{ x \in \F_2^{1\times n}\ \vert \ O_{*1}x+O_{**}y = y' \}\right\vert \sum_{x: O_{*1}x+O_{**}y = y'} \left\vert \zeta_a(O_{11}x+O_{1*}y) \right\vert^2\left\vert\zeta_b(x) \right\vert^2\ .
    \end{align}
    where the inequality is Cauchy-Schwarz. Now, $O_{*1}$ is some non zero element of $\F_{2}^{(t-1)\times 1}$ and $O_{1*}x \in \F_2^{(t-1) \times n}$ is an outer product. Fixing $y$ and $y'$, there is at most a single $x$ such that $O_{*1}x+O_{**}y = y'$. To see this, let $i^* \in [t-1]$ be such that $[O_{1*x}]_{t} =1$. Row $i^*$ of the constraint then reads
    \begin{align}
        x = [O_{**}y+y']_{j,1:t-1}\ ,
    \end{align}
    which completely determines $x$. Hence, we can instead sum over $y$ and $x$ since $y'$ and $x$ completely determine each for a fixed $y$.

    The linear map $y\mapsto O_{1*}y$ from $\F_2^{(t-1)\times n}$ to $\F_2^{1\times n}$ is surjective since $O_{1*}$ is non-zero. Hence, the number of $y\in F_{2}^{(t-1)\times n}$ such that $O_{1*}y = z$ must be the size of the kernel of this map which is $\vert \F_2^{(t-1)\times n} \vert / \vert F_{2}^{1\times n} = 2^{n(t-2)}$.

    \begin{align}
        \Vert H_{a,b}\Vert_2^2 & \leq \sum_x \vert \zeta_b(x)\vert^2 \sum_y \vert \zeta_a(O_{11}x+O_{1*}y)\vert^2 \\
        & =2^{n(t-2)}\sum_x \vert \zeta_b(x)\vert^2 \sum_{x'} \vert \zeta_a(x')\vert^2  = 2^{n(t-2)}\ .
    \end{align}
    Hence,
    \begin{align}
        \Vert H\Vert_2^2 & = \sum_{a,b} \lambda_a^2 \lambda_b^2 \Vert H_{a,b}\Vert_2^2 \leq 2^{n(t-2)} \left(\sum_{a}\lambda_a^2\right)^2 = 2^{n(t-2)}\Tr[V^\dagger V]^2\ .
    \end{align}
\end{proof}

\begin{lemma}\label{lem:cross_term}
    Let $O,W\in \Oc_t$. Then, for all PSD matrices $Q \in \mathcal{B}(\Hc)$, it holds that
    \begin{align}
        \left\vert \Tr[(Q\otimes \id_{[2:t]})] R(O) (Q\otimes \id_{[2:t]}) R(W)]\right\vert & \leq \Tr[Q]^2\cdot  2^{n(t+3-\chi(O,W))}\ .
    \end{align}
\end{lemma}
\begin{proof}
     First, we will reduce to the case of rank $1$ operators. This follows by letting $Q = \sum_{a} \lambda_a \ket{\zeta_a}\bra{\zeta_a}$ and applying triangle inequality:
    \begin{align}
        \left\vert \Tr[(Q\otimes \id_{[2:t]})] R(O) (Q\otimes \id_{[2:t]}) R(W)]\right\vert & = \left\vert \sum_{a,b} \lambda_a \lambda_b \Tr[(\ket{\zeta_a}\bra{\zeta_a} \otimes \id_{2{t}}) R(O) (\ket{\zeta_b}\bra{\zeta_b} \otimes \id_{2{t}}) R(W)] \right\vert\\
        & \leq \sum_{a,b}\lambda_a \lambda_b \vert \underbrace{\Tr[(\ket{\zeta_a}\bra{\zeta_a}\otimes \id_{[2:t]})R(O)\ket{\zeta_b}\bra{\zeta_b}\otimes \id_{[2:t]}) R(W)]}_{:=H_{a,b}}\vert \ 
    \end{align}
 Hence, it suffices to show that $\vert H_{a,b}\vert \leq 2^{n(t+3-\chi(O,W))}$. Let us use the block decompositions
    \begin{align}
        O  = \begin{pmatrix}
            O_{11} & O_{1*}\\
            O_{*1} & O_{**}
        \end{pmatrix}\ , \quad W = \begin{pmatrix}
            W_{11} & W_{1*}\\
            W_{*1} & W_{**}
        \end{pmatrix}\ .
    \end{align}
We will represent the standard tensor product basis on $\Hc^{\otimes t}$ with matrices $x \in \F_2^{t \times n}$, where each row represents a bitstring for one of the tensor factors. Then, $R(O)$ has the action $Ox$, which applies $O$ transversely across the $t$ copies. By $x_i$ we will denote the corresponding row. Let $m\in \F_2^{t\times n}$ label a basis vector for $\Hc^{\otimes t}$. Then, define the linear forms
\begin{align}
    t(m) & := (Wm)_1 = W_{11}m_1 + W_{1*}m_{>1}\in \F_2^{1 \times n}\ , \\
    u(m) & := (Wm)_{2:t} = W_{*1}m_1 + W_{**}m_{>1} \in \F_2^{(t-1) \times n} \ , \\
    v(m) & := (Om)_1 = O_{11}m_1 + O_{1*}m_{>1} \in \F_2^{1 \times n}\ ,\\
    w(m) & := (Om)_{2:t} = O_{*1}m_1 + O_{**}m_{>1} \in \F_2^{(t-1) \times n} \ .
\end{align}
Hence, $R(W)\ket{m} = \ket{u(m),v(m)}$ and $R(O)\ket{m} = \ket{v(m),w(m)}$. Then, direct evaluation yields
\begin{align}
    H_{a,b} & = \sum_{m}\bra{m} (\ket{\zeta_a}\bra{\zeta_a}\otimes \id_{[2:t]})R(O)\ket{\zeta_b}\bra{\zeta_b}\otimes \id_{[2:t]}) R(W) \ket{m}\\
    & = \sum_{m}\zeta_a(m_1) (\bra{\zeta_a} \otimes \bra{m_{>1}})R(O)\ket{\zeta_b}\bra{\zeta_b}\otimes \id_{[2:t]}) \ket{t(m), u(m)}\\
    & = \sum_{m} \zeta_a(m_1) (\bra{\zeta_a} \otimes \bra{m_{>1}})
     R(O) (\ket{\zeta_b}\otimes \ket{u(m)}) \overline{\zeta_b}(t(m))\\
    & = \sum_m \zeta_a(m_1) (\bra{\zeta_a} \otimes \bra{m_{>1}}) R(O) \sum_{z \in \F_2^{1 \times n}} \ket{z}\otimes \ket{u(m)} \zeta_b(z) \overline{\zeta_b}(t(m))\\
    & = \sum_{m,z}  \zeta_a(m_1)\zeta_b(z)\overline{\zeta}_{b}(u(m) (\bra{\zeta_a} \otimes \bra{m_{>1}}) \ket{v(zu(m)), w(zu(m))}\\
    & = \sum_{m\in \F_2^{t \times n}} \sum_{z \in \F_2^{1 \times n}} \overline{\zeta}_{b}(t(m)) \zeta_b(z) \overline{\zeta}_a (v(zu(m))) \zeta_a(m_1) \delta\{ w(zu(m))= m_{>1} \}\label{eq:big_sum}\ ,
\end{align}
where $zu(m)\in \F_2^{t \times n}$ is to be understood as stacking $z\in\F_2^{1 \times n}$  as a row on top of $u(m) \in \F_2^{(t-1)\times n}$.

Similar to Lemma~\ref{lem:cut_rank}, we will consider the sizes of the preimages for various settings of the inputs and then apply Cauchy-Schwarz to remove the dependencies on $\zeta_a$ and $\zeta_b$. The first step to doing this is to notice that the identity in the summation above is the same as
    \begin{align}
        (\id+O_{**}W_{**})m_{>1} = O_{*1}z + O_{**}W_{*1}m_1\ .
    \end{align}
    Recall that $m_1, z \in \F_2^{1 \times n}$ and $m_{>1} \in F_2^{(t-1) \times n}$. Then, this constraint must hold for each triple of columns (equivalently, for each transversal set of qubits) and each triple must lie in the solution~space 
    \begin{align}
        \mathcal{S}_i := \{ (s,s',s'') \in \F_{2}\times \F_2 \times \F_2^{t-1}: (\id+O_{**}W_{**})s'' = O_{*1}s+O_{**}W_{*1}s' \}\ .
    \end{align}
    Because this is a linear constraint, $S_i$ is a subspace of $\F_2^{t+1}$. We let $\delta : = \dim \mathcal{S}_1$. Then, the entire solution space has dimension $n\delta$.

    Let $\mathcal{S} = \mathcal{S}_1^{\times n}$ be the solution space, consisting of a copy of $\mathcal{S}_i$ for each column. Notice that $m_1$ and $a$ appear as arguments for $\zeta_a$ and $\zeta_b$. Similarly, $w$ and $(Wm)_1$ appear as arguments for $\zeta_a$ and $\zeta_b$ as well. The idea will be to apply Cauchy-Schwarz and sum over these arguments independently. Then, these factors sum to $1$ because $\zeta_a$ and $\zeta_b$ are normalized states. The main difficulty in doing this will be to count the amount of redundancy in these summation, as is quantified by the size of the preimages of the projection $\mathcal{S} \mapsto (z,m_1)$. That is, by applying Cauchy-Schwarz to Eq~\eqref{eq:big_sum}, we obtain
    \begin{align}
        \vert H_{c,c'}\vert & \leq \left( \sum_{(z,m_1,m_{>1})\in S} \vert \zeta_b(z)\vert^2 \vert  \zeta_a(m_1)\vert^2 \right)^{1/2} \left( \sum_{(z,m_1,m_{>1})\in S} \vert \zeta_b(t(m)) \vert^2 \vert  \zeta_a(v(zu(m)))\vert^2 \right)^{1/2}\ .
    \end{align}
    Because $(z,m_1,m_{>1})\mapsto (z,m_1)$ is a linear map, its fibers are of size $2^{n\kappa}$ for some $\kappa$. In particular, the fibers are determined by the cosets of $\{0\}\times \{0\}\times \ker(\id_{t-1}+O_{**}W_{**})$. Thus, $\kappa = \dim \ker(\id_{t-1}+O_{**}W_{**})$. Then, the first summation above can be upper bounded as
    \begin{align}
         \sum_{(z,m_1,m_{>1})\in S} \vert \zeta_b(z)\vert^2 \vert  \zeta_a(m_1)\vert^2  & \leq 2^{n\kappa} \sum_{z,m_1} \vert \zeta_b(z)\vert^2 \vert  \zeta_a(m_1)\vert^2 = 2^{n\kappa}\ ,
    \end{align}
    where the final equality follows from $\zeta_a$ and $\zeta_b$ being normalized states.

    For the second, we are interested in the pre-images of the map $(z,m_1,m_{>1}) \mapsto (W_{11}m_1+W_{1*}m_{>1}, O_{11}z+O_{1*}(W_{*1}m_1+W_{**}m_{>1}))$. On $S_i$, this is the linear map
    \begin{align}
        \Lambda_i(s,s',s'') =  (W_{11}s'+W_{1*}s'', O_{11}s+O_{1*}(W_{1*}s'+W_{**}s''))\ .
    \end{align}
    Then, the entire map is $\Lambda := \Lambda_1^{\times n}$ as each column is acted upon with the same map. As a linear map, the fibers are of size equal to the null space of $\Lambda$, which is the product of $n$ copies of the null space of $\Lambda_i$. By rank nullity, this is $\vert \mathcal{S}\vert / \vert \text{im} \Lambda \vert = 2^{n(\delta-\dim\text{im}\Lambda_1)}$. Now, the image of $\Lambda_i$ is only two bits so the rank is either $0$, $1$, or $2$. For an upper bound, however, we can just assume the rank is $0$ and obtain
    \begin{align}
         \sum_{(z,m_1,m_{>1})\in S} \vert \zeta_b(t(m)) \vert^2 \vert  \zeta_a(v(zu(m)))\vert^2  & \leq 2^{n\delta}\ .
    \end{align}
    It is useful to note that $\delta \leq \kappa +2$ (by simply adding in the two degrees of freedom for $s$ and $s'$). Hence, 
    \begin{align}
        \vert H_{c,c'}\vert & \leq 2^{n(\kappa+1)}\ .
    \end{align}

    The final step is to relate $\kappa$ to $\chi := \chi(O,W)$. Note that
    \begin{align}
        \id_{t-1}+O_{**}W_{**} & = (\id_t + OW)_{**} + O_{*1}W_{1*}\ .
    \end{align}
    Since rank decreases by at most one by deleting a row or column from a matrix, $\rank (\id_t + OW)_{**} \geq \chi - 2$. Next, $\rank O_{*1}W_{*1} \leq 1$ and adding a rank $1$ matrix can only decrease the rank by at most $1$. Hence, $\rank (\id_{t-1}+O_{**}W_{**}) \geq \chi - 3$ and $\kappa \leq t-1-(\chi-3) = t+2-\chi$.
\end{proof}

With these lemmas proven, we can now bound $\Vert S_{t,2}\Vert_2$.

\begin{lemma}\label{lem:2_norm_bound}
Let $1<t < n$, then
    \begin{align}
        \Vert (V\otimes \id_{[2:t]})S_t(V\otimes \id_{[2:t]})^\dagger \Vert_2^2 \leq \vert M_t\vert 2^{6t+2} 2^{n(t-2)}\Tr[V^\dagger V]^2\ .
    \end{align}
\end{lemma}
\begin{proof}
    Lemma~\ref{lem:cut_rank} and Cauchy-Schwarz yields the bound
    \begin{align}
        \vert G(O,W)\vert \leq 2^{n(t-2)}\Tr[V^\dagger V]^2\ 
    \end{align}
    for all pairs. We also have, from Lemma~\ref{lem:cross_term},
    \begin{align}
        \vert G(O,W)\vert \leq 2^{n(t+3-\chi(O,W))}\Tr[V^\dagger V]^2\ .
    \end{align}
    Fixing $O \in M_t$, $W\mapsto OW$ is a bijection and 
    \begin{align}
        \#\{W\in \Oc_t : \chi(O,W) = r\} \leq 2^{rt}\ ,
    \end{align}
    which was proven in Lemma~\ref{lem:fixed_counting}.

    When $\chi(O,W) \leq 5$, the bound from Lemma~\ref{lem:cut_rank} is smaller than or equal to that of Lemma~\ref{lem:cross_term}. So, using the first when $\chi(O,W) \leq 5$ and the second when $\chi(O,W) > 5$, we obtain, for fixed $O$,

    \begin{align}
        \sum_{W \in M_t} \vert G(O,W)\vert & \leq \Tr[V^\dagger V]^2\left(2^{n(t-2)}\sum_{r=0}^5 2^{rt} + 2^{n(t+3)}\sum_{r=6}^t 2^{r(t-n)} \right)\ .
    \end{align}
    Taking the geometric sums and noting that $2^{-(n-t)} \leq 1/2$, this can be bounded by
    \begin{align}
        \Tr[V^\dagger V]^2 \left( 2^{n(t-2)+5t+1} + 2^{n(t+3)-6(n-t)+1}\right) \leq 2^{6t+2}2^{n(t-2)}\Tr[V^\dagger V]^2\ .
    \end{align}
    Summing over all $O \in M_t$ completes the proof.
\end{proof}

Now that we have a bound on $\Vert S_{t,x_1}\Vert_2$, we would like to obtain a bound on $\Vert S_{t,x_1}\Vert_1$. To do so, we apply the general inequality $\Vert A\Vert_1 \leq \sqrt{\rank A} \Vert A\Vert_2$. Let
\begin{align}
    \Pi_{t-1} : = \frac{1}{\vert \Oc_{t-1}\vert }\sum_{O \in \Oc_{t-1}}R(O) \in \mathcal{B}(\Hc_{2:t}) \ ,
\end{align}
which is a projector acting on all but the first round. Taking $P:= \id_{\Mc} \otimes \Pi_{t-1}$, Fact~\ref{lem:mt_facts} (d) implies that $PS_t P = S_t$. Hence, the rank of $S_t$ is at most the dimension of the memory times the support of $\Pi_{t-1}$. The latter is bounded in the following argument.

\begin{lemma}\label{lem:projection_rank}
    Let $1< t \leq n$ and $\Pi_{t-1} := 1/\vert \Oc_{t-1}\vert \sum_{O \in \Oc_{t-1}}R(O)$ acting on $\Hc_{2:t}$. Then, $\rank \Pi_{t-1} \leq 2^{n(t-1)+1}/\vert  \Oc_{t-1}\vert$.
\end{lemma}
\begin{proof}
    Since $\Pi_{t-1}$ is a projector, the rank is equal to the trace. The trace is exactly
    \begin{align}
        \Tr[\Pi_{t-1}]& = \frac{1}{\vert \Oc_{t-1}\vert}\sum_{O \in \Oc_{t-1}}\Tr[R(O)]\\
        & = \frac{1}{\vert \Oc_{t-1}\vert} \sum_{O \in \Oc_{t-1}}2^{n(t-1 - \rank(\id + O))}\\
        & = \frac{2^{n(t-1)}}{\vert \Oc_{t-1}\vert} \left(1 + \sum_{r=1}^{t-1}\#\{O\in\Oc_{t-1}: \rank(\id+O) = r\}  2^{-nr}\right) \\
        & \leq \frac{2^{n(t-1)}}{\vert \Oc_{t-1}\vert}\left( \sum_{r=0}^{t-1} 2^{-r(n-t+1)}\right) \\
        & \leq \frac{2^{n(t-1)+1}}{\vert \Oc_{t-1}\vert}\ ,
    \end{align}
    where the first inequality is via Lemma~\ref{lem:fixed_counting} and the second uses the fact that $t \leq n$ to bound the summation by $1$.
\end{proof}

We can now complete the proof of Lemma~\ref{lem:mt_bound}. From Eq~\eqref{eq:st_sum}, it suffices to bound
\begin{align}
    \frac{1}{2^{nt}}\sum_{x_1} \Vert S_{t,x_1}\Vert_1 & \leq 2^{3t+1}\sqrt{\vert M_t\vert }2^{-n(t/2+1)}\sum_{x_1 }\sqrt{\rank S_{t,x_1}}\Tr[V_{x_1}^\dagger V_{x_1}]\\
    & \leq 2^{1-(3n-k-7t)/2}\Tr[V_{x_1}^\dagger V_{x_1}]\\
    & = 2^{1-(n-k-7t)/2}\ ,
\end{align}
where the first inequality is Lemma~\ref{lem:2_norm_bound}, the second is Lemma~\ref{lem:projection_rank} as well as Lemma~\ref{lem:mt_facts}(c), and the final equality is that $\sum_{x_1}V_{x_1}^\dagger V_{x_1} = 2^{n}$ by completeness of the protocol.

\section{Non-adaptive learning stabilizer states}
\subsection{Upper bounds}
\label{sec:nonadaptive-upper}

We prove the upper bound in Theorem~\ref{thm:stablearning}. We emphasize that since we are working with non-adaptive algorithms, the measurement
schedule is fixed before any outcome is observed, all choices unitaries, measurements are fixed in advance.  All classical post-processing is performed only after all copies have been measured. The coherent quantum memory allows for correlations between rounds, but the way it is used is fixed ahead of time and independent of all prior outcomes.

Our protocol has three ingredients. 
\begin{enumerate}
    \item We precommit to a small number of
random Clifford branches.\footnote{By branch here, we mean application of this Clifford to a fresh copy of the unknown state $\ket{\psi}$.} With constant probability, at least one branch puts the unknown stabilizer state
into full-support stabilizer form. 
\item We run a procedure that breaks this state into blocks and runs Bell-sampling on these blocks to learn the graph matrix. On a branch where the state is in full-support stabilizer form, this recover the matrix with high probability.
\item Simultaneously (and independently), we perform  random stabilizer-basis measurements; after the label subspace has been learned, we use these classical outcomes recover the stabilizer~signs. 
\end{enumerate}
Throughout this section, $\mathcal{Z}\le \F_2^{2n}$ denotes the all-$Z$ Lagrangian,
i.e. $\mathcal{Z}:=\{0\}\times\F_2^n$.  Let  
$\ket{\phi}$ be a state with  full-support stabilizer form with unsigned stabilizer group
$M_\phi\le\F_2^{2n}$.  Observe that $\ket{\phi}$ is full-support if and only if 
$    M_\phi\cap Z=\{0\}.$ 
Equivalently, there is a unique symmetric matrix
$A\in\F_2^{n\times n}$ such that
$    M_\phi=\{(u,Au):u\in\F_2^n\}.$ 
 We now state a simple fact that a random Clifford transforms a stabilizer state to a stabilizer state with \emph{full-support}. This is a folklore result, which we state and prove below for convenience.
\begin{lemma}
\label{lem:random-clifford-full-support}
Fix an $n$-qubit stabilizer state $\ket{\psi}$, and let $C$ be a uniformly
random Clifford.  Then
$$
    \Pr_C[C\ket{\psi}\text{ has full-support stabilizer form}]
    =
    \prod_{j=1}^n(1+2^{-j})^{-1}
    \ge 0.4 .
$$
\end{lemma}

\begin{proof}
Let $M_\psi\le \F_2^{2n}$ be the unsigned stabilizer group of $\ket{\psi}$. As we observed above, $C\ket{\psi}$ is full-support if and only if
$C(M_\psi)\cap Z=\{0\}$. The result then follows from Fact~\ref{fact:random_lagrangian}.
\end{proof}

We precommit to $R=O(\log(1/\delta))$ independent uniformly random Cliffords
$C_1,\ldots,C_R$.  By Lemma~\ref{lem:random-clifford-full-support}, with
probability at least $1-\delta$, at least one branch $C_j\ket{\psi}$ is
full-support.  The learner need not know during the measurement stage which
branch succeeds: it runs the same precommitted schedule for every branch, and
selects a successful branch later by classical postprocessing.

\paragraph{Graph learning on a Clifford branch}
Fix a branch $j$.  Since $C_j$ is part of the precommitted schedule, we
treat $\ket{\psi_j}:=C_j\ket{\psi}$ as the effective input state to the
block measurements on branch $j$; its (unsigned) stabilizer group is $M_j:=C_j(M_\psi)$.   Since Clifford conjugation sends Pauli
operators to Pauli operators, the induced symplectic action of $C_j$ maps
$M_\psi$ to $M_j$, i.e.,
$$
M_j=M_{\psi_j}=C_j(M_\psi).
$$

\begin{theorem}
\label{thm:nonadaptive-label-recovery}
Fix a Clifford branch $j$, and write $\ket{\psi_j}=C_j\ket{\psi}$ and
$M_j=C_j(M_\psi)$.  If $\ket{\psi_j}$ is full-support, then there is a fixed
nonadaptive measurement schedule using $O(n^2/k)$ copies and at most $k$ qubits
of coherent memory that recovers $M_j$ with failure probability
$2^{-\Omega(n)}$.
\end{theorem}

\begin{proof}
Fix a partition $[n]=S_1\sqcup\cdots\sqcup S_m$, where
$m=\lceil n/k\rceil$ and $|S_i|\le k$ for every $i$.  Since $\ket{\psi_j}$ is
full-support, there is a unique symmetric matrix $B_j\in\F_2^{n\times n}$ such
that $M_j=\{(u,B_j u):u\in\F_2^n\}$. We describe the two-copy measurement for a block $S\subseteq[n]$ with
$|S|\le k$. On the first copy of $\ket{\psi_j}$, store the qubits in $S$ and
measure the complementary qubits in the computational basis, obtaining
$z_{\bar S}$.  On the second copy, again measure the qubits in $\overline{S}$ in the computational basis, obtaining $z'_{\bar S}$.
Bell-measure the remaining $S$ register and the copy of the $S$ register stored in the memory, obtaining $(a_S,b_S)$.

Define $a\in\F_2^n$ by $a|_S=a_S$ and
$a|_{\bar S}=z_{\bar S}+z'_{\bar S}$.  The block measurement is exactly the
marginal of full Bell sampling in which we keep the full $X$-label $a$ and only
the $S$-coordinates of the $Z$-label.  Since full Bell sampling a stabilizer
state outputs a uniformly random label $(a,b)\in M_j+c_j$, for some offset $c_j$, and since
$M_j=\{(u,B_j u):u\in\F_2^n\}$, the recorded labels satisfy
$    b_S=\Pi_S b=\Pi_S B_j a +\Pi_S c_j$. Here $\Pi_S$ indicates projecting onto the coordinates in $S$ and discarding all others. Thus one block Bell sample gives the affine linear sample
$    (a,\Pi_S B_j a),$ 
with $a$ uniform over $\F_2^n$. Equivalently, each repetition gives one
uniformly random affine linear equation for the unknown row block $\Pi_SB_j$. To remove the offset, we fix one reference sample $(a^{(0)},b^{(0)}_S)$ per block and solve the linear equations $b^{(j)}\oplus b^{(0)}=\Pi_S B_j (a^{(i)}\oplus a^{(0)})$.

These equations will determine $\Pi_SB_j$ if the samples $\{a^{(i)}\oplus a^{(0)}\}_i$ span $\F_2^n$. But the $a_j$'s are uniformly random and, with $\Omega(n)$ samples, the samples span $\F_2^n$ with probability $\geq 1-2^{-\Omega(n)}$.  Conditioned on this, the linear equations uniquely determine $\Pi_SB_j$. We run this fixed procedure for every block $S_i$.  A union bound over $m=\lceil n/k\rceil\le n$ blocks shows that all row blocks $\Pi_{S_i}B_j$ are recovered with probability $1-2^{-\Omega(n)}$.  These row
blocks assemble into $B_j$, and therefore into
$M_j=\{(u,B_j u):u\in\F_2^n\}$. The number of copies used is $O(n)\cdot m=O(n^2/k)$.  The schedule is
fixed in advance: for each repetition, the branch $C_j$ and the block $S_i$ are
already specified.  Since each two-copy measurement stores at most $|S_i|\le k$ 
qubits coherently and the schedule is fixed ahead of the protocol, the above proves the theorem statement.
\end{proof}

At this point, the learner selects a branch $j$ and accepts if, for every block $S_i$, the
sampled $a$-labels span $\F_2^n$ and the recovered row blocks assemble into a symmetric matrix $B_j$.  Full-support branches pass with probability $1-2^{-\Omega(n)}$. A non-full-support branch is rejected deterministically by the spanning check.\footnote{
Indeed, the observed $a$-labels are precisely the $X$-coordinates of labels in $M_j$.  If the branch is not full-support, then the $X$-projection of $M_j$ has
dimension strictly smaller than $n$, so all sampled $a$-labels lie in a fixed proper subspace of $\F_2^n$ and can never span $\F_2^n$.  Thus the only
branch-selection failure is that a full-support branch exists but its random $a$-samples fail to span, which occurs with probability $2^{-\Omega(n)}$.}  The learner then chooses any accepted branch $j^\star$ and uses the already collected sign-recovery data for that same branch.
\paragraph{Sign recovery.} The block Bell procedure recovers the unsigned stabilizer group $M_j$, but it
does not recover the stabilizer signs.  Indeed, Bell sampling depends on
$|\braketbra{\psi_j^*}{P_u}{\psi_j}|^2$, so it cannot distinguish the cases
$P_u\ket{\psi_j}=+\ket{\psi_j}$ and
$P_u\ket{\psi_j}=-\ket{\psi_j}$.  We therefore recover signs using a separate
precommitted layer of random stabilizer-basis measurements. That is, in each branch $j$, we choose in advance independent uniformly random Cliffords $D_{j,1},\ldots D_{j,T}$. For each $s \in [T]$, we take a fresh copy, and measure $D_{j,s}C_{j}\ket{\psi}$ in the computational basis. All these choices are fixed before any outcomes are observed.

This works because measuring in the computational basis is equivalent to jointly measuring the commuting observables $\{Z(c) \ \vert  \ c\in \F_{2}^{n}\}$. Then, by applying a random Clifford before a computational basis measurement, this is the same as measuring in a uniformly random Lagrangian subspace $L \leq \F_2^{2n}$. This then reveals the correct sign for all $u \in L \cap M_{\psi_j}$. Explicitly, say that we apply Clifford $D$ to $\ket{\psi_j}$. Then, there is a linear function $g:\F_2^{2n}\rightarrow \F_2^{2n}$ and a known function $\omega:\F_2^{n}\rightarrow \{\pm 1\}$ such that
\begin{align}
    D^{-1}Z(c) D & = \omega(c) P_{g(c)}\ .
\end{align}
Say we measure the bitstring $y\in \F_2^n$. If $g(c) \in M_{\psi_j}$, then the sign of this Pauli is exactly
\begin{align}
    \omega(c)(-1)^{y\cdot c}\ ,
\end{align}
since $\ket{y}$ has eigenvalue $(-1)^{y\cdot c}$ with respect to the Pauli string $Z(c)$. It remains to argue that signs for the entire Lagrangian subspace $M_{\psi_j}$ can be recovered with high probability. To do so, it suffices to find the signs for a generating set.

\begin{lemma}
\label{lem:random-basis-spans-signs}
Let $M\le\F_2^{2n}$ be a fixed Lagrangian. Let $T=O(n)$ and let
$L_1,\ldots,L_T$ be independent uniformly random Lagrangians.  Then,
$$
    \Pr\left[
        \operatorname{span}\Bigl(\bigcup_{s=1}^T M\cap L_s\Bigr)\ne M
    \right]
    \le
2^{-\Omega(n)}.
$$
\end{lemma}

\begin{proof}
The span fails to equal $M$ if and only if there is a hyperplane
$H<M$ such that
$  M\cap L_s\subseteq H$ 
for every $s$.  There are at most $2^n-1$ hyperplanes in $M$, so by a
union bound it suffices to prove that, for any fixed hyperplane $H<M$ and a
uniformly random Lagrangian $L$,
$   \Pr[M\cap L\not\subseteq H]\ge c_1$  
for an absolute constant $c_1>0$. 
Let
$$
    N:=\left|(M\cap L)\cap(M\setminus H)\right|.
$$
We bound $\Pr[N\ge1]$ below by a positive constant using the second-moment
method. 
For any fixed nonzero vector $v\in M$, the probability that $v\in L$ is
$$
    \Pr[v\in L]
    =
    \frac{\prod_{j=1}^{n-1}(2^j+1)}
         {\prod_{j=1}^{n}(2^j+1)}
    =
    \frac{1}{2^n+1}.
$$
Indeed, Lagrangians containing $v$ correspond to Lagrangians in the symplectic
quotient $v^\perp/\langle v\rangle\cong\F_2^{2(n-1)}$.  Since
$|M\setminus H|=2^{n-1}$, we get
$$
    \E[ N]
    =
    \frac{2^{n-1}}{2^n+1}
    \ge \frac13
$$
for all $n\ge1$. 
For the second moment, note that we are working over $\F_2$, so any two
distinct nonzero vectors $v\ne w$ are automatically linearly independent.
Moreover both vectors lie in the Lagrangian $M$, so $[v,w]=0$, and hence
$\langle v,w\rangle$ is a two-dimensional isotropic subspace.  The
probability that a uniformly random Lagrangian $L$ contains both $v$ and
$w$ is
$$
    \frac{\prod_{j=1}^{n-2}(2^j+1)}
         {\prod_{j=1}^{n}(2^j+1)}
    =
    \frac{1}{(2^n+1)(2^{n-1}+1)},
$$
by reducing modulo $\langle v,w\rangle^\perp/\langle v,w\rangle\cong
\F_2^{2(n-2)}$.  Therefore, summing the diagonal $w=v$ contribution and the
$w\ne v$ contributions separately,
$$
    \E [N^2]
    =
    \frac{2^{n-1}}{2^n+1}
    +
    \frac{2^{n-1}(2^{n-1}-1)}
         {(2^n+1)(2^{n-1}+1)}
    \le
    \frac12+\frac{2^{n-1}-1}{2(2^{n-1}+1)}
    \le
    \frac12+\frac12=1.
$$
By Paley--Zygmund,
$$
    \Pr[N\ge1]\ge \frac{(\E N)^2}{\E N^2}\ge \frac{(1/3)^2}{1}=\frac19.
$$
Thus $c_1=1/9$ works.  The union bound over hyperplanes gives
$$
    \Pr\left[
        \operatorname{span}\Bigl(\bigcup_s M\cap L_s\Bigr)\ne M
    \right]
    \le
    (2^n-1)(1-c_1)^T
    \le
    2^n(1-c_1)^T,
$$
Setting $T=O(n)$, for a suitably chosen constant in the $O(\cdot)$ gives the lemma statement.
\end{proof}

\begin{corollary}
\label{prop:nonadaptive-sign-recovery}
Fix a branch $j$ for which $M_j$ has been recovered.  Using $T=O(n)$
additional precommitted random stabilizer-basis measurements, one recovers the
sign function $\chi_j$ on all of $M_j$ with failure probability
$2^{-\Omega(n)}$.
\end{corollary}

\begin{proof}
Observe that for each $u\in M_j\cap L_{j,s}$, the $s$-th computational-basis outcome
determines the eigenvalue of $P_u$ on $\ket{\psi_j}$, up to a known
Clifford phase.  Hence the data gives $\chi_j$ on a spanning set of
$M_j$.  By Fact~\ref{fact:cocycle-extension}, this determines $\chi_j$ on
all of $M_j$.
\end{proof}
\paragraph{Putting the protocol together.}
Choose $R$ independent random Clifford branches $C_1,\ldots,C_R$.  For each
branch, run the nonadaptive label-recovery procedure of
Theorem~\ref{thm:nonadaptive-label-recovery} and, independently, the sign-recovery
procedure of Proposition~\ref{prop:nonadaptive-sign-recovery}.  All Clifford
choices and all measurements are fixed before any outcomes are observed. By Lemma~\ref{lem:random-clifford-full-support}, with probability at least
$1-(1-c_0)^R$, at least one branch $j^\star$ puts the state in full-support
form.  By the branch-selection step after
Theorem~\ref{thm:nonadaptive-label-recovery}, postprocessing identifies such a
branch and recovers $M_{j^\star}$ with failure probability $2^{-\Omega(n)}$.
By Corollary~\ref{prop:nonadaptive-sign-recovery}, the already collected
random-basis data for the same branch recovers $\chi_{j^\star}$ with failure
probability $2^{-\Omega(n)}$.

Thus the learner obtains the signed stabilizer data
$(M_{j^\star},\chi_{j^\star})$ for the rotated state
$\ket{\psi_{j^\star}}=C_{j^\star}\ket{\psi}$.  Since $C_{j^\star}$ is known,
we classically pull this data back by $C_{j^\star}^{-1}$ to recover
$(M_\psi,\chi_\psi)$, and hence $\ket{\psi}$. Each branch uses $O(n^2/k)$ copies: $O(n^2/k)$ for label recovery and $O(n)$
for sign recovery, with $O(n)\le O(n^2/k)$ for $k\le n$.  Thus $R$ branches use
$O((n^2/k)R)$ copies.  Taking $R=O(1)$ gives constant success probability and
sample complexity $O(n^2/k)$.  Taking $R=O(\log(1/\delta))$ gives failure
probability at most $\delta+2^{-\Omega(n)}$ and sample complexity
$$
    O\!\left(n^2/k\cdot \log1/\delta\right).
$$
The protocol is fully nonadaptive: the learned branch $j^\star$ is used only in
classical postprocessing.  The only cross-copy coherent memory used is the
$k$-qubit memory in the block Bell samples.

\subsection{Lower bound}
\label{sec:learning-lower-bound}

We prove the lower bound already for the subclass of real degree-$2$ phase states.  Consider a phase state $\ket{\psi_A}=\frac{1}{\sqrt{2^n}}\sum_x (-1)^{x^\top A x}\ket{x}$. 
Since the state is parameterized by $A$, we have that
$
    H(A)={n(n+1)}/2=\Theta(n^2).
$
Thus any learner for this subclass must recover $\Theta(n^2)$ bits of classical
information.

\begin{lemma}\label{lem:moment_bound}
    Let $\{F_x\}_x$ be any POVM on $\mathcal{H}$. Then, for uniformly random $A\in\F_2^{n \times n}$, it holds that
    \begin{align}
        \E_A\left[ 2^n \sum_x \frac{{\Tr[F_x \psi_A]^2}}{\Tr[F_x]} \right] \leq 3\ .
    \end{align}
\end{lemma}
\begin{proof}
    Let $F_x = \sum_{y} w_{y\vert x} \varphi_{y\vert x}$ be a refinement of the POVM into rank one aspects for all $x$. Then, by Cauchy-Schwarz,
    \begin{align}
        \E_A 2^n \sum_x \frac{{\Tr[F_x \psi_A]^2}}{\Tr[F_x]} & \leq  E_A  2^n \sum_x \sum_{y} w_{y\vert x} \Tr[\varphi_{y\vert x}\psi_A]^2\ .
    \end{align}
    We claim that
    \begin{align}
        \E_A[\Tr[\psi_A \varphi]^2] \leq \frac{3}{2^{2n}}\ ,
    \end{align}
    for all unit vectors $\ket{\varphi} = \sum_x \alpha_x \ket{x}$. To see this, write $q_A(x) = x^\top A x$. Then,
    \begin{align}
        \E_A[\Tr[\psi_A \varphi]^2] & = \sum_{x,y,z,w} \alpha_x \alpha_y \overline{\alpha}_z \overline{\alpha}_w \E_A\left[ (-1)^{q_A(x)+q_A(y)+q_A(z)+q_A(w)} \right]\ .
    \end{align}
The expectation over $A$ is zero unless
$$
    q(x)+q(y)+q(z)+q(w)=0
$$
for every polynomial $q$ spanned by the monomials $x_i$ and $x_ix_j$ for
$i<j$.  This condition forces the four points $x,y,z,w$ to be paired:
$$
    x=y,\ z=w,
    \qquad
    x=z,\ y=w,
    \qquad
    x=w,\ y=z.
$$
Indeed, the linear monomials first imply $x+y+z+w=0$.  Writing
$z=x+a$ and $w=y+a$, the quadratic monomials imply
$a_i(x_j+y_j)+a_j(x_i+y_i)=0$ for every $i<j$, which forces either $a=0$ or
$x=y$.  These are exactly the three pairing patterns above. Therefore the fourth moment is bounded by the contribution of the three
pairings:
$$
    \E_A [|\bk{\varphi}{\psi_A}|^4]
    \le
    \frac{3}{2^{2n}}
    \left(\sum_x|\varphi_x|^2\right)^2
    =
    \frac{3}{2^{2n}}.
$$
By completeness of the POVM, $\sum_{x}\sum_{y\vert x}w_{y\vert x } = 2^n$, completing the proof.
\end{proof}

\begin{theorem}
\label{thm:quadratic-phase-lower-bound}
A non-adaptive learner that identifies an unknown state from the ensemble $\{\ket{\psi_A}\}$ with probability $\geq 2/3$, using at most
$k$ qubits of coherent quantum memory, requires
$
    \Omega\!\left({n^2}/({2k+\log 3})\right)
$~copies.  
\end{theorem}
\begin{proof}
    Let $\hat{A}$ be a final estimate of $A$ from the learning protocol (which can be obtained from the estimate of $\ket{\psi_A}$). $A$ itself is uniformly distributed over upper-triangular matrices in $\F_2^{n \times n}$.
    
     If the learner succeeds with probability at least $2/3$, it also succeeds on identifying $A$ with the same probability. Hence, Fano's inequality implies that
    \begin{align}
        I(A;\hat{A}) \geq H(A) - 1-\frac{1}{3}H(A) = \Omega(n^2)\ .
    \end{align}

    We will now show that each round can only provide at most $2k+\log 3$ bits of information. Fix a round $i$. Since the protocol is non-adaptive, outcomes in different rounds depends on each other only through the memory. That is, given memory state $\eta_{i-1}$ and history $x_{<i}$, $x_i$ only depends on $\eta_{i-1}$. Then, the POVM performed in round $i$ is some fixed measurement $\{E_{x}\}_x$ applied to $\psi_A \otimes \eta{i-1}$. For each $x_{<i}$ let $p(\cdot \vert x_{<i},A) = \Tr[E_x (\psi_A \otimes \eta_{A,x_{<i}})]$ be the conditional distribution and $\overline{p} := \E_{A\vert x_{<i}}p(\cdot \vert x_{<i},A)$ its average. Also let $u(x) = \Tr[E_x]/2^{n+k}$ be the distribution induced by the maximally mixed state. Then, we have that
    \begin{align}
        I(A;x_i\vert x_{<i}) & = \E_{A \vert x_{<i}}D(p(\cdot \vert x_{<i},A)\Vert \overline{p})\\
        & \leq \E_{A \vert x_{<i}}D(p(\cdot \vert x_{<i},A)\Vert u)\\
        & \leq \E_{A\vert x_{<i}}\log \sum_x \frac{p(x\vert x_{<i},A)^2}{u(x)}\ ,
    \end{align}
    where the final inequality is Jensen's. Now, $\eta_{i-1} \preceq \id_k$, and hence
    \begin{align}
        p(x\vert x_{<i},A) \leq 2^k q_A(x),\quad q_A(x) := \Tr[F_x \psi_A],\quad F_x := 2^{-k}\Tr_{\mathcal{M}}[E_x]\ .
    \end{align}
    Here $\{F_x\}_x$ is a complete POVM on $\mathcal{H}$, the single copy state space. Using this, we have
    \begin{align}
        I(A;x_i \vert x_{<i}) \leq 2k + \E_{A\vert x_{<i}}\log \sum_x \frac{q_A(x)^2}{u(x)}\ .
    \end{align}
    Notice that the summation inside the expectation value does not depend upon $x_{<i}$. Averaging over $x_{<i}$ then yields
    \begin{align}
        I(A;x_i\vert x_{<i}) & \leq 2k + \E_{x_{1:i=1}}\E_{A\vert x_{<i} = x_{1:i-1}}\log \sum_x \frac{q_A(x)^2}{u(x)}\\
        & \leq 2k + \log \E_A \sum_x \frac{q_A(x)^2}{u(x)}\\
        & \leq 2k + \log 3\ ,
    \end{align}
    where the final inequality is Lemma~\ref{lem:moment_bound}. Then, by data processing,
    \begin{align}
        I(A;\hat{A}) & \leq I(A;\vec{x})\\
        & = \sum_{i=1}^t I(A;x_i \vert x_{<i})\\
        & \leq t(2k+3)\ .
    \end{align}
    Combining the upper and lower bound on $I(A;\hat{A})$ gives $t(2k+3) = \Omega(n^2)$, providing the theorem.

\end{proof}

\bibliographystyle{alpha}
\bibliography{ref}

\appendix
\section{Total variational distance for measures}\label{app:tvd}
In our lower bound proof in section~\ref{sec:lb_testing} we work with an unnormalized measure $\mathcal{P}_{\sigma_P}$. Here we show that it is no issue to consider likelihood ratios with respect to this measure. Let $(\Sigma, \mathcal{F})$ be a $\sigma$-algebra. For measures $\mu$ and $\nu$ (which we take to be bounded for simplicity) on this space, we can define the total variational distance to be
\begin{align}
    \TV(\mu,\nu) : = \frac{1}{2} \Vert \mu - \nu\Vert_1 = \frac{1}{2}\int_\Sigma \left\vert \frac{d\mu}{d\lambda} - \frac{d\nu}{d\lambda}\right\vert d\lambda\ ,
\end{align}
where $d\lambda$ is some measure that dominates both $\mu$ and $\eta$ and the integrand is over the corresponding Radon derivatives. Our proofs in section~\ref{sec:lb_testing} required usage of the triangle inequality, which clearly still holds here.

Further, the proof of Lemma~\ref{lem:like_ratio_prob} expanded the total variational distance as the sum of outcomes where one probability distribution has no less weight than the other. Here we will quickly show that we can still work with likelihood ratios even when the measures are not normalized. Assume that $\mu \ll \nu$. Then, there exists a Radon derivative $f:=d\mu/d\nu$ such that $\mu(E) = \int_E fd\nu$ for any measurable set $E \in \mathcal{F}$. Let $\Sigma_+ = \{\omega \in \Sigma: f \geq 1\}$ and $\Sigma_- : =\{\omega \in \Sigma: f < 1\}$. Since $\mu \ll \nu$, the total variation distance is
\begin{align}
    \TV(\mu,\eta) & = \frac{1}{2}\int_\Sigma \vert f-1\vert d\nu\\
    & = \frac{1}{2}\int_{\Sigma_+} (f-1) d\nu + \frac{1}{2}\int_{\Sigma_-}(1-f) d\nu\\
    & = \int_{\Sigma_-}(1-f)d\nu + \frac{1}{2}\left(\mu(\Sigma) - \nu(\Sigma)\right)\ .
\end{align}
Now say that $E=\{\omega: f(\omega) \geq 1-\delta\}$ is such that $\nu(E^c) \leq \beta$. We then have that $E = \Sigma_+ \sqcup (\Sigma_-\cap E)$. This allows us to bound the total variational distance as
\begin{align}
    \TV(\mu,\nu) & = \frac{1}{2}\left(\mu(\Sigma)-\nu(\Sigma)\right) + \int_{\Sigma_-\cap E}(1-f)d\nu + \int_{E^c} (1-f)d\nu\\
    & \leq \frac{1}{2}\left(\mu(\Sigma)-\nu(\Sigma)\right) + \delta \nu(\Sigma_- \cap E) + \nu(E^c)\\
    & \leq \frac{1}{2}\left(\mu(\Sigma)-\nu(\Sigma)\right) + \delta \nu(\Sigma) + \beta\\ & \leq \frac{1}{2}\left(\mu(\Sigma)-\nu(\Sigma)\right) + \delta + \beta\ .
\end{align}

In Section~\ref{sec:lb_testing} we consider $\nu = \mathcal{P}_{mm}$ and $\mu = \mathcal{P}_{\sigma_P}$. Clearly $\mu \ll\nu$ and hence the derivation above applies. Further,
\begin{align}
    \frac{1}{2}\vert \mu(\Sigma) - \nu(\Sigma)\vert & = \frac{1}{2}\vert \Tr[\sigma_P]-1\vert\\
    & = \frac{1}{2}\left\vert \Tr[\sigma_P - \rho_P]\right\vert\\
    & \leq \vert \sigma_P - \rho_P\Vert_1 \leq \left( O(2^{t-n} \right)\ .
\end{align}
Hence, we can still appeal to Lemma~\ref{lem:le_cam_exp} to upper bound $\TV(\mathcal{P}_{mm},\mathcal{P}_{\sigma_P})$.

\section{Purity Testing Lower Bound}\label{app:purity}
To complete a rigorous proof of the lower bound in Theorem~\ref{thm:stabtestingupper} we require that $\TV(\mathcal{P}_{mm},\mathcal{P}_H)$ must be small, where $\mathcal{P}_H$ is the distribution induced by measuring $t$ copies of a Haar random state. Here we will show that it is difficult to distinguish a Haar random state from the maximally mixed state using nearly identical steps to our stabilizer testing lower bound. The two key ingredients of the stabilizer testing bound are the following:
\begin{enumerate}
    \item $\E_\psi[\psi^{\otimes t}] \approx 2^{-nt}\sum_{O \in G \leq \Oc_{t}} R(O)$. 
    \item We can count the number of $O \in G$ such that $\rank(\id +O) =r$.
\end{enumerate}

Given these ingredients, we can apply the approach of Theorem~\ref{thm{lb_tv}}. That is, Lemma~\ref{lem:cut_rank} and Lemma~\ref{lem:cross_term} still apply and can be used to obtain an analogue of Lemma~\ref{lem:mt_bound}. The dependence on $t$ in the exponent will come from counting $O \in G$ such that $\rank(\id +O) =r$. Then, we can apply recursion as was done in Theorem~\ref{thm{lb_tv}} to complete the proof.

The first ingredient does indeed hold for Haar random states:
\begin{fact}[{\cite[Proposition 6]{harrow2013church}}]
    The ensemble average of $t$ copies of a Haar random state is given by
    \begin{align}
        \rho_H & = \frac{\sum_{\pi \in S_t} R(\pi)}{\prod_{j=0}^{t-1}(2^n+j)}\ .
    \end{align}
\end{fact}
Note that $R(\pi)$ here is exactly the same representation that appears for $\Oc_t$ (and more generally elements of the Clifford commutant). For permutations, this simplifies to permutations of the copies. Further, $S_t \leq \Oc_t$. 

The second ingredient holds as well:
\begin{lemma}\label{lem:sn_counting}
        To $\pi \in S_t$ associate a permutation matrix $\pi \in \F_2^{t \times t}$. The number of $\pi \in S_t$ such that $\rank(\id+\pi)=r$ is no more than $t^{2r}$.
    \end{lemma}
    \begin{proof}
        For permutations, $\rank(\id+\pi)$ is exactly $\vert \pi \vert$, the transposition length of the permutation (the minimum number of transpositions to produce $\pi$), which is exactly the Cayley distance on $S_t$ generated by transpositions. Then, there are $\binom{t}{2}$ transpositions and hence
        \begin{align}
            \#\{\pi\in S_t : \vert \pi \vert = r\} \leq \binom{t}{2}\#\{\pi\in S_t : \vert \pi \vert = r-1\}\ .
        \end{align}
        For $r=0$, there is only the identity, hence recursion yields the bound $\left( \binom{t}{2} \right)^r \leq t^{2r}$.
    \end{proof}
It is interesting to note that this quantity has previously appeared in proving lower bounds via PPT methods~\cite{harrow2023approximate}.

Now, we can use these properties to prove hardness of purity testing.

\begin{theorem}\label{thm:purity_lb}
    There is an absolute constant $c > 0$ such that if $t^2 < c2^n$, then, for any protocol using $k$-qubits of memory and single-copy measurements,
    \begin{align}
        \E_{\vec{x}\sim P_{mm}}[\vert L(\vec{x})-1\vert] \leq \frac{t^2}{2^n} + 2^{3/2-(n-k-15\log t)/2}\ .
    \end{align}
    Consequently, $t = 2^{\Omega(n-k)}$ samples are required to distinguish between a Haar random state and the maximally mixed state.
\end{theorem}
\begin{proof}
    We will consider the likelihood ratios
    \begin{align}
        L(\vec{x}) & = \frac{2^{nt}}{\prod_{j=0}^{t-1}(2^n+j)} \frac{1}{\Tr[E_{\vec{x}}]}\sum_{\pi \in S_t} \Tr[R(\pi) E_{\vec{x}}] \\
        & \geq \left(1-\frac{t}{2^n} \right)^{t}\frac{1}{\Tr[E_{\vec{x}}]} \sum_{\pi \in S_t} \Tr[R(\pi) E_{\vec{x}}]\ .
    \end{align}
    Then, the expected likelihood ratios satisfy
    \begin{align}
        \E_{\vec{x}\sim P_{mm}}\left[\vert L(\vec{x})-1 \vert \right] & \leq 1-\left(1-\frac{t}{2^n} \right)^{t} + \frac{1}{2^{nt}}\sum_{\vec{x}} \left\vert \sum_{\substack{\pi \in S_t \\ \pi \neq \id }} \Tr[R(\pi) E_{\vec{x}}] \right\vert\\
        & \leq \frac{t^2}{2^n} + \frac{1}{2^{nt}}\sum_{\vec{x}} \left\vert \sum_{\substack{\pi \in S_t \\ \pi \neq \id }} \Tr[R(\pi) E_{\vec{x}}] \right\vert\ .
    \end{align}
    Similar to the case of $\Oc_t$, we will split $S_t$ into $S_{t-1}$, identified with all permutations $\pi$ such that $\pi(1) = 1$, and $N_t := S_t \backslash S_{t-1}$. By the same argument as in Section~\ref{sec:lb_testing}, it suffices to bound the bias from $N_t$ and then recursively apply the same bound to all remaining terms. This results in Lemma~\ref{lem:nt_bound}, which we now apply to obtain:
    \begin{align}
        \E_{\vec{x}\sim P_{mm}}\left[ \left\vert L(\vec{x})-1\right\vert\right] & \leq \frac{t^2}{2^n} + \frac{1}{2^{nt}}\sum_{\vec{x}} \left\vert \sum_{\substack{\pi \in S_{t-1} \\ \pi \neq \id}} \Tr[R(\pi) E_{\vec{x}}] \right\vert + 2^{3/2-(n-k-13\log t)/2} \\ 
        & \leq \frac{t^2}{2^n} + 2^{3/2-(n-k)/2}\sum_{j=2}^t 2^{13/2 \log j}\\
        & \leq \frac{t^2}{2}+2^{3/2-(n-k - 15\log t)/2}\ .
    \end{align}
    \end{proof}
    To obtain Lemma~\ref{lem:nt_bound} we note that  Lemma~\ref{lem:cut_rank} and Lemma~\ref{lem:cross_term} still apply and we can use them again here. Lemma~\ref{lem:sn_counting} improves the counting argument for $\chi(\pi,\sigma)$ as well, which will result in the claimed exponential lower bound.

    Let $B_{t} : = \sum_{O \in N_t}R(O)$. The following analogue of Lemma~\ref{lem:mt_facts} holds:
    \begin{lemma}\label{lem:nt_facts}
        \begin{enumerate}
            \item $\pi \in N_t$ if and only if $\pi_{1*}$ and $\pi_{*1}$ are non-zero.
            \item $B_t^\dagger = B_t$.
            \item $\vert N_t \vert = (t-1)\vert S_{t-1}\vert = (t-1)^2 (t-2)!$
            \item $B_t \sigma = \sigma B_t = B_t$ for any $\sigma \in S_{t-1}$.
        \end{enumerate}
    \end{lemma}
    \begin{proof}
        $(1), (2), (4)$ are immediate from Lemma~\ref{lem:mt_facts}. $(3)$ is readily computed as $\vert S_t\vert - \vert S_{t-1}\vert$.
    \end{proof}

    Fixing an outcome $x_1$ in the first round, we will again define
    \begin{align}
        B_{t,x_1} := (V_{x_1}\otimes \id) B_t (V_{x_1} \otimes \id)^\dagger\ .
    \end{align}
    As in the proof of Lemma~\ref{lem:mt_bound}, it suffices to bound $\Vert B_{t,x_{1}}\Vert_1$, which we do via the same steps as Lemma~\ref{lem:2_norm_bound}.

    \begin{lemma}\label{lem:2norm_st}
        Let $t^2 \leq 2^{n-1}$, then for any $V:\Hc \rightarrow \Mc$,
        \begin{align}
            \Vert (V\otimes \id_{2:t}) B_t (V\otimes \id_{2:t})^\dagger \vert_2^2 \leq \vert N_t \vert 2^{12\log t+2} 2^{n(t-2)}\Tr[V^\dagger V]^2 \ .
        \end{align}
    \end{lemma}
    \begin{proof}
        We have that
        \begin{align}
            \Vert (V\otimes \id_{2:t}) B_t (V\otimes \id_{2:t})^\dagger \vert_2^2 & = \sum_{\pi,\sigma \in N_t} G(\pi,\sigma)\ ,
        \end{align}
        where $G(\pi,\sigma) = \Tr[(V^\dagger V\otimes \id) R(\pi) (V^\dagger V\otimes \id) R(\sigma)]$.
        We use Lemma~\ref{lem:cut_rank} and Lemma~\ref{lem:cross_term} to obtain the bound
        \begin{align}
            \vert G(\pi,\sigma)\vert \leq \Tr[V^\dagger V]^2 \min\left\{ 2^{n(t-2)}, 2^{n(t+3-\chi(\pi,\sigma))} \right\}\ .
        \end{align}
        For fixed $\pi$, we use Lemma~\ref{lem:sn_counting} to obtain
        \begin{align}
            \sum_{\sigma \in N_t} \vert G(\pi,\sigma)\vert \leq \Tr[V^\dagger V]^2 \left(2^{n(t-2)}\sum_{r=0}^5 t^{2r} + 2^{n(t+3)}\sum_{r=6}^t 2^{r(2\log t - n)}\right)\ .
        \end{align}
        By the assumption that $t \leq 2^{n-1}$, the second term can be upper bounded by $2^{n(t+3)-6(n-2\log t)+1}$. The first can be upper bounded with $2^{n(t-2)+10t+1}$. 
    \end{proof}
    \begin{corollary}\label{cor:1norm_st}
        Let $t^2 \leq 2^{n-1}$. Then, for any $V:\Hc \rightarrow \Mc$, it holds that
        \begin{align}
            \Vert (V\otimes \id_{2:t}) B_t (V\otimes \id_{2:t})\Vert_1 & \leq 2^{nt-3n/2+k/2+13/2\log t+3/2}\Tr[V^\dagger V]\ .
        \end{align}
    \end{corollary}
    \begin{proof}
        Let $P_t : = \id_1 \otimes \frac{1}{t!}\sum_{\sigma \in S_{t-1}}R(\sigma)$. Then, Lemma~\ref{lem:nt_facts} implies that $P_t B_t P_t = B_t$ and hence $\rank B_{t,x_1} \leq 2^k \sum_{\sigma \in S_{t-1}}\Tr[R(\sigma)]/(t-1)!$. The trace here is the dimension of the symmetric subspace $\overset{t-1}{\vee}\mathbb{C}^{2^n}$, which is $\binom{2^n+t-2}{t-1}$~\cite{harrow2013church}. We then use Lemma~\ref{lem:2norm_st} combined with $\Vert A\Vert_1 \leq \sqrt{\rank(A)} \Vert A\Vert_2$ for arbitrary operators $A$. Note that 
        \begin{align}
            \vert N_t \vert \binom{2^n+t-2}{t-1} & = \frac{(2^n+t-2)! (t-1)}{(2^n-1)!}\\
            & = (t-1) 2^{n(t-1)}\prod_{j=0}^{t-2}1+\frac{j}{2^n}\\
            & \leq 2^{n(t-1)+\log t+t^2/2^n} \\
            & \leq 2^{n(t-1)+\log t + 1}\ ,
        \end{align}
        where the final inequality follows from the assumption that $t^2 \leq 2^{n-1}$ (really $t^2 \leq 2^n$ would suffice) and the general inequality $1+x \leq e^x$. Lemma~\ref{lem:2norm_st} then completes the proof.
    \end{proof}

    Now we can prove an analogue of Lemma~\ref{lem:mt_bound} for $N_t$.

    \begin{lemma}\label{lem:nt_bound}
        Let $t^2 < c2^n$ for an absolute constant $c$. Then, for any protocol using single-copy measurements and $k$ qubits of memory, it holds that
        \begin{align}
            \frac{1}{2^{nt}}\sum_{\vec{x}} \left\vert \sum_{\sigma \in N_t} \Tr[R(\sigma) E_{\vec{x}}]\right\vert \leq 2^{3/2-(n-k-13\log t)/2}\ .
        \end{align}
    \end{lemma}
    \begin{proof}
        Fix an outcome $x_1$ in the first round. We have that
        \begin{align}
            \frac{1}{2^{nt}} \sum_{\vec{x}} \left\vert \sum_{\sigma \in N_t} \Tr[R(\sigma) E_{\vec{x}}]\right\vert & = \frac{1}{2^{nt}}\sum_{x_1} \sum_{\vec{x}_{>1}} \left\vert \sum_{\sigma \in N_t} \Tr[R(\sigma) E_{\vec{x}}]\right\vert\\
            & = \frac{1}{2^{nt}}\sum_{x_1} \sum_{\vec{x}_{>1}} \left\vert \Tr[B_{t,x_1} \left( V_{x_1}^{\vec{x}_{>1}} \right)^\dagger\left( V_{x_1}^{\vec{x}_{>1}} \right)] \right\vert\\
            & \leq \frac{1}{2^{nt}}\sum_{x_1} \Vert B_{t,x_1}\Vert_1\\
            & \leq 2^{-3n/2+k/2+13/2\log t+3/2} \sum_{x_1}\Tr[V_{x_1}^\dagger V_{x_1}]\\
            & = 2^{3/2-(n-k-13\log t)/2}\ ,
        \end{align}
        where the first inequality is Fact~\ref{fact:trace_norm_inequality}, the second is Corollary~\ref{cor:1norm_st}, and the final equality is completeness of the $V_{x_1}$'s.
    \end{proof}
    
    Piecing this together with the recursive technique in the proof of lower bound of  Theorem~\ref{thm:stabtestingupper} proves the claimed lower bound.

\end{document}